\documentclass[a4paper]{jpconf}
\usepackage{amsfonts,amsthm,amsmath,amssymb,
amsbsy,euscript,wasysym} 

\newtheorem{theor}{Theorem}
\theoremstyle{definition}

\newtheorem{proposition}[theor]{Proposition}
\newtheorem{lemma}[theor]{Lemma}
\newtheorem{cor}[theor]{Corollary}
\newtheorem{define}{Definition}

\newtheorem{example}{Example}[section]
\theoremstyle{remark}
\newtheorem{rem}{Remark}[section]

\newcommand{\BBR}{\mathbb{R}}\newcommand{\BBC}{\mathbb{C}}
\newcommand{\BBN}{\mathbb{N}}

\newcommand{\BBZ}{\mathbb{Z}}

\newcommand{\cE}{\mathcal{E}}

\newcommand{\cEEL}{{\cE}_{\text{\textup{EL}}}}

\newcommand{\cF}{\mathcal{F}}

\newcommand{\cL}{\mathcal{L}}
\newcommand{\cO}{\mathcal{O}}

\newcommand{\bc}{{\mathbf{c}}}

\newcommand{\boldi}{{\boldsymbol{i}}}

\newcommand{\bp}{{\boldsymbol{p}}}
\newcommand{\bq}{{\boldsymbol{q}}}
\newcommand{\bolds}{{\boldsymbol{s}}}
\newcommand{\bu}{{\boldsymbol{u}}}

\newcommand{\bx}{{\boldsymbol{x}}}
\newcommand{\bby}{{\boldsymbol{y}}}
\newcommand{\bz}{{\boldsymbol{z}}}

\newcommand{\bF}{{\boldsymbol{F}}}

\newcommand{\bQ}{{\boldsymbol{Q}}}

\newcommand{\bzeta}{{\boldsymbol{\zeta}}}

\newcommand{\gN}{\mathfrak{N}}
\newcommand{\gM}{\mathfrak{M}}
\newcommand{\veps}{\varepsilon}

\newcommand{\dd}{\partial}
\newcommand{\Id}{{\mathrm d}}

\newcommand{\diftat}[2]{ \left. \frac{\Id}{\Id #1} \right|_{#1=#2} }    
\newcommand{\BV}{{\text{\textup{BV}}}}

\DeclareMathOperator{\rank}{rank}

\DeclareMathOperator{\dvol}{d
vol}

\DeclareMathOperator{\const}{const}

\DeclareMathOperator{\GH}{gh}
\DeclareMathOperator*{\bigotimesk}{{\bigotimes\nolimits_{\Bbbk}}}

\newcommand{\ov}{\overline}

\newcommand{\EL}{{\text{EL}}}

\newcommand{\schouten}[1]{\lshad {#1} \rshad}

\newcommand{\by}[1]{\textrm{{#1}}}
\newcommand{\jour}[1]{\textit{{#1}}}
\newcommand{\vol}[1]{\textbf{{#1}}}
\newcommand{\book}[1]{\textit{{#1}}}

\begin{document} 
\title
{The geometry of variations\\ in Batalin\/--\/Vilkovisky formalism}


\author
{Arthemy V. Kiselev}
\address{
Johann Ber\-nou\-lli Institute for Mathematics and Computer Science, University of Groningen,
P.O.~Box 407, 9700~AK Groningen, The Netherlands} 
\ead{
A.V.Kiselev@rug.nl}



\begin{abstract}
We explain why no sources of divergence are built into the Batalin\/--\/Vilkovisky (BV)~Laplacian, 
whence there is no need to postulate any \textsl{ad hoc} conventions 
such as ``$\boldsymbol{\delta}(0)=0$'' and ``$\log\boldsymbol{\delta}(0)=0$''
within 
BV-\/approach to quantisation of gauge systems. 
Remarkably, the geometry of iterated variations
does not refer at all to the construction of Dirac's 
$\boldsymbol{\delta}$-\/function 
as a limit of smooth kernels. We illustrate the reasoning by re\/-\/deriving --but not just `formally postulating'-- the standard properties of 
BV-\/Laplacian and Schouten bracket and by verifying their basic inter\/-\/relations (e.g., 
cohomology preservation by gauge symmetries of the quantum master\/-\/equation).
\end{abstract}

\pagestyle{plain}

\section*{Introduction}
This is a paper about geometry of variations. 
We formulate 
definitions of the objects and structures which are cornerstones of 
Batalin\/--\/Vilkovisky formalism~\cite{BarnichBrandtHenneaux1995,BV,GitmanTyutin,HenneauxTeitelboim,ZinnJustin1975}. 
To confirm 
the intrinsic self\/-\/regularisation of 
BV-\/Laplacian, 
we explain why there are no divergencies in it (such excessive elements are traditionally encoded by using derivatives 
of Dirac's $\boldsymbol{\delta}$-\/distribution). 
Namely, we specify the geometry in which the following canonical inter\/-\/relations between the variational Schouten bracket~$\lshad\,,\,\rshad$ and BV-\/Laplacian~$\Delta$ are rigorously proven for any BV-\/functionals~$F,G,H$:
\begin{subequations}\label{EqAllIntro}
\begin{align}
\lshad F,G\cdot H\rshad &= \lshad F,G\rshad\cdot H 
+ (-)^{(|F|-1)\cdot|G|} G\cdot\lshad F,H\rshad,
\\
\Delta(F\cdot G) &= \Delta F\cdot G + (-)^{|F|}\lshad F,G\rshad
+(-)^{|F|}F\cdot\Delta G,\label{EqDeviationDerivationIntro}
\\
\Delta\bigl(\lshad F,G\rshad\bigr) &= \lshad\Delta F,G\rshad 
+(-)^{|F|-1}\lshad F,\Delta G\rshad,
\label{EqZimes}
\\
\Delta^2 &= 0\qquad \Longleftrightarrow\qquad \text{Jacobi}\bigl(\lshad\,,\,\rshad\bigr)=0.\label{EqDeltaSquareIntro}
\end{align}
\end{subequations}
There is an immense literature on this subject's intrinsic difficulties and attempts of regularisation of apparent divergencies in it (e.g., see~\cite{Costello2007,CostelloBook,HidingAnomalies,TroostNieuwProeyen,VTSh} vs~\cite{GomisParisSamuel}). While the BV-\/quantisation technique has advanced far from its sources~\cite{BV,BRST}, it is still admitted that it lacks sound mathematical consistency (\cite[\S15]{HenneauxTeitelboim} or~\cite[\S3]{BarnichHabilitation}). The calculus in this field is thus reduced to formal operation with expressions which are expected 
to render 
the theory's main objects and structures. Several \textsl{ad hoc} techniques for cancellation of divergencies, allowing one to strike through calculations and obtain meaningful results, are adopted by repetition; we 
briefly review the plurality of such tricks in what follows.

Our reasoning is independent from such conventional schemes for cancellation of infinities or from other practised roundabouts for regularisation of terms which are believed to be infinite (e.g., by erasing 
`infinite constants' \cite{CattaneoFelderCMP2000}). In particular, we do not pronounce the traditional password
\begin{equation}\label{EqSur}
\boldsymbol{\delta}(0)\mathrel{{:}{=}}0
\end{equation}
which lets one enter the existing paradigm and use its quantum alchemistry for operation with what remains from 
Dirac's $\boldsymbol{\delta}$-\/distribution.\footnote{Another convention is
$\log\boldsymbol{\delta}(0)=0$; 
we show that natural counterparts of the true geometry of variations lead to 
this intuitive convention and simultaneously to~\eqref{EqSur} --- none of the two being actually 
required.}
Our message is this: we do not propose to replace `bad slogans' with `good slogans,' which would mean that a choice of
conventions is still left to the one who attempts regularisation in the BV-\/setup. Such deficiency 
would symptomise that 
the theory remains a formal procedure. We now focus on the true sources of known difficulties. 
By analysing the geometry of variations of functionals at a very basic level, we prove the absence of apparently divergent
essences. The intrinsically regularised 
definitions of the BV-\/Laplacian~$\Delta$ and Schouten bracket~$\lshad\,,\,\rshad$ are the main result of this paper.

The new understanding leaves intact 
but substantiates the bulk of results which have been obtained 
by using various \textsl{ad hoc} 
techniques (that is, explicitly or tacitly referring to the surreal 
equalities $\boldsymbol{\delta}(0)=0$
and $\log\boldsymbol{\delta}(0)=0$); we refer to a detailed review~\cite{BarnichHabilitation} for an account of early developments in BV-\/formalism.
We do not aim at a reformulation or reproduction of any old or recent achievements, accomplishing here a different task.

In fact, we invent nothing new. 
It is the coupling of dual vector spaces which ensures the intrinsic self\/-\/regularisation of BV-\/Laplacian and validity of equalities~\eqref{EqAllIntro}, with~\eqref{EqZimes} in particular. Therefore, it would be redundant to start developing any brand\/-\/new formalism 
(cf.~\cite{VTSh}); on the other hand, we \textsl{prove} properties~\eqref{EqAllIntro} and not just \textsl{postulate} these assertions 
(cf.~\cite{GomisParisSamuel}).

We employ standard notions, constructions, and techniques from the geometry of jet spaces~\cite{TwelveLectures,ClassSym,Olver}. Because the geometry of BV-\/objects is essentially variational, it would be methodologically incomplete to handle them as if the space\/-\/time, that is, the base manifold in the bundles of physical fields, 
were just a point~(\cite{VoronovKhudaverdian2013,SchwarzBV} or~\cite{YKS2008SIGMA}). The language of jet spaces is extensively used in the study of BV-\/models, see~\cite{BarnichHabilitation,BarnichBrandtHenneauxPhysRep,GomisParisSamuel,McCloud}: the bundles of jets of sections usually appear in such traditional contexts as calculation of symmetries or conservation laws. In this paper we apply these geometric techniques at a much more profound level and give rigorous definitions for BV-\/objects. Let us emphasize that 
we do not aim at extending one's ability to write more formulas according to a regularly emended system of accepted algorithms; we explicate the genuine nature of objects and their canonical matchings, not taking any formulas for quasi\/-\/definitions.

This paper is structured 
as follows. Containing a brief 
overview of traditional approaches to regularisation of the BV-\/formalism, 
this introduction concludes with a pa\-ra\-ble; the line of our reasoning is reminiscent to that of \textit{Lettres persanes} by Montesquieu. 

In section~\ref{SecEL} we describe the true geometry of variations; 
we first reveal the correspondence between action
functionals and infinitesimal shifts of classical trajectories or physical fields. An understanding of 
nontrivial 
mechanism of such matching achieved for \textit{one} variation, the picture of \textit{many} variations becomes clear.
This approach resolves the obstructions for regularisation of iterated variations in BV-formalism; we remark that Dirac's 
$\boldsymbol{\delta}$-\/function does not appear in section~\ref{SecDef}
at all.
\footnote{We refer to~\cite{GelfandShilov} for the theory of distributions. Let us specify that singular linear integral operators which emerge in the course of our reasoning will \textbf{not} be approached via parametric families of regular linear integral functionals with piecewise continuous or smooth kernels (in which context the notation ``$\boldsymbol{\delta}(0)$'' for Dirac's \textsl{function} is used in the literature).}

In section~\ref{SecBVzoo} we recall in proper detail the standard construction of Batalin\/--\/Vilkovisky (BV) vector bundles
with canonically conjugate pairs of ghost parity\/-\/even and odd variables. 
In this specific setup we analyse the construction of two distinct couplings of
the BV-\/fibres' ghost parity\/-\/homogeneous vector subspaces with their respective duals. In particular, in section~\ref{SecSigns} we focus on the rule of signs which determines the anti\/-\/commutation of differential one\/-\/forms in the geometry at hand.
Applying the geometric concept of iterated variations in section~\ref{SecVariations}, we represent the left-{} and right variations of functionals in terms of left-{} or right\/-\/directed singular linear integral operators;
this framework ensures the intrinsic regularisation of iterated variations.
We then formulate in section~\ref{SecDefinitions} the 
definitions of BV-\/Laplacian~$\Delta$ and variational
Schouten bracket~$\lshad\,,\,\rshad$ (or \textsl{antibracket}). 
We show 
that these definitions 
are operational, 
amounting to natural, well\/-\/defined reconfigurations of the geometry 
(but not to any hand\/-\/made algorithms for cancellation of divergent terms; for those do not appear at all).
Our main result, which is contained in section~\ref{SecProof}, is an explicit proof --\,that is, starting from 
basic principles\,-- of relations~\eqref{EqAllIntro}. 
In other words, we neither postulate a validity of these properties nor elaborate 
a cunning syllogism the aim of which would be to convince why such assertions should hold provided that one knows when
various (derivatives of) Dirac's $\boldsymbol{\delta}$-functions must be erased in the course of so arguable a reasoning.

For consistency, we first apply the above theory to a standard derivation of the quantum master\/-\/equation from the Schwinger\/--\/Dyson condition that essentially eliminates a dependence on the unphysical, ghost parity\/-\/odd dimensions (see section~\ref{SecMaster}); we also recall here the construction of quantum BV-\/differential. 
The point is that neither divergencies nor \textsl{ad hoc} cancellations occur in the entire argument. 
On the same grounds we address in section~\ref{SecGauge} the quantum BV-\/cohomology preservation by infinitesimal gauge symmetries of the quantum master\/-\/equation. (We refer to~\cite{BV,BRST,GitmanTyutin,HenneauxTeitelboim} and also~\cite{AKZS,KontsevichSoibelman,VTSh} in this context; several methodological comments, which 
highlight our concept, are placed in section~\ref{SecFeynman} along the lines of a well\/-\/known reasoning.)

The paper concludes with a statement that an intrinsic regularisation in the geometry of iterated variations relies on the principle of locality (which manifests also through causality). We argue that a logical complexity of geometric objects grows while they accumulate the (iterated) variations\,; a conversion of such composite\/-\/structure objects into maps which take physical field configurations to numbers entails a decrease of the complexity via a loss of information. Having motivated this claim in section~\ref{SecDef}, we prove that the logic of analytic reasonings may not be interrupted\,; for example, the right\/-\/hand side of~\eqref{EqZimes} is not assembled from the would\/-\/be constituent blocks~$\Delta F$ and~$\Delta G$ for which it is known in advance how they take field configurations to numbers whenever the functionals~$F$ and~$G$ are given.

The paper explicitly answers the question what variations are~--- in particular, what iterated variations are. Moreover, we tacitly describe a geometric mechanism which is responsible for the anti\/-\/commutation of differential one\/-\/forms\,; such mechanism ensures that the results of calculations match  empiric data even if the exterior algebras of forms are introduced by hand. The roots of this principle 
are none other that the ordering of dual vector spaces which stem in the course of variations in models of nonlinear phenomena (this picture is addressed in section~\ref{SecSigns}).

We illustrate our 
approach with elementary starting section~\ref{SecEL} 
in which we inspect the 
matching of geometries --one for an action functional, the other for a field's test 
shift-- in the course of derivation of Euler\/--\/Lagrange equation of motion in field 
theory. The second example on pp.~\pageref{Countercounterexample}--\pageref{pEndCountercounterexample} 
clarifies the idea specifically 
in the BV-\/setup of (anti)\/fields and (anti)\/ghosts. 
We thus provide a pattern for all types of calculations which involve the Schouten 
bracket and BV-\/Laplacian in any model.

\subsection*{Historical context: an overview}
There is a class of significant papers in which the BV-\/formalism is 
developed 
under assumption that the space\/-\/time is a point. Indeed, such hypothesis is equivalent
 to an agreement that the only admissible sections of bundles over space\/-\/time are 
constant; this implies that even if their derivatives are nominally present in some 
formulas, they are always equal to zero. The calculus of variations then reduces to usual differential geometry on the 
bundles' fibres. It must be noted that 
publications containing the above assumption did contribute to the subject 
and in many cases guided its further development (we recall the respective comment in~\cite{VTSh} and refer 
to~\cite{Costello2007,HenneauxTeitelboim,VoronovKhudaverdian2013,YKS2008SIGMA,
SchwarzBV,WittenAntibracket}). Moreover, the no\/-\/derivatives reduction sometimes allows one to jump at conclusions which are correct; an integration by parts over the base manifold~$M^n$ is restored --whenever possible-- at the end of the day. Still this oversimplification is potentially dangerous because variational calculus of integral functionals conceptually exceeds any classical differential geometry on the fibres (see~\cite{Norway13} for discussion and~\cite{TwelveLectures,Galli10}). In the variational setup, the objects and their properties become geometrically different from their analogues on usual manifolds even if the terminology is kept unchanged. 
Here we recall for example that variational multivectors do not split to wedge products of variational one\/-\/vectors and likewise, several Leibniz rules are irreparably lost but this can not be noticed when all derivatives equal zero. In fact, it is the abyss between classical geometry of manifolds and geometry of variations for jet spaces of maps of manifolds which motivated our earlier study~\cite{Galli10}. Yet the misconception is still present in active research, e.g., 
see~\cite{BarnichBialowieza,VoronovKhudaverdian2013,YKS2008SIGMA,MerkulovWillwacher}.

The fact of incompleteness of such heuristic analogies from usual geometry of manifolds is signalled in~\cite{VTSh}. Paradoxically, it is simultaneously not true that a solution of the regularisation problem for BV-\/Laplacian has no analogues in the case of ODE dynamics on manifolds. From section~\ref{SecEL} below it is readily seen that good old techniques persist in the finite\/-\/dimensional ODE geometry at the level of standard linear algebra of dual vector spaces.\footnote{On the other hand, the variational setup highlights the fundamental concept of a physical field as a system with degrees of freedom attached at every point of the space\/-\/time~$M^n$; 
we 
focus on this aspect in what follows.}

The article~\cite{VTSh} is a considerable step towards a solution of the regularisation problem in BV-\/formalism.
A weighted, critical overview of various inconsistencies, \textsl{ad hoc} practices, and roundabouts is summed up there. 
The object of~\cite{VTSh} was to formulate a 
self\/-\/contained analytic concept which would make the variational calculus of functionals free from divergencies
and infinities. Still it remained unclear from~\cite{VTSh} what the generality of underlying geometry is and why such
self\/-\/consistent formalism should actually exist at the level of objects, i.e., beyond a mere ability to write formulas.
In particular, it remained unnoticed that the main motivating example --namely, the canonical BV-\/setup-- itself is the
only class of geometries in which the technique is grounded.
\footnote{The integration of closed algebra of gauge symmetries for the quantum 
master\/-\/equation to a group of transformations of the master\/-\/action~$S^\hbar$ 
remains a separate problem, which is also addressed in~\cite{VTSh}. Suppose that the 
standard cohomological obstructions to such integration vanish (see section~\ref{SecGauge} 
below), whence (i) all 
infinitesimal transformations of the functional~$S^\hbar$ 
are exact, i.e., they are generated by odd ghost\/-\/parity elements~$F$, and also (ii) such transformations can be
extended from the master\/-\/action~$S^\hbar$ to evolution of the observables~$\cO$. We remark that, unlike it is claimed
in~\cite{VTSh}, neither of the two groups of functionals' transformations is induced by any well\/-\/defined change
of BV-\/coordinates; of course, evolutionary vector fields are well\/-\/defined objects in that geometry and one could
study them regardless of these functionals' transformations. We shall recall in section~\ref{SecGauge} the standard
construction of automorphisms for quantum BV-\/cohomology groups; it illustrates our concept because the notion of
quantum gauge symmetries explicitly refers to all basic 
properties of the BV-\/Laplacian and Schouten bracket, see~\eqref{EqAllIntro} on p.~\pageref{EqAllIntro}.}
A correctness but incompleteness of the approach in~\cite{VTSh} means the following in 
practice. Whenever a theorist refers to the formalism of \textit{loc.\ cit.}, Nature immediately creates a new, principally
inobservable essence --a metric field which is denoted by $E(x_1,\ldots,x_n;\Gamma)$ in~\cite{VTSh}-- on top of the
electromagnetic and weak gauge connections, as well as the fields for strong force, gravity, or any other gauge
fields~$\Gamma$. It is perhaps this methodological difficulty which hints us why the approach 
of~\cite{VTSh} is considered ``formal'' by many experts; that conceptual paper remains scarcely known to a wider
community.%
\footnote{An attempt to interpret the formalism of~\cite{VTSh} in terms of the language of PDE geometry (particularly, 
in the context of~\cite{KuperCotangent}, see also~\cite{TwelveLectures,ClassSym,Olver}) was performed
in~\cite{Memorandum1641} and published in abridged form in~\cite{Kiev2003}. The construction of Schouten bracket
in~\cite{Memorandum1641} relies on the notion of variational cotangent bundle~\cite{KuperCotangent} and on classical
approach to the theory of variations. On one hand, this ensures the validity of Jacobi identity for the bracket
(see the \textsl{second} half of Eq.~\eqref{EqDeltaSquareIntro} but \textsl{not} the first one). But on the other hand,
we have showed by a counterexample in~\cite[\S3]{Laplace13} that the old approach fails to relate by~\eqref{EqZimes}
the Schouten bracket to BV-\/Laplacian. In other words, the BV-\/Laplacian did not entirely generate the variational
Schouten bracket, making only Eq.~\eqref{EqDeviationDerivationIntro} but not~\eqref{EqZimes} possible in that geometry
(cf.~\cite{YKS2008SIGMA}).}

To demystify the notion of a ``metric field~$E(x_1,\ldots,x_n;\Gamma)$,'' 
we describe in this paper an elementary geometric mechanism for the long\/-\/expected but still intuitively paradoxical analytic behaviour of variations.
This mechanism implies that Nature is not obliged to respond to the needs of a theorist and create such multi\/-\/entry distributions upon request.

Another line of reasoning, which led to much progress in a revision of BV-\/structures and 
regularisation of divergences, was pursued 
in~\cite{Costello2007,CostelloBook}.
We recall that the language of \textit{loc.\ cit.}\ is 
functional analytic so that the theory's objects are viewed as (Dirac's) distributions (and heat kernels are implemented).
According to~\cite[\S1.8]{Costello2007}, the BV-\/Laplacian~$\Delta$ which is used in physical theories is ill\/-\/defined
because for a given action~$S$ over space\/-\/time~$M^n$ of positive dimension~$n$ the object~$\Delta S$ involves
a multiplication of singular distributions (and thus --a quotation from~\cite{Costello2007} continues-- $\Delta S$~has
the same kind of singularities as appear in one\/-\/loop Feynman diagrams).
The regularisation technique proposed in~\cite{Costello2007,CostelloBook} 
stems from analysis of the distributions' limit behaviour as one approaches the ``physical'' structures by using 
regular ones. 

The resolution to apparent difficulties is that there are several distinct geometric
constructions which yield the same singular linear operators with support on the diagonal
(in what follows we study in detail on which space such operators are defined).
\\
\centerline{\rule{1in}{0.7pt}}

\enlargethispage{1.2\baselineskip}
\noindent%
We now discuss a peculiar, well\/-\/established domain, the very form of existence of which could be hardly believed in. 
In that theory, there is a serious lack of rigorous definitions for the most elementary objects; at the same time, 
there is a rapidly growing number of monumental reviews. 
Whereas the theory's difficulties are clearly inherited from a deficit of boring rigour at the initial stage, 
such hardships are proclaimed the theory's immanent components. 
At expert level it is mandatory to have a firm knowledge of the built\/-\/in difficulties and readily classify 
the descriptive objects which those apparent obstructions bring into the mathematical apparatus. 
(There is no firm guarantee that the (un)\/necessary objects really exist beyond written formulas.) 
The way of handling inconveniences largely amounts not to resolving them by a thorough study of their origins 
but to some \textsl{ad hoc} methods for hiding their presence. 
Doing research is thus substituted by practising a ritual.

However, the community of experts who mature in operation with formulas (a part of 
which are believed to express something objectively existing) maintains a considerable pluralism about a proper way to mask the symptoms of troubles:
\begin{itemize}
\item The radicals declare that 
undefined objects which seem to make trouble must be set equal to zero.
\item The revisionist approach prescribes a postfactum erasing of not the entire objects 
(which are still undefined) but of 
undesirable elements in those objects' description. 
\item A diplomatic viewpoint is that there might be sources of trouble
but their contribution to final results is suppressed as soon as the objects' desired properties are postulated (regardless of the actual presence or absence of such sources and one's 
ability to substantiate those properties).
\end{itemize}
For an external observer, this state\/-\/of\/-\/the\/-\/art could seem atypical for a consistent theory. Indeed, the reliability of its main pillar is a matter of irrational belief.

\section{The geometry of variations}\label{SecEL}
\noindent
Let us first analyse the basic geometry of variations of functionals; by comprehending 
the full setup of a one-time variation, we shall understand the geometry 
of many. Specifically, in this section we reveal the interrelation of bundles in the course of 
integration by parts; we also explain a rigorous construction of iterated variations.

The core of traditional difficulties in this domain is that a use of only fibre bundles $\pi$ of physical fields, which are
subjected to test shifts, is insufficient. We argue that the tangent bundles $T\pi$ to the bundles $\pi$ may not be
discarded (see Fig.~\ref{FigEL}). 
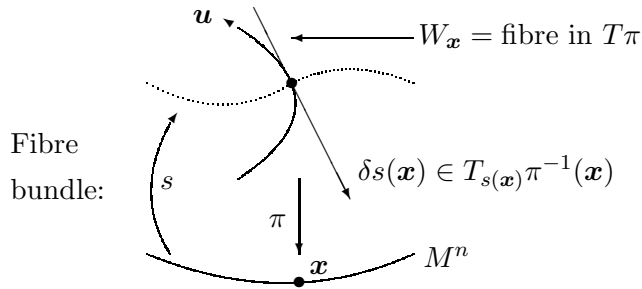
\begin{figure}[htb]
\begin{center}
\unitlength=1mm
\linethickness{0.4pt}
\begin{picture}(101.33,42.67)
\bezier{152}(25.00,10.00)(43.33,2.33)(60.00,10.00)
\put(45.00,6.33){\circle*{1.33}}
\put(45.00,20.00){\vector(0,-1){10.00}}
\bezier{144}(37.00,20.00)(52.00,30.00)(37.00,40.00)
\put(44.00,32.67){\circle*{1.33}}
\put(44.00,32.67){\vector(1,-2){7.67}}
\put(44.00,32.67){\line(-1,2){5.00}}
\put(60.00,38.67){\vector(-1,0){16.00}}
\put(37.00,40.00){\vector(-3,2){2}}       
\bezier{30}(25.00,32.67)(35.33,27.00)(44.00,32.67)
\bezier{30}(44.00,32.67)(51.33,36.33)(60.00,32.67)
\bezier{80}(28.00,10.00)(23.00,18.00)(28.00,27.33)
\put(27.67,27.00){\vector(3,4){1.33}}
\put(31.00,40.67){\makebox(0,0)[lb]{$\bu$}}
\put(60.67,37.33){\makebox(0,0)[lb]{$W_{\bx}=\text{fibre in }T\pi$}}
\put(52.67,18.33){\makebox(0,0)[lb]{$\delta s(\bx)\in T_{s(\bx)}\pi^{-1}(\bx)$}}
\put(46.33,7.33){\makebox(0,0)[lb]{$\bx$}}
\put(61.33,8.33){\makebox(0,0)[lb]{$M^n$}}
\put(41.00,13.67){\makebox(0,0)[lb]{$\pi$}}
\put(26.67,18.67){\makebox(0,0)[lb]{$s$}}
\put(7.00,23.33){\makebox(0,0)[lb]{Fibre}}
\put(7.00,17.33){\makebox(0,0)[lb]{bundle:}}
\end{picture}
\caption{The fibre bundle~$\pi$ of fields~$s$ and vector bundle~$T\pi$ of their variations~$\delta s$.}\label{FigEL}
\end{center}
\end{figure}
For identities~\eqref{EqAllIntro} to hold one must substantiate why higher-order
variational derivatives are (graded-)\/permutable whenever one inspects the response of a given functional to shifts of its
argument along several directions. To resolve the difficulties, we properly enlarge the space of functionals and adjust
a description of the geometry for the functionals' variations: in fact, each variation brings its own copy of the base $M^n$
into the picture (see Fig.~\ref{FigLelystad} on p.~\pageref{FigLelystad}). 

\subsection{Notation}
We now fix some notations, 
in most cases matching that from~\cite{
TwelveLectures} 
(for a more detailed
exposition of these matters, see for example~\cite{TwelveLectures,
ClassSym,Olver}). 

Let $\pi\colon E\to M$ be a smooth fibre bundle%
\footnote{Vector bundles are primary examples but 
we do not actually use the linear vector space structure of their fibres so that $\pi$ could be any smooth fibre bundle.}
with $m$-dimensional fibres $\pi^{-1}(\bx)$ over points $\bx$ of a smooth real 
oriented manifold~$M$ of dimension~$n$; 
we assume that all mappings, including those which determine the smoothness class of manifolds, 
are infinitely smooth.

We let $x^i$ denote local coordinates in a chart $U_{\alpha}\subseteq M^n$ and $u_j$ be the fibre coordinates.
We denote by $[\bu]$ a differential dependence of the fibre variables (specifically in the BV-setup, a differential
dependence $[\bq]$ on physical fields and other ghost parity-even variables, and we denote by $[\bq^{\dagger}]$
that of ghost parity-odd BV-variables).

\begin{rem}\label{RemStartNonGraded}
We suppose that the initially given bundle $\pi$ of physical fields is not graded. In what follows, starting with~$\pi$, we
shall construct new bundles whose fibres are endowed with the $\BBZ_2$-valued ghost parity $\GH(\,\cdot\,)$. However, our
reasoning remains valid for 
superbundles $\pi^{(0|1)}$ over supermanifolds $M^{(n_0|n_1)}$ 
(\cite{BerezinAA,Voronov2002}) 
and to a non\-com\-mu\-ta\-ti\-ve setup of cyclic\/-\/invariant words 
(see~\cite{SQS11,Lorentz12} and references therein), cf.\ Fig.~\ref{FigPants} below.
\end{rem}


We take the infinite jet space $\pi_\infty\colon J^{\infty}(\pi)\to M$ associated with this bundle~\cite{Ehresmann,
Olver}; 
a point from the jet space is then
$\theta=(x^i,\bu,u^k_{x^i},u^k_{x^i x^j},\dots,\bu_\sigma,\dots) \in J^{\infty}(\pi)$,
where $\sigma$~is a multi\/-\/index and we put~$\bu_\varnothing\equiv\bu$. 
If~$s \in \Gamma(\pi)$
is a section of~$\pi$, or a \textsl{field}, we denote by $j^{\infty}(s)$ its infinite jet, which is a section 
$j^{\infty}(s) \in \Gamma(\pi_\infty)$. Its value at $\bx \in M$ is 
$j^{\infty}_\bx(s)=(x^i, s^\alpha(x),\frac{\dd s^\alpha}{\dd x^i}(x),\dots,
\frac{\partial^{|\sigma|}s^\alpha}{\partial x^\sigma}(x),\dots)\in J^{\infty}(\pi)$.

We denote by $\cF(\pi)$ the properly understood algebra of finite differential order smooth functions on the infinite 
jet space $J^{\infty}(\pi)$, see \cite{TwelveLectures,ClassSym} for details. The space of top-degree horizontal forms on
$J^{\infty}(\pi)$ is denoted by $\ov\Lambda^n(\pi)$; let us also assume that at every $\bx\in M$ a volume element
$\dvol(\bx)$ is specified 
so that its pull\/-\/back under $\pi^*_{\infty}$ is an $n$-th degree form in $\ov\Lambda^n(\pi)$, cf.\ Remark~\ref{RemRolesInt} on p.~\pageref{RemRolesInt}.

The highest horizontal cohomology, i.\,e., the space of equivalence classes of
$n$-\/forms from $\ov{\Lambda}^n(\pi)$ modulo the image of the horizontal exterior differential $\ov{\Id}$ on $J^{\infty}(\pi)$,
is denoted by $\ov{H}^n(\pi)$; the equivalence class of $\omega\in\ov\Lambda^n(\pi)$ is denoted by
$\int\omega\in\ov{H}^n(\pi)$. 
We assume that sections $s\in\Gamma(\pi)$ are such that integration of functionals $\Gamma(\pi)\to\Bbbk$ by parts
is allowed and does not
result in any boundary terms (for example, the base manifold is closed, 
or the sections all have compact support, or
decay sufficiently fast towards infinity, or are periodic).

\subsection{Euler\/--\/Lagrange equations}
A derivation of Euler\/--\/Lagrange equations $\cEEL$ for a given action functional $S=\int\cL(\bx,[\bu])\dvol(\bx)$
is a model example which illustrates the correlation of two geometries:%
\footnote{An arrow over a variational derivative indicates the direction along which the shift $\delta s$ is transported
left- or rightmost. While the objects are non-graded commutative, this indication is not important.
It becomes mandatory in the $\BBZ_2$-graded commutative setup (see section~\ref{SecDef}):
likewise, the arrows are also mandatory and fix the direction of rotation for non-commutative cyclic 
words~\cite{SQS11,Lorentz12,KontsevichCyclic}; 
note that our formalism is extended verbatim to the variational calculus of such necklaces and their brackets.}
one for ``trajectories'' $s\in\Gamma(\pi)$ and the other for shifts~$\delta s$.
It is well known that the functional's response to a test shift $\delta s$ of its argument $s\in\Gamma(\pi)$ is described
by the formula~\cite[\S 12]{MMKM}
\begin{equation}\label{EqVariation}
\left.\frac{\Id}{\Id\veps}\right|_{\veps=0}S(s+\veps\cdot\overleftarrow{\delta}\!s)=
\int_M\dvol(\bx)\:\delta s(\bx)\cdot
\left.\frac{\overleftarrow{\delta}\!\cL(\bx,[\bu])}{\delta\bu}\right|_{j^{\infty}_{\bx}(s)}.
\end{equation}
We now claim that this one-step procedure is a correct consequence of definitions but itself not a definition of the
functional's variation. The above formula conceals 
a longer, nontrivial reasoning of which the right-hand side in~\eqref{EqVariation}
is an implication --- provided that the functional $S$ will not be varied by using any other test shifts, i.\,e., if the
correspondence $S\mapsto\cEEL$ yields the object $\cEEL$ of further study (cf.~\cite[\S 13]{MMKM}). Indeed, we notice
that the left-hand side of~\eqref{EqVariation} refers to \textit{three} bundles (namely,
the fibre bundle $\pi$ for a section $s\in\Gamma(\pi)$ whose infinite jet is $j^{\infty}(s)\in\Gamma(\pi_{\infty})$, 
the bundle $\pi_{\infty}$ for the integral functional $S\in\ov{H}^n(\pi)$, and the tangent vector bundle $T\pi$ such that
$\delta s\in\Gamma(T\pi)$ at the graph of $s$ in $\pi$,
see Fig.~\ref{FigEL}. (In what follows, a reference to attachment points $s(\bx)\in\pi^{-1}(\bx)$ will always be implicit
in the notation for $\delta s$: for a given section $s\in\Gamma(\pi)$, the base manifold $M^n$ is the domain of definition
for a test shift $\delta s(\bx,s(\bx))=\delta s(\bx)$ that takes values in $T_{s(\bx)}\pi^{-1}(\bx)$.) Let us figure
out how the domains of definition for the sections $s$ and $\delta s$ merge to one copy of the manifold $M^n$ over which
an integration is performed in the right\/-\/hand side of~\eqref{EqVariation}. Strictly speaking, 
from~\eqref{EqVariation} it is unclear whether the variational derivative,
$$
\frac{\overleftarrow{\delta}\!\cL(\bx,[\bu])}{\delta\bu}=
\sum_{|\sigma|\ge0}
\left(-\frac{\vec{\Id}}{\Id\bx}\right)^{\sigma}\frac{\vec{\dd}\!\cL(\bx,[\bu])}{\dd\bu_{\sigma}},
$$
stems from one (which would be false) or both (true!) copies of the base $M$.

To have a clear vision of the variations' geometry and by this avoid an appearance of phantoms in 
description, we now
vary the action functional $S$ at $s\in\Gamma(\pi)$ along $\delta s\in\Gamma(T\pi)$, commenting on each step we make.
In fact, it suffices to figure out where the objects and structures at hand belong to --- in particular, we should explain
the nature of binary operation $\cdot$ in the right-hand side of conventional formula~\eqref{EqVariation}. The key idea is to
understand what we are actually doing but not what we have got used to think we do in order to obtain an understandable
result~\cite[\S 13]{MMKM}. The discovery is that this ``multiplication of functions'' is a shorthand notation for the 
canonically defined coupling between vectors and covectors from (co)tangent spaces 
$W_{s(\bx)}$ and $W^{\dagger}_{s(\bx)}$,
respectively, at the points $s(\bx)$ of fibres $\pi^{-1}(\bx)$ in the bundle $\pi$.

To encode this linear-algebraic setup, let $i,j$ run from 1 to $m=\dim(\pi^{-1}(\bx))=\rank(T\pi)$ and take a local basis
$\vec{e}_i(\bby)$ 
in the tangent spaces $W_{s(\bby)}=T_{s(\bby)}(\pi^{-1}(\bby))$ at $s(\bby)$ over 
base points $\bby\in M$. Introduce the dual
basis $\vec{e}^{{}\,\dagger j}(\bx)$ in $W^{\dagger}_{s(\bx)}$ attached at $s(\bx)$ over $\bx\in M$. By construction, this 
means that the value
\begin{equation}\label{EqLocality}
\left\langle\vec{e}_i(\bby),\vec{e}^{{}\,\dagger j}(\bx)\right\rangle
\end{equation}
is equal to the Kronecker symbol $\delta_i^j$ if and only if $\bx=\bby$ and the 
vector $\vec{e}_i(\bby)\in W_{p_1}$ and
covector $\vec{e}^{{}\,\dagger j}(\bx)\in W^{\dagger}_{p_2}$ 
are attached at the same point $p_1=p_2$ of the fibre $\pi^{-1}(\bx)$ over $\bx=\bby\in M$.

The \textsl{locality} of this coupling is an absolute geometric postulate: the coupling is not defined whenever $\bx\neq\bby$
or the values $p_1=s_1(\bby)$ and $p_2=s_2(\bx)$ of two local sections $s_1,s_2\in\Gamma(\pi)$ are not equal at $\bx=\bby$.
Physically speaking, the coupling is then not defined because there is no channel of information which would communicate
the value $\delta s^i(\bby)\cdot\vec{e}_i(\bby)$ of excitation of the physical field $s\in\Gamma(\pi)$ at a point $\bby\in M$
to another point $\bx\neq\bby$ of the space-time $M$.


\begin{rem}
Let us remember that the definition of coupling between sections of (co)tangent bundles --- i.\,e., (co)tangent to either 
a given manifold or a given 
bundle $\pi$ which is the case here for Euler\/--\/Lagrange equations ---
forces the congruence $\{\bx=\bby,\ s_1(\bby)=s_2(\bx)\}$ of the (co)\/vectors' attachment points. We notice further that such
congruence mechanism does not refer to any limiting procedure for smooth distributed 
kernels and regular linear operators on the space of (co)vector fields. Indeed, vectors couple with their duals at a given
point regardless of any phantom limiting procedure which would grasp the (co)vector's values at any other points of the
manifold.%
\footnote{We recall that a similar, purely local geometric principle, not referring to the objects' values at non-coinciding
points, works in the definition of Hirota's bilinear derivative.}
\end{rem}

\begin{rem}
The coupling is a matching between test-shift vector fields which are tangent to the fibres of $\pi$ and, on the other hand,
with the elements of $\Gamma(T^*\pi)$ which are determined by the Lagrangian $\cL$. This binary operation yields the 
singular integral operator $\int_M\Id\bby\,\langle\delta s^i(\bby)\vec{e}_i(\bby)|$ with support on the diagonal. Independently,
the same operator can reappear as the limit in a parametric family of regular integral operators with smooth, distributed
kernels. This shows that the same object is constructed by using several algorithms. Yet the analytic behaviour of the limit
is determined not only by the limit itself but also by an algorithm how it is attained. 
Consequently, the object's analytic properties in the course of derivations could be (and actually, indeed they are)
drastically different for different scenarios. This is the key point in a regularisation of the formalism; to achieve this
goal, we properly identify the objects which are de facto handled.
\end{rem}

\begin{rem}\label{RemDconst0}
Referring to a concept of locality of events, this definition of coupling $\langle\,,\,\rangle$ ensures a very interesting
analytic behaviour of the value $\langle\vec{e}_i(\bby),\vec{e}^{{}\,\dagger i}(\bx)\rangle$ of pairing for dual objects
$\vec{e}_i(\bby)$ and $\vec{e}^{{}\,\dagger i}(\bx)$ at fixed $i$. Namely, this value is a constant scalar field which equals
unit $1\in\Bbbk$ 
at all points of the manifold $M$; the scalar field's partial derivatives with respect to $x^j$ or $y^k$, $1\le j,k\le n$,
vanish identically. We shall use this property in what follows (see Remark~\ref{RemNoEffect} on p.~\pageref{RemNoEffect}).
We also note that the logarithm of this coupling's unit value vanishes as well: 
$\log\langle\vec{e}_i(\bby),\vec{e}^{{}\,\dagger i}(\bx)\rangle=0$
whenever the coupling is well defined and $1\le i\le m$.
\end{rem}

Now let us return to the initial setup in context of Euler\/--\/Lagrange equation $\cEEL$ and one-step correspondence
$S\mapsto\cEEL$, see Fig.~\ref{FigEL}. We have that $S=\int\cL(\bx,[\bu])\dvol(\bx)$ is an integral functional; we let
$s\in\Gamma(\pi)$ be a background section (e.\,g., a sought\/-\/for solution of the Euler\/--\/Lagrange stationary 
point equation $\left.\delta S\right|_s=0$) and $\delta s\in\Gamma(T\pi)$ be a test shift of $s$. The linear term in
a response of $S\colon\Gamma(\pi)\to\Bbbk$ to a shift of its argument $s$ along $\delta s$ is
(cf.\eqref{EqLeftRightVariations} on p.~\pageref{EqLeftRightVariations})
\begin{multline}\label{EqStartSum}
\left.\frac{\Id}{\Id\veps}\right|_{\veps=0}S(s+\veps\overleftarrow{\delta}\!s)={}
\\{}=
\sum_{i,j}\sum_{|\sigma|\ge0}\int_M\Id\bby\int_M\dvol(\bx)\left\langle
(\delta s^i)\left(\frac{\smash{\overleftarrow{\dd}}}{\dd\bby}\right)^{\sigma} (\bby)\,\vec{e}_i(\bby),
\vec{e}^{{}\,\dagger j}(\bx)\left.\frac{\overrightarrow{\dd}\!\cL(\bx,[\bu])}{\dd u_{\sigma}^j}\right|_{j^{\infty}_{\bx}(s)}
\right\rangle.
\end{multline}

\begin{rem}\label{RemRolesInt}
The r\^oles of two integral signs in~\eqref{EqStartSum} are different. 
Namely, the 
volume form $\dvol(\bx)$ at
$\bx\in M^n$ comes from the integral functional $S\in\ov{H}^n(\pi)$; 
should a formal choice of the volume form be different,
the Euler\/--\/Lagrange equations would also change.%
\footnote{There are natural classes of geometries in which the Lagrangian $\cL(\bx,[\bu])$ in the action $S$ is a 
well-defined top-degree differential form, e.\,g., if the unknowns $\bu$ are differential one-forms (we recall the
Yang\/--\/Mills or Chern\/--\/Simons gauge theories in this context). Let us remember also that a construction of $\cL$ 
could refer to a choice of volume form $\dvol(\bx)$ on $M^n$. For instance, such is the case when the Hodge structure $*$
is involved (the Yang\/--\/Mills Lagrangian yields an example: $\cL\sim F_{\mu\nu}*F^{\mu\nu}$ in standard notation for
the stress tensor). To avoid excessive case-study, we use a uniform notation thus writing $\dvol(\bx)$ explicitly.

We recall further that the integration measure $\dvol\bigl(\bx,s(\bx)\bigr)
={\sqrt{|\det\bigl(g_{\mu\nu}(\bx,s)\bigr)|}}\Id\bx$ is field\/-\/dependent by virtue of Einstein's general relativity equations which --\.in their right\/-\/hand sides\,-- absorb the energy\/-\/momentum tensor of physical fields~$s\in\Gamma(\pi)$. The volume element will be denoted by~$\dvol(\bx)$ in order to emphasize that the space\/-\/time~$M^n$ is unique: Namely, field\/-\/dependent objects interact at its points only if the local geometry of underlying space\/-\/time is the same near~$\bx\in M^n$ for all objects (see Theorem~\ref{ThLaplaceOnProduct} and Remark~\ref{RemTwoToOneInSchouten} on p.~\pageref{RemTwoToOneInSchouten} for a realisation of this principle for the smooth manifold~$M^n$ endowed with metric tensor~$g_{\mu\nu}$).}
At the same time, the other integral sign $\int\Id\bby$ denotes the singular linear operator $\Gamma(T^*\pi)\to\Bbbk$
with support on the diagonal~\cite{GelfandShilov}; in fact, this notation means that a point $\bby$ runs through 
the entire integration domain $M$.
\end{rem}

\subsection{Integration by parts}\label{SecByParts}
The most interesting things start to happen when one integrates by parts over the domain $M^n$ of test shifts $\delta s$.
(By default, we let the supports of local perturbations $\delta s$ be such that no boundary terms appear in the course of
integration by parts over~$M$.)

For the sake of transparency 
let us first consider a model situation when there is just one derivative falling on $\delta s$ at $\bby$; all higher-order
cases are processed recursively. By the definition of a (partial) derivative $\dd/\dd y^i$, we have that%
\footnote{In the definition of derivative, the calculation of length $|\Delta\bby|$ in denominators refers to the standard
Euclidean metric in the linear vector spaces which determine 
coordinate neighbourhoods near points of the manifold $M$ at hand.}
\begin{multline*}
\int_M\Id\bby\int_M\dvol(\bx)\left\langle
(\delta s)\frac{\overleftarrow{\dd}}{\dd\bby} (\bby)\,\vec{e}(\bby),
\vec{e}^{{}\,\dagger}(\bx)
\left.\frac{\overrightarrow{\dd}\!\cL(\bx,[\bu])}{\dd\bu_{\sigma}}\right|_{j^{\infty}_{\bx}(s)}\right\rangle=\\=
-\int_M\Id\bby\int_M\dvol(\bx)\,\delta s(\bby)\frac{\overrightarrow{\dd}}{\dd\bby}
\left\{\langle\vec{e}(\bby),\vec{e}^{{}\,\dagger}(\bx)\rangle
\left.\frac{\overrightarrow{\dd}\!\cL(\bx,[\bu])}{\dd\bu_{\sigma}}\right|_{j^{\infty}_{\bx}(s)}\right\}.
\end{multline*}
By using a 
definition of the partial derivative which falls on the comultiple of $\delta s$, we obtain the difference%
\footnote{Here and in the equalities below we suppress the indexes $i$ running through $1,\dots,m$ at $\delta s^i(\bby)$ and
$\vec{e}_i(\bby)$ or $\vec{e}^{{}\,\dagger i}(\bx)$, or at $u^i_{\sigma}$ in the derivative which acts on $\cL$; we thus
avoid an agglomeration of formulas.}
\begin{multline*}
{}\stackrel{\text{def}}{=}-\int_M\Id\bby\int_M\dvol(\bx)\,\delta s(\bby)\cdot
\lim_{|\Delta\bby|\to0}\frac{1}{|\Delta\bby|}\Bigl\{
\left\langle\vec{e}(\bby+\Delta\bby),\vec{e}(\bx)\right\rangle
\left.\frac{\overrightarrow{\dd}\!\cL(\bx,[\bu])}{\dd\bu_{\sigma}}\right|_{j^{\infty}_{\bx}(s)}-\\
-\left\langle\vec{e}(\bby),\vec{e}(\bx)\right\rangle
\left.\frac{\overrightarrow{\dd}\!\cL(\bx,[\bu])}{\dd\bu_{\sigma}}\right|_{j^{\infty}_{\bx}(s)}\Bigr\}.
\end{multline*}
The locality postulate for coupling between (co)vectors $\vec{e}$ and $\vec{e}^{{}\,\dagger}$ forces the equality 
$\bby+\Delta\bby=\bx$ in the minuend, which yields the two different points at which the restriction of Lagrangian $\cL$
to the jet $j^{\infty}(s)$ of section $s\in\Gamma(\pi)$ is evaluated:
\begin{multline*}
=-\int_M\Id\bby\int_M\dvol(\bx)\,\delta s(\bby)\cdot\langle\vec{e}(\bby),\vec{e}^{{}\,\dagger}(\bx)\rangle\cdot\\\cdot
\lim_{|\Delta\bby|\to0}\frac1{|\Delta\bby|}
\Bigl\{\left.\frac{\overrightarrow{\dd}\!\cL(\bx,[\bu])}{\dd\bu_{\sigma}}\right|_{j^{\infty}_{\bx+\Delta\bby}(s)}-
\left.\frac{\overrightarrow{\dd}\!\cL(\bx,[\bu])}{\dd\bu_{\sigma}}\right|_{j^{\infty}_{\bx}(s)}\Bigr\}.
\end{multline*}
(Here we use the fact that the scalar product $\langle\,,\,\rangle$, whenever defined, is the Kronecker symbol.) 
We continue the equality,
\begin{equation*}
\stackrel{\text{def}}{=}-\int_M\Id\bby\int_M\dvol(\bx)
\left\langle\delta s(\bby)\,\vec{e}(\bby),\vec{e}^{{}\,\dagger}(\bx)
\frac{\overrightarrow{\dd}}{\dd\bx}\left(
\left.\frac{\overrightarrow{\dd}\!\cL(\bx,[\bu])}{\dd\bu_{\sigma}}\right|_{j^{\infty}_{\bx}(s)}\right)\right\rangle.
\end{equation*}
We finally recall that the total derivative $\Id/\Id\bx$ is defined\footnote{By definition,
$(\vec{\Id} f/\Id x^i){\bigr|}_{j^{\infty}(s)}(\bx)=\bigl(\vec{\dd}/\dd x^i (
f{\bigr|}_{j^{\infty}(s)})\bigr)(\bx)$ for differential functions~$f$, 
see~\cite{TwelveLectures,ClassSym,Olver}.}
via an application of $\dd/\dd\bx$ to restriction to
infinite jets $j^{\infty}(s)$ of sections $s$ at base points $\bx$. Therefore,
the above expression is equal to
\begin{equation*}
\stackrel{\text{def}}{=}\int_M\Id\bby\int_M\dvol(\bx)
\left\langle\delta s(\bby)\,\vec{e}(\bby),\vec{e}^{{}\,\dagger}(\bx)
\left.\left(\left(-\frac{\vec{\Id}}{\Id\bx}\right)
\frac{\overrightarrow{\dd}\!\cL(\bx,[\bu])}{\dd\bu_{\sigma}}\right)\right|_{j^{\infty}_{\bx}(s)}\right\rangle.
\end{equation*}
This shows that an integration by parts over the base $M$ in the geometry of test shift $\delta s$ reappears as
integration by parts in the bundle where lives the background section~$s\in\Gamma(\pi)$. 

Repeating the integration by parts $|\sigma|\ge0$ times in each term of the sum in~\eqref{EqStartSum}, we obtain
the expression
\begin{equation*}
\sum_{i,j}\sum_{|\sigma|\ge0}\int_M\Id\bby\int_M\dvol(\bx)
\left\langle\delta s^i(\bby)\,\vec{e}_i(\bby),\vec{e}^{{}\,\dagger j}(\bx)
\left.\left(\left(-\frac{\vec{\Id}}{\Id\bx}\right)^{\sigma}
\frac{\overrightarrow{\dd}\!\cL(\bx,[\bu])}{\dd\bu^j_{\sigma}}\right)\right|_{j^{\infty}_{\bx}(s)}\right\rangle.
\end{equation*}
Let us recall once more that the coupling's support is the diagonal in $M\times M$, at points of which the value
$\langle\vec{e}_i(\bby),\vec{e}^{{}\,\dagger j}(\bx)\rangle$ is the Kronecker symbol $\delta_i^j$.
Consequently, we arrive~at
\begin{equation*}
{}=\sum_{i,j}\sum_{|\sigma|\ge0}\int_M\dvol(\bx)\,\delta s^i(\bx)\cdot
\left.\left(\left(-\frac{\vec{\Id}}{\Id\bx}\right)^{\sigma}
\frac{\overrightarrow{\dd}\!\cL(\bx,[\bu])}{\dd\bu^i_{\sigma}}\right)\right|_{j^{\infty}_{\bx}(s)}.
\end{equation*}
This is formula~\eqref{EqVariation}; it is familiar from any textbook on variational principles of classical mechanics
(e.\,g., see~\cite[\S 12\/--\/13]{MMKM}).

A standard reasoning shows that, whenever a response of the functional's value $S(s)$ to a test shift of $s$ along any
direction $\delta s$ vanishes, the Euler\/--\/Lagrange equation holds:
\begin{equation}\label{EqEqEL}
\left.\frac{\overleftarrow{\delta}\!S}{\delta\bu}\right|_{j^{\infty}(s)}=0.
\end{equation}
Its left-hand side belongs to the space $\Gamma(T^*\pi)$ of sections of the cotangent bundle to $\pi$.

\begin{rem}\label{RemELEqLinear}
This conclusion tells us that traditional attempts of a brute-force labelling of equations in a given system~\eqref{EqEqEL}
by using the unknowns $\bu$ is not geometric. Indeed, the equations' left-hand sides are sections of a \textsl{vector} bundle,
thus forming linear $\Bbbk$-vector spaces so that  addition is well defined for the equations within a system. On the other
hand, the fibres in the bundle $\pi$ can be smooth manifolds (i.\,e., not necessarily being vector spaces) so that one may
not add points of those fibres; for such operation is in general not defined at all. Even if $\pi$ is a vector bundle,
the fibres of which are endowed with linear vector space structure, the two structures are not related.
\end{rem}

\begin{rem}\label{RemNoEffect}
The integration by parts transforms a derivative $\dd/\dd\bby$ along one copy of the base $M$ to the minus derivative
$-\dd/\dd\bx$ along the other copy. This produces no visible effect on the mechanism which ensures a restriction onto
the diagonal in $M\times M$, i.\,e., there appears no would-be third term in the Leibniz rule for the product which is
defined only on the diagonal. A desperate prescription~\eqref{EqSur} was introduced in the literature in order to mimick
this paradoxical analytic behaviour of the coupling between elements of dual bases.
\end{rem}

\subsection{Why are variations permutable\,?}\label{SecSynonyms}
Having outlined the matching of geometries in the course of one sequence of integrations by parts for \textsl{one} fixed
pair $M\times M\ni(\bby,\bx)$ of copies of the base manifold, we emphasize that such integrations must be performed last,
i.\,e., only when the objects at hand are finally viewed as maps $\Gamma(\pi)\to\Bbbk$.

Should one haste in absence of clear understanding of what is actually being done and for which purpose, further
calculation of higher-order variations could predictably but uncontrollably lead to meaningless, manifestly erroneous
conclusions (e.\,g., compare left- and right-hand sides in~\eqref{Eq2Ways} below).

Namely, there exist integral functionals which determine equal maps $\Gamma(\pi)\to\Bbbk$ but, belonging to different spaces, 
behave differently in the course of variations, should one attempt any. 
We say that such functionals are 
\textsl{synonyms}; for instance, see Example~\ref{Countercounterexample} in the next section for a nontrivial synonym $\Delta G$ of the zero
functional (cf.\ Fig.~\ref{FigSynonyms}). 
\begin{figure}[htb]
\begin{center}{\unitlength=1mm
\linethickness{0.4pt}
\begin{picture}(50.00,42.00)
\put(0.00,42.00){\line(0,-1){4.00}}
\put(0.00,40.00){\line(1,0){50.00}}
\put(50.00,38.00){\line(0,1){4.00}}
\put(40.00,38.00){\vector(0,-1){28.00}}
\put(35.00,7.00){\line(0,-1){4.00}}
\put(35.00,5.00){\line(1,0){10.00}}
\put(45.00,3.00){\line(0,1){4.00}}
\put(-12,3.33){\llap{Map:}}
\put(-10,3.33){$\Gamma(\pi_{\BV})\to\Bbbk$}
\put(47.00,3){\makebox(0,0)[lb]{${}\ne0.$}}
\put(-1,30){\llap{Obj:}}
\put(2.67,28.00){\makebox(0,0)[lb]{$\overbrace{[\![\underbrace{\ \ F\ \ }{},\underbrace{\ \ \Delta G\ \ }{}]\!]}{}$}}
\put(41.33,22.00){\makebox(0,0)[lb]{$\int$}}
\put(5.33,26.33){\line(0,-1){2.33}}
\put(5.33,25.33){\line(1,0){5.00}}
\put(10.33,25.33){\line(0,1){1.00}}
\put(10.33,26.33){\line(0,-1){2.33}}
\put(17.00,26.33){\line(0,-1){2.33}}
\put(17.00,24.00){\line(0,1){1.33}}
\put(17.00,25.33){\line(1,0){9.00}}
\put(26.00,25.33){\line(0,1){1.00}}
\put(26.00,26.33){\line(0,-1){2.33}}
\put(20.67,22.67){\vector(0,-1){13.67}}
\put(20,3.33){\makebox(0,0)[lb]{0.}}
\put(17.33,14.00){\makebox(0,0)[lb]{$\int$}}
\end{picture}
}\end{center}
\caption{The synonyms~$\Delta G$ of zero functional yield constant maps 
$0\colon\Gamma(\pi_{\BV})\to\Bbbk$ yet they can nontrivially contribute to larger structures 
such as $\lshad F,\Delta G\rshad$, see Example~\ref{Countercounterexample}
on p.~\pageref{Countercounterexample}.}\label{FigSynonyms}
\end{figure}
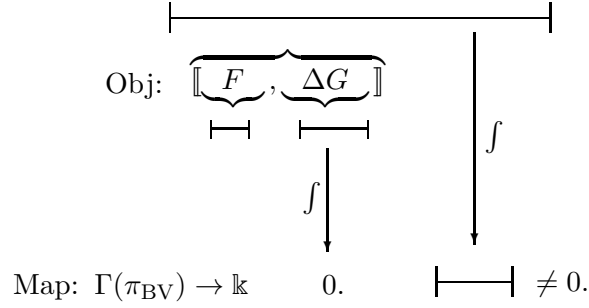
Informally speaking, the composite structure objects with repeated integrals over products 
$M\times M\times\ldots\times M$ of the base retain a kind of memory of the way how they were obtained from primary objects 
such as the action $S$.
Let us illustrate these claims.

\begin{example}
Let $\delta s_1\in\Gamma(T\pi)$ be a test shift at $s\in\Gamma(\pi)$ for an integral functional 
$S=\int\cL(\bx,[\bu])\dvol(\bx)$ with density $\cL$ of positive differential order. (That is, we suppose that some 
positive-order derivatives are always present in densities of all representatives of the cohomology class 
$S\in\ov{H}^n(\pi)$; this assumption is not to any extent restrictive but it allows us to not take  into account
$\ov{\Id}$-exact terms whose orders may not be bounded.) By using $S$, let us construct two new integral functionals.
First, we set
$$
F=\sum_i\sum_{|\sigma|\ge0}\int\dvol(\bx)\,\delta s_1^i(\bx)\cdot\left(-\frac{\vec{\Id}}{\Id\bx}\right)^{\sigma}
\left(\frac{\overrightarrow{\dd}\!\cL(\bx,[\bu])}{\dd u^i_{\sigma}}\right)\in\ov{H}^n(\pi),
$$
so that the mapping $F\colon\Gamma(\pi)\to\Bbbk$ is defined at $s\in\Gamma(\pi)$ by restriction of the integrand 
to the jet $j^{\infty}(s)$ and then by actual integration over $M$.

Let the other functional $G\in\ov{H}^{2n}(\pi,T\pi)$ be such that its value at the same section $s\in\Gamma(\pi)$ is
\begin{equation*}
G(s)=\sum_{i,j}\sum_{|\sigma|\ge0}\int_M\Id\bby\int_M\dvol(\bx)
\left\langle
(\delta s_1^i)\left(\frac{\overleftarrow{\dd}}{\dd\bby}\right)^{\sigma}(\bby)\,\vec{e}(\bby),
\vec{e}^{{}\,\dagger j}(\bx)\left.\frac{\overrightarrow{\dd}\!\cL(\bx,[\bu])}{\dd u^j_{\sigma}}
\right|_{j^{\infty}_{\bx}(s)}\right\rangle.
\end{equation*}
From the previous section it is clear that $F$ and $G$ are indistinguishable as mappings to $\Bbbk$ for every 
$s\in\Gamma(\pi)$. Yet their variations, i.\,e., the responses to an extra shift $\delta s_2\in\Gamma(T\pi)$, are
different. Indeed, they are equal to, first,
\begin{multline*}
\left(\left.\frac{\Id}{\Id\veps_2}\right|_{\veps_2=0}F\right)(s+\veps_2\overleftarrow{\delta}\!\!s_2)=
\sum_{i_1,i_2}\sum_{\substack{|\sigma_1|\ge0\\|\sigma_2|\ge0}}
\int_M\dvol(\bx)\,
\delta s_2^{i_2}(\bx)\delta s_1^{i_1}(\bx)\cdot\\
\cdot\left.\left\{\left(-\frac{\overrightarrow{\Id}}{\Id\bx}\right)^{\sigma_2}
\frac{\overrightarrow{\dd}}{\dd u^{i_2}_{\sigma_2}}
\left(\left(-\frac{\overrightarrow{\Id}}{\Id\bx}\right)^{\sigma_1}
\frac{\overrightarrow{\dd}\!\cL(\bx,[\bu])}{\dd u^{i_1}_{\sigma_1}}\right)\right\}\right|_{j^{\infty}_{\bx}(s)}.
\end{multline*}
The above formula corresponds to a step-by-step calculation within a na\"\i ve approach to the geometry of variations.
However, the genuine value of second variation of the integral functional $S$ along $\delta s_1$ and then $\delta s_2$
at a section $s$ is
\begin{multline*}
\left(\left.\frac{\Id}{\Id\veps_2}\right|_{\veps_2=0}G\right)(s+\veps_2\overleftarrow{\delta}\!\!s_2)=\sum_{\substack{i_1,i_2\\j_1,j_2}}\sum_{\substack{|\sigma_1|\ge0\\|\sigma_2|\ge0}}
\int_M\Id\bby_2\int_M\Id\bby_1\int_M\dvol(\bx) \\
\left\{
(\delta s^{i_2}_2)\left(\frac{\overleftarrow{\dd}}{\dd\bby_2}\right)^{\sigma_2}(\bby_2)\,\langle\vec{e}_{i_2}(\bby_2), 
\vec{e}^{{}\,\dagger j_2}(\bx)\rangle
\cdot 
(\delta s^{i_1}_1)\left(\frac{\overleftarrow{\dd}}{\dd\bby_1}\right)^{\sigma_1}(\bby_1)\,\langle\vec{e}_{i_1}(\bby_1)
,\vec{e}^{{}\,\dagger j_1}(\bx)\rangle
\right\}\\
\cdot\left.\frac{\overrightarrow{\dd}^2\cL(\bx,[\bu])}{\dd u^{j_2}_{\sigma_2}\dd u^{j_1}_{\sigma_1}}\right|_{j^{\infty}_{\bx}(s)}.
\end{multline*}
The analytic distinction between the operators
\begin{equation}\label{Eq2Ways}
\underbrace{
\left(-\frac{\overrightarrow{\Id}}{\Id\bx}\right)^{\sigma_2}\circ\frac{\overrightarrow{\dd}}{\dd u^{i_2}_{\sigma_2}}
\circ
\left(-\frac{\overrightarrow{\Id}}{\Id\bx}\right)^{\sigma_1}\circ\frac{\overrightarrow{\dd}}{\dd u^{i_1}_{\sigma_1}}
}_{\text{na\"\i ve approach}}
\quad\text{ and }\quad
\underbrace{
\left(-\frac{\overrightarrow{\Id}}{\Id\bx}\right)^{\sigma_1\cup\sigma_2}\circ
\frac{\overrightarrow{\dd}^2}{\dd u^{i_2}_{\sigma_2}\dd u^{i_1}_{\sigma_1}}
}_{\text{geometric theory}}
\end{equation}
reveals why in positive-order Lagrangian models it is forbidden to haste, which would imply that the derivatives along
distinct copies of $M$ for variations $\delta s_1,\ \ldots,\ \delta s_k$ are too early transformed to derivatives along the
functional's own base. Such a conceptual error would repercuss with inexplicable, redundant terms in variations to-follow.
\end{example}

On the other hand, as soon as the product-bundle geometry of iterated variations is properly realized --- so that all
restrictions to the diagonals are postponed as late as possible, --- the variations become (graded-)permutable.%
\footnote{\label{FootMKEW}An idea that iterated variations must be taken at nominally different points $\bx$ and $\bby$ has been in the air
for a long time (let us refer to~\cite[\S 1]{KontsevichVishikLong} which contains due credits to E.~Witten).
A somewhat less obvious fact is that those different points belong to different copies of the manifold $M$ in the product
bundle ${\pi\times T\pi\times\ldots\times T\pi}$ over ${M\times M\times\ldots\times M}$.}
Namely, denote by $|u^i|$, $1\le i\le m$, the overall $\BBZ_2$-valued parities of the fibre coordinates $u^i$; the ghost
parity $\GH(u^i)$ or individual $\BBZ$-{} or $\BBZ_2$-valued gradings in the bundle $\pi$ contribute additively to $|u^i|$ 
and then a residue modulo 2 is taken. Suppose that $\delta s_1=(\delta s_1^{i_1})$ and $\delta s_2=(\delta s_2^{i_2})$ are
test shifts and $S=\int\cL(\bx,[\bu])\dvol(\bx)$ is an integral functional which maps a section $s\in\Gamma(\pi)$ to $\Bbbk$.
Then, \emph{after} the integrations by parts in the product-bundle geometry $\pi\times T\pi\times T\pi$ which is described
above, there remains
\begin{multline*}
\sum_{i_1,i_2}\sum_{\substack{|\sigma_1|\ge0\\|\sigma_2|\ge0}}\int_M\dvol(\bx)\,\delta s_2^{i_2}(\bx)\delta s_1^{i_1}(\bx)
\left.\left(\left(-\frac{\overrightarrow{\Id}}{\Id\bx}\right)^{\sigma_1\cup\sigma_2}
\frac{\overrightarrow{\dd}^2\cL(\bx,[\bu])}
{\dd u^{i_2}_{\sigma_2}\dd u^{i_1}_{\sigma_1}}\right)\right|_{j^{\infty}_{\bx}(s)}\\
=
\sum_{i_1,i_2}\sum_{\substack{|\sigma_1|\ge0\\|\sigma_2|\ge0}}
(-)^{|u^{i_1}|\cdot|u^{i_2}|}
\int_M\dvol(\bx)\,\delta s_1^{i_1}(\bx)\delta s_2^{i_2}(\bx)
\left.\left(\left(-\frac{\overrightarrow{\Id}}{\Id\bx}\right)^{\sigma_1\cup\sigma_2}
\frac{\overrightarrow{\dd}^2\cL(\bx,[\bu])}
{\dd u^{i_1}_{\sigma_1}\dd u^{i_2}_{\sigma_2}}\right)\right|_{j^{\infty}_{\bx}(s)}.
\end{multline*}
Likewise, higher-order iterated variations with $k\ge2$ test shifts $\delta s_1,\dots,\delta s_k$ are (gra\-ded-)\/permutable
with the same rule of signs for permutations of order in which the (gra\-ded) partial derivatives
$\overrightarrow{\dd}/\dd u^{i_1}_{\sigma_1},\dots,\overrightarrow{\dd}/\dd u^{i_k}_{\sigma_k}$ fall from the left on the 
density $\cL$ of the functional $S$. (A case of $\BBZ_2$-graded base manifold $M^{(n_0|n_1)}$ would bring more signs which
are also captured in a standard way.)

Let there be $k\ge2$ variations $\delta s_1,\dots,\delta s_k\in\Gamma(T\pi)$. We finally have that
\begin{multline}\label{EqkthVariation}
\left.\frac{\Id}{\Id\veps_k}\right|_{\veps_k=0}\circ\ldots\circ\left.\frac{\Id}{\Id\veps_1}\right|_{\veps_1=0}
S(s+\veps_1\overleftarrow{\delta}\!\!s_1+\ldots+\veps_k\overleftarrow{\delta}\!\!s_k)=\\=
\sum_{\substack{i_1,\dots,i_k\\j_1,\dots,j_k}}\sum_{\substack{|\sigma_1|\ge0\\\dots\\|\sigma_k|\ge0}}
\int_M\Id\bby_k\ldots\int_M\Id\bby_1\int_M\dvol(\bx)\cdot{} \\
\left\{
(\delta s_k^{i_k})\left(\frac{\overleftarrow{\dd}}{\dd\bby_k}\right)^{\sigma_k}(\bby_k)
\left\langle\vec{e}_{i_k}(\bby_k),\vec{e}^{{}\,\dagger j_k}(\bx)\right\rangle
\cdot\ldots\cdot 
(\delta s_1^{i_1})\left(\frac{\overleftarrow{\dd}}{\dd\bby_1}\right)^{\sigma_1}(\bby_1)
\left\langle\vec{e}_{i_1}(\bby_1),\vec{e}^{{}\,\dagger j_1}(\bx)\right\rangle
\right\}\\
\cdot\left.\frac{\overrightarrow{\dd}^k\cL(\bx,[\bu])}{\dd u^{j_k}_{\sigma_k}\ldots\dd u^{j_1}_{\sigma_1}}
\right|_{j^{\infty}_{\bx}(s)}.
\end{multline}
Whenever any $k-1$ variation(s) are fixed in the above formula, the co-multiple $|\,\rangle$ of the remaining, $\ell$th
variation $\delta s_{\ell}=\langle\delta s^{i_{\ell}}_{\ell}(\bby_k)\vec{e}_{i_{\ell}}(\bby_{\ell})|$ is an element of
the cotangent vector space $T^*_{s(\bx)}\pi^{-1}(\bx)=V_{\bx}^{\dagger}$ at the point $s(\bx)$ in the fibre $\pi^{-1}(\bx)$
over a base point $\bx\in M^n$.

\begin{rem}
The composite object in the left-hand side of equality~\eqref{EqkthVariation} is an integral functional in the bundle
${\pi\times T\pi\times\ldots\times T\pi}$ which properly contains the geometry of $k$ variations from $\Gamma(T\pi)$,
see Fig.~\ref{FigLelystad}.
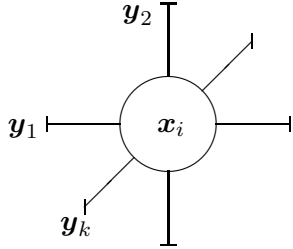
\begin{figure}[htb]
\begin{center}{\unitlength=1mm
\begin{picture}(40,30)(-20,-15) 
\put(0,0){\circle{12}}
\put(-1.5,-1){$\bx_i$}
\put(6.25,0){\line(1,0){9.75}}
\put(-6.25,0){\line(-1,0){9.75}}
\put(0,6.25){\line(0,1){9.75}}
\put(0,-6.25){\line(0,-1){9.75}}
\put(-11,-11){\line(1,1){6.5}}
\put(4.45,4.45){\line(1,1){6.55}}
\put(-16,-1){\line(0,1){2}}
\put(16,-1){\line(0,1){2}}
\put(-1,16){\line(1,0){2}}
\put(-1,-16){\line(1,0){2}}
\put(-11,-12){\line(0,1){2}}
\put(11,10){\line(0,1){2}}
\put(-17,-1){\llap{$\bby_1$}}
\put(-2,14){\llap{$\bby_2$}}
\put(-10,-14){\llap{$\bby_k$}}
\end{picture}
}\end{center}
\caption{Each variation $\delta s_1$,\ $\ldots$,\ $\delta s_k$ brings its own copy
of the base $M^n\ni\bby_\ell$ into the product bundle 
$\pi\times T\pi\times\ldots\times T\pi$ 
over $M\times M\times\ldots\times M$.}\label{FigLelystad}
\end{figure}
This construction lives not on a Whitney sum ${\pi\mathbin{{\times}_M}T\pi\mathbin{{\times}_M}\ldots\mathbin{{\times}_M}T\pi}$ over the base manifold~$M$;
that would force an untimely restriction to the diagonal in the product 
${M\times M\times\ldots\times M}$ of bases and 
hence reproduce the old difficulties of the theory.
\end{rem}

\subsection{The spaces of functionals}
The integral functionals $S\in\ov{H}^n(\pi)$, which we have been dealing with until now, are building blocks in a wider class
of mappings $\Gamma(\pi)\to\Bbbk$. By viewing elements of $\Gamma(\pi)$ as ``points'' and functionals from
$\ov{H}^n(\pi)$ as ``elementary functions'' (see~\cite{ClassSym} and references therein), we consider pointwise-defined
(formal sums of) products of such maps, e.\,g., we let
$$(S_1\cdot S_2)(s)\stackrel{\text{def}}{=}S_1(s)\cdot S_2(s)$$
for any two already defined functionals $S_1$ and $S_2$; the binary operation $\cdot$ for their values at $s\in\Gamma(\pi)$
is the usual multiplication of $\Bbbk$-numbers ($\Bbbk=\BBR$ or $\BBC$). By definition, we put
$$
\ov{\gN}^n(\pi,T\pi)=\bigoplus_{\ell=1}^{+\infty}\bigotimesk_{i=1}^{\ell}\bigoplus_{k=0}^{+\infty}
\ov{H}^{n(1+k)}(\pi\times\underbrace{T\pi\times\ldots\times T\pi}_{k\text{ variations}}).
$$
This space contains the linear subspace of \textsl{local functionals},
$$
\ov{\gM}^n(\pi)=\bigoplus_{\ell=1}^{+\infty}\bigotimesk_{i=1}^{\ell}\ov{H}^n(\pi),$$
for instance, such as the standard weight factor $\exp(\tfrac{\boldi}{\hbar}S^{\hbar})$ in BV-models with quantum BV-action
$S^{\hbar}$ (see section~\ref{SecGauge} below, cf.~\cite{BeilinsonDrinfeld}).
The larger space $\ov{\gN}^n(\pi,T\pi)\supsetneq\ov{\gM}^n(\pi)$
harbours local functionals \emph{and} their variations of arbitrarily high order. The (products of) integral functionals
in $\ov{\gM}^n(\pi)\supset\ov{H}^n(\pi)$ could be viewed as primary objects. In the course of variations, their descendants 
in $\ov{\gN}^n(\pi,T\pi)$ absorb new test shifts and retain the information about initial building blocks from $\ov{H}^n(\pi)$.
This memory governs the analytic behaviour of descendants in operations such as 
calculation of the BV-Laplacian or taking
the Schouten bracket; we also refer to sections~\ref{SecSynonyms} above
and~\ref{SecMaster} in what follows. The composite structure of the bundle 
$\pi\times T\pi\times\ldots\times T\pi$ is crucial whenever one wants to not only describe 
initial setup such as a given
BV-model but to perform rigorous calculations in it, handling higher\/-\/order variations of objects 
(e.\,g., third\/-\/order
variations occur in~\eqref{EqZimes} on p.~\pageref{EqZimes}, see also Example~\ref{Countercounterexample} 
on p.~\pageref{Countercounterexample} below, ---~and the order is equal to
four in property~\eqref{EqDeltaSquareIntro} for the BV-Laplacian $\Delta$ to be a differential). The geometric approach to 
(graded-)permutable variations of functionals makes such calculations well-defined and proofs free from 
any \textsl{ad hoc} regularisation recipes.

\section{The geometry of Batalin\/--\/Vilkovisky formalism}\label{SecDef}
\noindent%
The geometry of variations which we analysed in the previous section was not specific to a bundle $\pi$ of unknowns.
In this section we first recall a construction of the BV-superbundle whose fibres are endowed with $\BBZ_2$-valued
ghost parity. By definition, the BV-bundle $\boldsymbol{\pi}_{\BV}^{(0|1)}=\pi^*_{\infty}(\boldsymbol{\zeta}_\infty^{(0|1)})$
is induced from the Whitney sum
$
\boldsymbol{\zeta}^{(0|1)}=\zeta_0\mathbin{{\times}_M}\zeta_1\mathbin{{\times}_M}\ldots\mathbin{{\times}_M}\zeta_{\lambda}
\mathbin{{\times}_M}
\Pi\widehat{\zeta}_0\mathbin{{\times}_M}\Pi\widehat{\zeta}_1\mathbin{{\times}_M}\ldots\mathbin{{\times}_M}
\Pi\widehat{\zeta}_{\lambda}
$
of some $\BBZ_2$-graded vector bundles over $M$ (in what follows we sum up the construction of 
$\zeta_0,\ldots,\zeta_{\lambda}$
and their parity-reversed duals
$\Pi\widehat\zeta_0,\ldots,\Pi\widehat\zeta_{\lambda}$)
by the infinite jet bundle $\pi_{\infty}\colon J^{\infty}(\pi)\to M$ associated with the smooth fibre bundle $\pi$ 
of physical fields.%
\footnote{A subtle point, which we reconsider in section~\ref{SecBVzoo} (see also Remark~\ref{RemELEqLinear}), is that
the \emph{fibre} bundle $\pi$ is often identified with the \emph{vector} bundle component $\zeta_0$ in
$\boldsymbol{\zeta}^{(0|1)}$. Nevertheless, it is the construction of induced bundle 
$\pi^*_{\infty}(\zeta_0\mathbin{{\times}_M}\ldots)$ 
by using which the physical fields and their derivatives are remembered by the Euler\/--\/Lagrange equations (referred to
$\zeta_0$), Noether's identities (in $\zeta_1$), and higher geberations of syzygies from $\zeta_2,\ldots,\zeta_{\lambda}$
(if any).}

\subsection{The BV-zoo}\label{SecBVzoo}
Let a fibre bundle $\pi$ of physical fields over the base manifold $M^n$ be given and denote by $\phi$ the fibre coordinates
in it. Suppose that 
$${S_0=\int\cL_0(\bx,[\phi])\dvol(\bx)\in\ov{H}^n(\pi)}$$
is the action of a field model under study. By using the theory and techniques from section~\ref{SecEL} we know how one
derives, via the stationary point condition $\overleftarrow{\delta S}{\bigr|}_s=0$ at $s\in\Gamma(\pi)$ the Euler\/--\/Lagrange
equations of motion $\cE_{\EL}=\{\overleftarrow{\delta S_0}/\delta\phi=0\}$ whose left-hand sides belong to the
$C^{\infty}(J^{\infty}(\pi))$-module of sections $P_0=\Gamma(\pi^*_{\infty}(\zeta_0))$ for the cotangent bundle $\zeta_0$
to $\pi$ such that 
$\overleftarrow{\delta S_0}/\delta\phi|_{j^{\infty}(s)}\cdot\dvol(\cdot)\in\Gamma(T^*\pi)\otimes_{C^{\infty}(M)}\Lambda^n(M)$
for any field configuration $s\in\Gamma(\pi)$.

We recall from Remark~\ref{RemELEqLinear} that by following a misfortunate but long-established tradition it is the unknowns
$\phi$ in $\pi$ but not the global coordinates $\bF$ in the fibre of cotangent bundle $T^*\pi$ to $\pi$ which are used to
parametrise the equations within Euler\/--\/Lagrange system $\cE_{\EL}$ at points of the graph of a 
section~$\phi\in\Gamma(\pi)$.

If the model at hand is gauge-invariant, then it admits an off-shell differential dependence
$\boldsymbol{\Phi}(\bx,[\phi];[\bF])\equiv0\in\Gamma((\pi_{\infty}\mathbin{{\times}_M}\zeta_{0,\infty})^*(\zeta_1))$
between the left-hand sides $\bF$ of equations $\cE_{\EL}$. We recall further
that the dependence of Noether's identities $\boldsymbol{\Phi}$ on (the derivatives of) $\bF$ is \emph{linear} for
Euler\/--\/Lagrange systems $\cE_{\EL}$; the generators $\bp(\bx,[\phi])\in\widehat P_1=\Gamma(\pi^*_{\infty}(\widehat\zeta_1))$
of Noether's gauge symmetries for $S_0$ are sections of the bundle $\widehat\zeta_1$ which is induced from the dual to
$\zeta_1$ with respect to the top-degree horizontal form-valued coupling $\langle\,,\,\rangle$. Indeed, if
$$
0\equiv\left\langle\bp,\boldsymbol{\Phi}(\bx,[\phi];[\bF])\right\rangle
$$
and $\boldsymbol{\Phi}$ is linear in $\bF$ or its finite-order derivatives,
$$\boldsymbol{\Phi}(\bx,[\phi];[\bF])=\ell^{(\bF)}_{\boldsymbol{\Phi}}(\bF)\equiv0,$$
then an integration by parts yields that
$$
0\cong\left\langle(\ell^{\,(\bF)}_{\boldsymbol{\Phi}})^{\dagger}(\bp),\delta S_0/\delta\phi\right\rangle\cong
\vec{\dd}^{\,(\phi)}_{(\ell^{(\bF)}_{\boldsymbol{\Phi}})^{\dagger}(\bp)}(S_0).
$$
This shows that the evolutionary vector field $\vec{\dd}^{\,(\phi)}_{A(\bp)}$
with $A=(\ell^{(\bF)}_{\boldsymbol{\Phi}})^{\dagger}$ and $\bp=\bp(\bx,[\phi])$ is a Noether symmetry of the action $S_0$.
By reading the above equalities backwards, one obtains the linear Noether relations
$\boldsymbol{\Phi}=A^{\dagger}(\bF)$ between the Euler\/--\/Lagrange equations of motion.

Likewise, there could in principle appear higher generations of linear identities\linebreak
${\Psi_2(\bx,[\phi],[\bF];[\boldsymbol{\Phi}])\equiv0}$, $\dots$, 
${\Psi_{\lambda}(\bx,[\phi],[\bF],[\boldsymbol{\Phi}],\dots,[\Psi_{\lambda-2}];[\Psi_{\lambda-1}])\equiv0}$
which hold for all~$\phi$, sections~$\bF$ in~$\zeta_0$, and so on up to the coordinates $\Psi_{\lambda-2}$.
Each $i$th generation of such identities arises with the respective vector bundle $\zeta_i$
with fibre dimension~$m_i$; 
the total number of generations
is bounded from above by a constant $\lambda\in\BBN\cup\{0\}$ due to Hilbert's theorem on syzygies~\cite{Eisenbud}:
$0\le i\le\lambda\le n$, where $n$ is the dimension of base manifold $M^n$. For example, we have that $\lambda=1$
for Yang\/--\/Mills theory, and $\lambda=2$ for gravity over a fourfold $M^4$.

We denote by $\bF$ (alas! at once identifying this global $m$-tuple in $\zeta_0$ for the equations with the local field 
variables $\phi$), and by $\boldsymbol{\gamma}^{\dagger}$,\ $\bc^{\dagger}$,\ $\dots$,\ $\bc_{\lambda}^{\dagger}$ the global fibre coordinates
in $\zeta_1$ for Noether's identities, and so on up to $\zeta_{\lambda}$, respectively (see Fig.~\ref{FigBVSetup}).
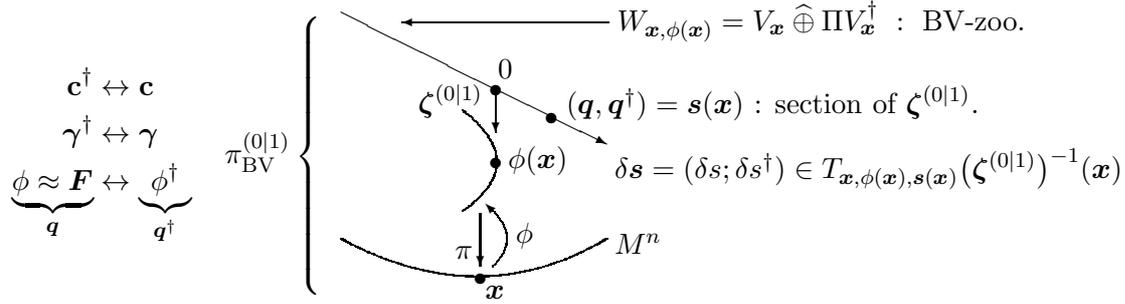
\begin{figure}[htb]
$\begin{aligned}
\bc^{\dagger}&\leftrightarrow\bc\\
\boldsymbol{\gamma}^{\dagger}&\leftrightarrow\boldsymbol{\gamma}\\
\underbrace{\phi\approx\bF}_{\bq}&\leftrightarrow\lefteqn{\underbrace{\phi^{\dagger}}_{\bq^{\dagger}}}
\end{aligned}
\qquad\pi_{\BV}^{(0|1)}\left\{
\text{
\unitlength=1mm
\linethickness{0.4pt}
\begin{picture}(102.00,20.00)
\put(-5,-20){\begin{picture}(0,0)
\bezier{160}(5.00,10.00)(23.33,0.00)(40.00,10.00)
\put(41,8){\makebox(0,0)[lb]{$M^n$}}
\put(23.33,4.67){\circle*{1.33}}
\put(24,2){\makebox(0,0)[lb]{$\bx$}}
\put(23.33,13.33){\vector(0,-1){7.00}}
\bezier{88}(21.00,13.67)(30.00,20.00)(21.00,27.00)
\put(25.33,20.00){\circle*{1.33}}
\put(5.00,40.00){\vector(2,-1){35.00}}
\put(25.33,29.67){\circle*{1.33}}
\put(25.33,29.67){\vector(0,-1){5.67}}
\put(40.33,38.67){\vector(-1,0){27.33}}
\put(41.33,36.17){\makebox(0,0)[lb]{$W_{\bx,\phi(\bx)}=V_{\bx}
\mathbin{\widehat{\oplus}}\Pi V_{\bx}^{\dagger}\ :\text{ BV-zoo.}$}}
\put(25.67,31.33){\makebox(0,0)[lb]{$0$}}
\put(41.00,16.5){\makebox(0,0)[lb]{$\delta\bolds=
   (\delta s;\delta s^{\dagger})\in T_{\bx,\phi(\bx),\bolds(\bx)}\bigl(\boldsymbol{\zeta}^{(0|1)}\bigr)^{-1}(\bx)$}}
\put(27.00,18.33){\makebox(0,0)[lb]{$\phi(\bx)$}}
\put(15.17,26.00){\makebox(0,0)[lb]{$\boldsymbol{\zeta}^{(0|1)}$}}
\put(20.00,7.33){\makebox(0,0)[lb]{$\pi$}}
\bezier{40}(25.00,6.33)(28.67,10.00)(25.00,13.67)
\put(25.00,13.67){\vector(-1,1){1.00}}
\put(28.00,8.67){\makebox(0,0)[lb]{$\phi$}}
\put(32.67,26.00){\circle*{1.33}}
\put(34.67,25.67){\makebox(0,0)[lb]{$(\bq,\bq^{\dagger})=\bolds(\bx)$ : section of $\boldsymbol{\zeta}^{(0|1)}$.}}
\end{picture}}\end{picture}
}
\right.$
\caption{The fibre bundle $\pi$ of physical fields~$\phi$, the bundle~$\boldsymbol{\zeta}^{(0|1)}$ of
BV-\/variables~$(\bq,\bq^{\dagger})$, and the vector bundle~$T\boldsymbol{\zeta}^{(0|1)}$ of
their variations~$\delta\bolds=(\delta s;\delta s^{\dagger})$.}\label{FigBVSetup}
\end{figure}

In turn, each vector bundle $\zeta_0$,\ $\dots$,\ $\zeta_{\lambda}$ brings its $\langle\,,\,\rangle_i$-\/dual $\widehat\zeta_i$
into the picture. (Note that the equations $\overleftarrow{\delta S_0}{\bigr|}_s=0$ upon $s\in\Gamma(\pi)$ for
$S_0=\int\cL(\bx,[\phi])\cdot\dvol(\bx)$ and all equations' linear-differential descendants retain the volume form
$\dvol(\bx)$ from the model's action $S_0$ at all points~$\bx\in M^n$.)

We now reverse the parity of linear vector space fibres in 
$\widehat\zeta_0$,\ $\dots$,\ $\widehat\zeta_{\lambda}$ by introducing
the $\BBZ_2$-valued ghost parity $\GH(\cdot)$ and considering the odd \textsl{neighbours} 
$\Pi\widehat\zeta_0$,\ $\dots$,\ $\Pi\widehat\zeta_{\lambda}$ of the dual vector bundles (see~\cite{Galli10, Voronov2002} and also
Appendix~A in~\cite{Norway13} for discussion). Let us denote by 
$\phi^{\dagger}$,\ $\boldsymbol{\gamma}$,\ $\bc$,\ $\ldots$,\ $\bc_{\lambda}$ the ghost parity-odd global coordinates 
along linear vector space fibres in 
$\Pi\widehat\zeta_0$,\ $\dots$,\ $\Pi\widehat\zeta_{\lambda}$, respectively.
These variables' proper names are easily recognized from 
the standard notation: $\phi$ replacing $\bF$ are the fields and $\phi^{\dagger}$ are odd-parity \textsl{antifields},
$\boldsymbol{\gamma}$ are the odd \textsl{ghosts} and $\boldsymbol{\gamma}^{\dagger}$ are the parity-even \textsl{antighosts},
whereas the canonically conjugate variables $\bc\leftrightarrow\bc^{\dagger}$, \dots, 
$\bc_{\lambda}\leftrightarrow\bc_{\lambda}^{\dagger}$ are higher ghost-antighost pairs of opposite ghost parities 
(resp., odd and even). We denote by $\bq$ the agglomeration of ghost parity-even variables and by $\bq^{\dagger}$ their
respective canonically conjugate parity-odd neighbours.%
\footnote{Consider Feynman's path integral $\int_{\Gamma(\zeta^0)}[D\bq]\,\cO([\bq],[\bq^{\dagger}])$ of an 
observable~$\cO$ over the space of ghost parity-even sections.
The BV-\/Laplacian $\Delta$ is the tool which ensures the integral's effective independence from the unphysical
ghost parity-odd variables~$\bq^{\dagger}$, see section~\ref{SecMaster}.}

\begin{rem}\label{RemBiGraded}
Let us emphasize that by using the word ``parity'' we always refer to the ghost parity $\GH(\,\cdot\,)$ of objects.%
\footnote{By construction, the ghost parities of canonically conjugate BV-variables are complementary modulo 2, that is,
to each even-parity variable $q$ there corresponds its odd-parity dual neighbour $q^{\dagger}$. Of course, there remains
much freedom in a choice of the integer ghost numbers followed by the group homomorphism 
$(-)^{\GH(\,\cdot\,)}\colon\BBZ\to\BBZ_2$. For example, let $(q,q^{\dagger})$ be a pair of conjugate BV-variables; 
then one balances $\GH(q)=\GH(q^{\dagger})\pm1$ or $\GH(q)=-\GH(q^{\dagger})\pm1$,
 or by using any other integers such that
one is even and the other is odd. Obviously any shift by an even integer (e.\,g., 
$\GH(q)\mapsto-\GH(q)=\GH(q)-2\cdot\GH(q)$)
does not alter any values in the parity group $\BBZ_2$; this is no more than another way to describe the same theory.}
In this paper we aim at understanding the geometry of variations so that the graded arithmetic and algebra of derivations
play auxiliary 
r\^oles. However, as soon as the interaction of geometries is properly fixed,
their extension to a $\BBZ_2$-graded setup of superbundle $\pi\colon E^{(m_0+n_0|m_1+n_1)}\to M^{(n_0|n_1)}$ 
of physical fields (possibly, over a base supermanifold $M^{(n_0|n_1)})$ makes no conceptual difficulty
(\cite{BerezinAA}, see also~\cite{HenneauxTeitelboim} and references therein). The theory then becomes bi\/-\/graded: 
it involves 
(i) the $\BBZ_2$-grading $|\cdot|$ in the ring of field coordinates, which echoes in the $\BBZ_2$-grading of 
Euler\/--\/Lagrange equations of motion, Noether identities, etc., (the model's action functional $S_0$
has even grading by default), and 
(ii) the ghost parity $\GH(\cdot)$, see~\cite{Voronov2002}. 

The $\BBZ_2$-grading $|\cdot|$ and the ghost parity
$\GH(\cdot)$ are independent from each other. We denote by $\bq=\bq^{(0|1)}$ the ghost parity-even BV-fibre variables, 
which are then grouped in even- and odd-grading components. Likewise, the ghost parity-odd BV-variables
$\bq^{\dagger}=(\bq^{\dagger})^{(0|1)}$ are arranged in exactly the same way. By construction, the values of
$\BBZ_2$-gradings for
canonically conjugate variables $(\bq,\bq^{\dagger})$ coincide: we have that
$|\bq|=|\bq^{\dagger}|$ and $\GH(\bq^{\dagger})\equiv\GH(\bq)+1\mod2$.
\end{rem}

Next, we take the Whitney sum
$$\boldsymbol{\zeta}^{(0|1)}\stackrel{\text{def}}{=}\zeta_0\mathbin{{\times}_M}\zeta_1\mathbin{{\times}_M}\ldots\mathbin{{\times}_M}\zeta_{\lambda}
\mathbin{{\times}_M}\Pi\widehat\zeta_0\mathbin{{\times}_M}\Pi\widehat\zeta_1\mathbin{{\times}_M}\ldots\mathbin{{\times}_M}\Pi\widehat\zeta_{\lambda}$$
of the double set of dual bundles with opposite ghost parities of fibre coordinates. 
Fin\-al\-ly, let us lift the Whitney sum of infinite jets of those bundles, 
putting it over the bun\-dle of physical fields by using a pull-back under $\pi_{\infty}$. We denote the
resulting bundle over the total space $J^{\infty}(\pi)\to M$ by 
\[
\pi_{\BV}^{(0|1)}=\pi_\infty^*\bigl(\boldsymbol{\zeta}_\infty^{(0|1)}\bigr).
\]
The fibre $W_{\bx}=V_{\bx}\mathbin{\widehat{\oplus}}\Pi V_{\bx}^{\dagger}$ of $\boldsymbol{\zeta}^{(0|1)}$
admits the canonical decomposition in two dual halves of opposite parities;%
\footnote{To highlight this duality between ghost parity-even vector space $V_{\bx}$ and ghost parity-odd subspace
$\Pi V_{\bx}^{\dagger}$ in $W_{\bx}$, we use the notation $\widehat{\oplus}$ for their direct sum; whenever a coordinate
in $V_{\bx}$ is rescaled by $\const$ times, the respective conjugate variable in $\Pi V_{\bx}^{\dagger}$ is transformed
inverse-proportionally by $\const^{-1}$ times, see Remark~\ref{RemRescaleBoth} below.}
this is shown in Fig.~\ref{FigDual}.
\begin{figure}[htb]
\begin{center}{\unitlength=0.5mm
\linethickness{0.4pt}
\begin{picture}(90.00,35.00)
\put(20.00,20.00){\circle*{1.33}}
\put(20.00,20.00){\vector(1,0){20.00}}
\put(20.00,20.00){\vector(-1,0){20.00}}
\put(20.00,20.00){\vector(0,1){20.00}}
\put(20.00,20.00){\vector(0,-1){20.00}}
\put(27.33,31.00){\makebox(0,0)[lb]{$V_x$}}
\put(-3,18.33){\llap{$W_{\bx}={}$}}
\put(42.67,16){\makebox(0,0)[lb]{$\widehat{\bigoplus}$}}
\put(55.00,20.00){\vector(1,0){18.00}}
\put(95.00,20.00){\vector(-1,0){18.00}}
\put(75.00,40.00){\vector(0,-1){18.00}}
\put(75.00,0.00){\vector(0,1){18.00}}
\put(75.00,20.00){\circle{1.33}}
\put(84.00,30.67){\makebox(0,0)[lb]{$(\Pi)V_x^{\dagger}$}}
\end{picture}
}\end{center}
\caption{The BV-\/fibre is a direct sum of dual vector spaces; one is parity\/-\/even
and the other is proclaimed ghost parity\/-\/odd.}\label{FigDual}
\end{figure}
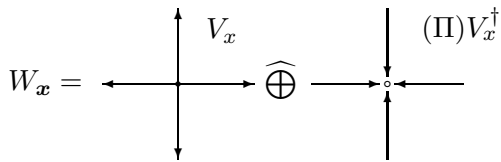

Bearing in mind that the fields $\phi$ are artifically incorporated into the newly built fibre by $\zeta_0$, we shall 
omit 
an ever-present reference to points $(\bx,\phi(\bx))$ of jets of sections of the initial bundle $\pi$ when dealing with
variations $\delta\bolds=(\delta s;\delta s^{\dagger})$ for sections $s$ of $\boldsymbol{\zeta}^{(0|1)}$ at $\phi(\bx)$,
see Fig.~\ref{FigBVSetup}. 

\subsection{The signs convention in Nature}\label{SecSigns}
The construction of canonically conjugate pairs of global coordinates $(\bq,\bq^{\dagger})$ in the fibres
$W_{\bx}=V_{\bx}\mathbin{\widehat{\oplus}}\Pi V_{\bx}^{\dagger}$ refers to a choice of the smooth field of dual bases in the two 
subspaces of even and odd ghost parity. Suppose that $\vec{e}_i(\bx)$ is a frame in $V_{\bx}$ and 
$\vec{e}^{{}\,\dagger i}(\bx)$ is its dual in $\Pi V_{\bx}^{\dagger}$, where the index $i$ runs from 1 to the 
total dimension of even- and odd-parity component in the fibre of $\boldsymbol{\zeta}^{(0|1)}$; we denote by
$N=m+m_1+\ldots+m_{\lambda}$ each of the two dimensions so that the fibre of the Whitney sum $\boldsymbol{\zeta}^{(0|1)}$
has superdimension $(N|N)$.

Let us recall that it is the parity of coordinates $\bq^{\dagger}$ but not of the vectors $\vec{e}^{{}\,\dagger i}$
in a basis which is reversed by the operation $\Pi$. The odd-parity component in the vector bundle
$\boldsymbol{\zeta}^{(0|1)}$ is topologically indistinguishable from 
$\widehat{\zeta}_0\mathbin{{\times}_M}\ldots\mathbin{{\times}_M}\widehat{\zeta}_{\lambda}$
but the rules become new for arithmetic in the algebra of coordinate functions on the total space. Therefore, we let the
notation $\vec{e}^{{}\,\dagger i}(\bx)$ be identical for the same bases in $V_{\bx}^{\dagger}$ and $\Pi V_{\bx}^{\dagger}$.

\begin{rem}\label{RemTwoCouplings}
The presence of \emph{two} dual vector spaces, $V_{\bx}$ and $(\Pi)V_{\bx}^\dagger$, standardly implies that there are \emph{two}
couplings,
\begin{equation}\label{EqTwoCouplings}
\langle\,,\,\rangle\colon V_{\bx}\times(\Pi)V_{\bx}^{\dagger}\to\Bbbk \quad \text{and}\quad
\langle\,,\,\rangle\colon(\Pi)V_{\bx}^{\dagger}\times V_{\bx}\to\Bbbk;
\end{equation}
we denote both operations in the same way because the order of arguments uniquely determines the choice. Let us remember
also that it is not the linear vector space fibres of the superbundle $\boldsymbol{\zeta}^{(0|1)}$ over the bundle $\pi$
of physical fields but it is the tangent spaces
$$T_{(\bx,\phi(\bx),s(\bx))}\left(V_{\bx}\mathbin{\widehat{\oplus}}\Pi V_{\bx}^{\dagger}\right)\cong
V_{\bx}\mathbin{\widehat{\oplus}}\Pi V_{\bx}^{\dagger}$$
to those fibres which harbour the variations $\delta\bolds=(\delta s;\delta s^{\dagger})$ of sections $s$ of the BV-bundle.

A reason to study the geometry of variations in tangent spaces to the fibres is clear from section~\ref{SecEL}. In fact,
although we have substantiated in section~\ref{SecBVzoo} that Euler\/--\/Lagrange equations and their descendants do form
linear vector spaces, this structure is incidental for the BV-formalism while Feynman path integration is not yet begun.
The guiding geometric principle is that linear vector spaces appear only in the course of inspection of functionals'
responses to infinitesimal test shifts of their arguments.

Couplings~\eqref{EqTwoCouplings} are defined only if the linear vector spaces $V_{\bx}\ni\delta s(\bx)$ and
$\Pi V_{\bx}^{\dagger}\ni\delta s^{\dagger}(\bx)$ are located over the same point $\bx\in M^n$ of the base manifold, and
over it they are attached as the two components of tangent space 
$T_{s(\bx)}\bigl(\boldsymbol{\zeta}^{(0|1)}\bigr)^{-1}\bigl(\bx,\phi(\bx)\bigr)$,
at the same point $s(\bx)=s\left(\bx,\phi(\bx)\right)$ of fibre in the superbundle $\boldsymbol{\zeta}^{(0|1)}$ over a point
$(\bx,\phi(\bx))$ of the total space for the bundle $\pi$ of physical fields (see Fig.~\ref{FigBVSetup}).

A distinction between the vector space $V_{\bx}$ and its parity-reversed dual nontrivially determines the couplings' values 
whenever they are defined. Namely, each of the two finite-dimensional vector spaces is reflexive,
\begin{equation}\label{EqTakeDual}
\left((V_{\bx})^{\dagger}\right)^{\dagger}\cong V_{\bx}\quad \text{and}\quad
\left((\Pi V_{\bx}^{\dagger})^{\dagger}\right)^{\dagger}\cong\Pi V_{\bx}^{\dagger},
\end{equation}
but these isomorphisms are not always identity mappings. 
We have that
\begin{equation}\label{EqChoiceSign}
\left\langle\vec{e}_i(\bx),\vec{e}^{{}\,\dagger j}(\bx)\right\rangle=\boldsymbol{\delta}^j_i\quad\text{yet}\quad
\left\langle\vec{e}^{{}\,\dagger j}(\bx),\vec{e}_i(\bx),\right\rangle=-\boldsymbol{\delta}^j_i,
\end{equation}
where $\boldsymbol{\delta}^j_i$ is the Kronecker symbol whose value is the unit iff $i=j$ and which is set equal to zero
otherwise.
\end{rem}

\begin{rem}
We claim that this mechanism is responsible, in particular, for the skew-symmetry of various Poisson brackets (e.\,g., of the
parity-odd Schouten bracket). Let us emphasize that this is a principle of order between geometric objects; the concept is
not restricted to the BV-setup which we study here. Actually, Eq.~\eqref{EqChoiceSign} is the fundamental reason for 
differential $1$-\/forms to anticommute%
\footnote{That is, this argument reveals why a mathematical axiom that differential forms do anticommute in the course of
calculations leads to verifiable and relevant theoretic predictions which match
experimental data.}
(in the class of geometries for which a coupling is defined between the linear vector spaces of co-multiples under the
wedge product $\wedge$; for instance, such is the case of the Helmholtz criterion 
$\psi=\delta S/\delta\bq$ $\Leftrightarrow$ $\vec{\ell}_{\psi}^{\,(\bq)}=\bigl(\vec{\ell}_{\psi}^{\,(\bq)}\bigr)^{\dagger}$
for images of the variational derivative~\cite{TwelveLectures,Olver}). Physically speaking, the binary count by ``a vector space,''
``not the former, hence its dual,'' and ``not the dual, but the initial space's image under central symmetry''
builds on the notion of order and realizes the law of the excluded middle.
\end{rem}

\subsection{Left\/-{} and right\/-\/variations via operators}\label{SecVariations}
Suppose that $$S=\int\cL(\bx,[\bq],[\bq^{\dagger}])\dvol(\bx)$$ is an integral functional $\Gamma(\pi_{\BV})\to\Bbbk$.
Let us focus on the correspondence between test shifts 
$\delta\bolds=(\delta s;\delta s^{\dagger})=\delta s^i\cdot\vec{e}_i+\delta s_i^{\dagger}\cdot\vec{e}^{{}\,\dagger i}$
of BV-fields $s\in\Gamma(\pi_{\BV})$ and, on the other hand, left- or right-acting linear singular integral operators
$\overleftarrow{\delta}\!\!\bolds$ and $\overrightarrow{\delta}\!\!\bolds$ which yield  the functional's responses to shifts
of its argument~$\bolds$. By definition, we put
\begin{subequations}\label{EqDefVariations}
\begin{multline}
\overrightarrow{\delta}\!\!\bolds=\int_M\Id\bby\,\Bigl\{
(\delta s^i)\left(\frac{\overleftarrow{\dd}}{\dd\bby}\right)^{\sigma}(\bby)\cdot
\left\langle(\vec{e}^{{}\,\dagger i})^{\dagger} (\bby),\vec{e}^{{}\,\dagger j}(\cdot)\right\rangle\frac{\overrightarrow{\dd}}
{\dd q^j_{\sigma}}+{}\\
{}+(\delta s^{\dagger}_i)\left(\frac{\overleftarrow{\dd}}{\dd\bby}\right)^{\sigma}(\bby)\cdot
\left\langle(\vec{e}_i)^{\dagger}(\bby),\vec{e}_j(\cdot)\right\rangle\frac{\overrightarrow{\dd}}
{\dd q^{\dagger}_{j,\sigma}}\Bigr\} 
\end{multline}
and
\begin{multline}
\overleftarrow{\delta}\!\!\bolds=\int_M\Id\bby\,\Bigl\{
\frac{\overleftarrow{\dd}}{\dd q^j_{\sigma}}
\left\langle\vec{e}^{{}\,\dagger j}(\cdot),{}^{\dagger}(\vec{e}^{{}\,\dagger i})(\bby)\right\rangle 
\left(\frac{\overrightarrow{\dd}}{\dd\bby}\right)^{\sigma}(\delta s^i)(\bby)+{}\\
{}+\frac{\overleftarrow{\dd}}{\dd q^{\dagger}_{j,\sigma}}
\left\langle\vec{e}_j(\cdot),{}^{\dagger}(\vec{e}_i)(\bby)\right\rangle\left(\frac{\overrightarrow{\dd}}
{\dd\bby}\right)^{\sigma}(\delta s^{\dagger}_i)(\bby)\Bigr\}. 
\end{multline}
\end{subequations}
The above formulas for directed operators $\overrightarrow{\delta}\!\!\bolds$ and $\overleftarrow{\delta}\!\!\bolds$ 
contain new notation $(\vec{e}_i)^{\dagger},\ (\vec{e}^{{}\,\dagger i})^{\dagger}$ and
${}^{\dagger}(\vec{e}_i),\ {}^{\dagger}(\vec{e}^{{}\,\dagger i})$,
also referring to an important sign convention which fully determines those adjoint objects. Namely, let us agree that
over every $\bx\in M^n$ the covectors
$$
\left.
\vec{e}^{{}\,\dagger i}(\bx)
\left(\frac{\overrightarrow{\dd}}{\dd q^i_{\sigma}}\cL(\bx,[\bq],[\bq^{\dagger}])\right)
\right|_{j^{\infty}_{\bx}(\bolds)}+
\left.\vec{e}_i(\bx)\left(\frac{\overrightarrow{\dd}}{\dd q^{\dagger}_{i,\sigma}}\cL(\bx,[\bq],[\bq^{\dagger}])\right)
\right|_{j^{\infty}_{\bx}(\bolds)}
$$
and
$$
\left.\left(\cL(\bx,[\bq],[\bq^{\dagger}])\frac{\overleftarrow{\dd}}{\dd q^i_{\sigma}}\right)
\right|_{j^{\infty}_{\bx}(\bolds)}\vec{e}^{{}\,\dagger i}(\bx)+
\left.\left(\cL(\bx,[\bq],[\bq^{\dagger}])\frac{\overleftarrow{\dd}}{\dd q^{\dagger}_{i,\sigma}}\right)
\right|_{j^{\infty}_{\bx}(\bolds)}\vec{e}_i(\bx)
$$
are expanded in the cotangent space 
$T^*_{\bolds(\bx)}W_{\bx}\cong V^{\dagger}_{\bx}\mathbin{\widehat{\oplus}}(TV^{\dagger}_{\bx})^{\dagger}$
with respect to the original basis $(+\vec{e}^{{}\,\dagger i},+\vec{e}_i)$; note the signs (any other convention here would
nohow alter the theory's content but it would (in)\/appropriately modify the signs in~\eqref{EqAlmostComplex} below).
The normalization of left- and right-adjoint objects $(\vec{e}_i)^{\dagger},\ (\vec{e}^{{}\,\dagger i})^{\dagger}$ and
${}^{\dagger}(\vec{e}_i),\ {}^{\dagger}(\vec{e}^{{}\,\dagger i})$ is immediate under assumption that the couplings'
equations yield~\eqref{EqStartSum} and then~\eqref{EqVariation} 
after integration by parts --- no extra sign factors appear 
in those formulas. This requirement determines the table
\begin{equation}\label{EqAlmostComplex}
\begin{aligned}
(\vec{e}^{{}\,\dagger i})^{\dagger}&=\phantom{+}\vec{e}_i,\\
(\vec{e}_i)^{\dagger}&=-\vec{e}^{{}\,\dagger i},
\end{aligned}\qquad
\begin{aligned}
{}^{\dagger}(\vec{e}^{{}\,\dagger i})&=-\vec{e}_i,\\
{}^{\dagger}(\vec{e}_i)&=\phantom{+}\vec{e}^{{}\,\dagger i},
\end{aligned}
\end{equation}
so that the following defining relations hold:
\begin{equation*}
\left\langle(\vec{e}_i)^{\dagger},\vec{e}_i\right\rangle=
\left\langle\vec{e}_i,{}^{\dagger}(\vec{e}_i)\right\rangle=
\left\langle(\vec{e}^{{}\,\dagger i})^{\dagger},\vec{e}^{{}\,\dagger i}\right\rangle=
\left\langle\vec{e}^{{}\,\dagger i},{}^{\dagger}(\vec{e}^{{}\,\dagger i})\right\rangle
=+1.
\end{equation*}
Let us notice 
that the left- and right-acting operation ${}^{\dagger}$ provides the analogue of left and right
$\langle\,,\,\rangle$-dual in this ordered world; the first column in~\eqref{EqAlmostComplex} determines a clockwise
rotation in the oriented plane spanned by $\vec{e}_i\prec\vec{e}^{{}\,\dagger i}$, whereas taking the adjoints
${}^{\dagger}(\cdot)\colon\vec{e}_i\mapsto\vec{e}^{{}\,\dagger i}$ and $\vec{e}^{{}\,\dagger i}\mapsto-\vec{e}_i$ induces
the counterclockwise rotation in that plane as shown in Fig.~\ref{FigMill}.
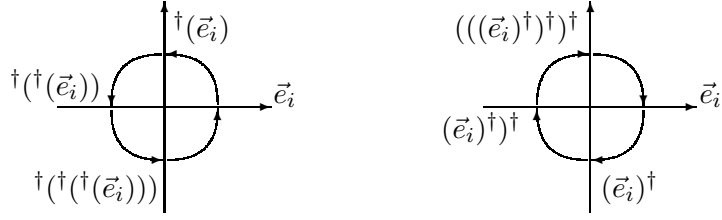
\begin{figure}[htb]
\begin{center}{\unitlength=0.7mm
\linethickness{0.4pt}
\begin{picture}(160.00,27.00)(0,18)
\put(40.00,10.00){\vector(0,1){40}}
\put(20.00,30.00){\vector(1,0){40.00}}
\bezier{76}(50.00,30.67)(50.00,40.00)(40.67,40.00)
\bezier{76}(30.00,29.33)(30.00,20.00)(39.33,20.00)
\bezier{76}(40.67,20.00)(50.00,20.00)(50.00,29.33)
\bezier{76}(39.33,40.00)(30.00,40.00)(30.00,30.67)
\put(30.00,32.00){\vector(0,-1){1.67}}
\put(38.00,20.00){\vector(1,0){1.67}}
\put(50.00,28.00){\vector(0,1){1.67}}
\put(42.00,40.00){\vector(-1,0){1.67}}
\put(41.67,42.00){\makebox(0,0)[lb]{${}^{\dagger}(\vec{e}_i)$}}
\put(60.5,30){\makebox(0,0)[lb]{$\vec{e}_i$}}
\put(39.00,11.67){\makebox(0,0)[rb]{${}^{\dagger}(^{\dagger}(^{\dagger}(\vec{e}_i)))$}}
\put(11.00,31.00){\makebox(0,0)[lb]{${}^{\dagger}(^{\dagger}(\vec{e}_i))$}}
\put(120.00,10.00){\vector(0,1){40.00}}
\put(100.00,30.00){\vector(1,0){40.00}}
\bezier{76}(130.00,30.67)(130.00,40.00)(120.67,40.00)
\bezier{76}(110.00,29.33)(110.00,20.00)(119.33,20.00)
\bezier{76}(120.67,20.00)(130.00,20.00)(130.00,29.33)
\bezier{76}(119.33,40.00)(110.00,40.00)(110.00,30.67)
\put(110.00,28.00){\vector(0,1){1.67}}
\put(122.00,20.00){\vector(-1,0){1.67}}
\put(130.00,32.00){\vector(0,-1){1.67}}
\put(118.00,40.00){\vector(1,0){1.67}}
\put(118,42.00){\makebox(0,0)[rb]{$(((\vec{e}_i)^{\dagger})^{\dagger})^{\dagger}$}}
\put(140.5,30){\makebox(0,0)[lb]{$\vec{e}_i$}}
\put(122,11.67){\makebox(0,0)[lb]{$(\vec{e}_i)^{\dagger}$}}
\put(92.00,29.00){\makebox(0,0)[lt]{$(\vec{e}_i)^{\dagger})^{\dagger}$}}
\end{picture}
}\end{center}
\caption{The orientation $\vec{e}_i\prec\vec{e}^{\,\dagger i}$ and configuration
of the left-{} and right-{} duals with respect to the couplings~$\langle\ ,\,\rangle$.}\label{FigMill}
\end{figure}

\begin{example}
Identities~\eqref{EqAlmostComplex} show up in the directed variations
$\left.\overleftarrow{\delta}\!\!S\right|_s^{\delta\bolds}=\overrightarrow{\delta}\!\!\bolds(S)(s)$ and
$\left.\overrightarrow{\delta}\!\!S\right|_s^{\delta\bolds}=(S)\overleftarrow{\delta}\!\!\bolds(s)$ 
of an integral functional $S=\int\cL(\bx,[\bq],[\bq^{\dagger}])\cdot\dvol(\bx)$. Namely, we have that
\begin{subequations}
\label{EqLeftRightVariations}
\begin{multline}
\left.\overleftarrow{\delta}\!\!S\right|_s^{(\delta s,\delta s^{\dagger})}=\\=
\int_M\Id\bby\int_M\dvol(\bx)\Biggl\{
(\delta s^i)\left(\frac{\overleftarrow{\dd}}{\dd\bby}\right)^{\sigma}(\bby)
\left\langle\vec{e}_i(\bby),\vec{e}^{{}\,\dagger j}(\bx)\right\rangle
\left.\left(\frac{\overrightarrow{\dd}}{\dd q^j_{\sigma}}
\cL(\bx,[\bq],[\bq^{\dagger}])\right)\right|_{j^{\infty}_{\bx}(s)}+\\+
(\delta s^{\dagger}_i)\left(\frac{\overleftarrow{\dd}}{\dd\bby}\right)^{\sigma}(\bby)
\left\langle-\vec{e}^{{}\,\dagger i}(\bby),\vec{e}_j(\bx)\right\rangle
\left.\left(
\frac{\overrightarrow{\dd}}{\dd q^{\dagger}_{j,\sigma}}\cL(\bx,[\bq],[\bq^{\dagger}])
\right)\right|_{j^{\infty}_{\bx}(s)}
\Biggr\}
\end{multline}
and
\begin{multline}
\left.\overrightarrow{\delta}\!\!S\right|_s^{(\delta s,\delta s^{\dagger})}=\\=
\int_M\Id\bby\int_M\dvol(\bx)\Biggl\{
\left.\left(
\cL(\bx,[\bq],[\bq^{\dagger}])\frac{\overleftarrow{\dd}}{\dd q^j_{\sigma}}
\right)
\right|_{j^{\infty}_{\bx}(s)}
\left\langle\vec{e}^{{}\,\dagger j}(\bx),-\vec{e}_i(\bby)\right\rangle
\left(\frac{\overrightarrow{\dd}}{\dd\bby}\right)^{\sigma}(\delta s^i)(\bby)+{}\\+
\left.\left(
\cL(\bx,[\bq],[\bq^{\dagger}])\frac{\overleftarrow{\dd}}{\dd q^{\dagger}_{j,\sigma}}
\right)\right|_{j^{\infty}_{\bx}(s)}
\left\langle\vec{e}_j(\bx),\vec{e}^{{}\,\dagger i}(\bby)\right\rangle
\left(\frac{\overrightarrow{\dd}}{\dd\bby}\right)^{\sigma}(\delta s^{\dagger}_i)(\bby)\Biggr\}.
\end{multline}
\end{subequations}
The operators $\overrightarrow{\delta}\!\!\bolds$ and $\overleftarrow{\delta}\!\!\bolds$ act via ghost-parity graded Leibniz' 
rule on formal products of integral functionals (and on their inages under other infinitesimal variation operators as well),
so that the two operators are defined on the entire space $\ov{\gN}^n(\pi_{\BV},T\pi_{\BV})$, see section~\ref{SecSchouten} 
below.
\end{example}

\begin{rem}
A reversion $\overleftarrow{\delta}\!\!\bolds\rightleftarrows\overrightarrow{\delta}\!\!\bolds$
of the direction along which such an operator acts means that the initially given operator (for definition, let it be
$\overleftarrow{\delta}\!\!\bolds$ which acts to the left) is \emph{destroyed} and in its place the other, 
opposite-direction operator is created (here it would be $\overrightarrow{\delta}\!\!\bolds$). Note that the variation 
$\delta\bolds\in\Gamma(T\pi)$ itself stays unchanged; it is the two realizations of this object via 
$\overleftarrow{\delta}\!\!\bolds$ and then via $\overrightarrow{\delta}\!\!\bolds$ which differ.
(This concept of test shifts as primary geometric objects which contain information about the operators will be essential
in Definition~\ref{DefSchouten} of the variational Schouten bracket.)
\end{rem}

\begin{rem}\label{RemRescaleBoth}
The postulate of duality between $\vec{e}_i(\bx)$ and $\vec{e}^{{}\,\dagger i}(\bx)$ correlates their transformation laws
under dilations: a rescaling $\vec{e}_i\mapsto\const\cdot\vec{e}_i$ with $\const\in\Bbbk\setminus\{0\}$ determines the
inverse-proportional mapping $\vec{e}^{{}\,\dagger i}\mapsto\const^{-1}\cdot\vec{e}^{{}\,\dagger i}$ of respective dual
vectors. (Likewise, the coordinates in $V_{\bx}$ and $\Pi V_{\bx}^{\dagger}$ are then rescaled by 
$q^i\mapsto\const^{-1}\cdot q^i$ and $q_i^{\dagger}\mapsto\const\cdot q_i^{\dagger}$. respectively.)
\end{rem}

Consider a variation $\delta\bolds=(\delta s;\delta s^{\dagger})\in\Gamma(T\boldsymbol{\zeta}^{(0|1)})$ of a BV-section
$s\in\Gamma(\boldsymbol{\zeta}^{(0|1)})$ over a given field configuration $\phi\in\Gamma(\pi)$ in the BV-bundle
$\boldsymbol{\pi}^{(0|1)}_{\BV}$.
The infinitesimal variation vectors $\delta\bolds=(\delta s;\delta s^{\dagger})$ can be naturally split to ghost 
parity-homogeneous components:
\begin{equation}\label{EqSplitVariations}
\delta\bolds=(\delta s;0)+(0;\delta s^{\dagger}).
\end{equation}
Here we explicitly use the linear vector space structure in fibres of the tangent bundle $T\boldsymbol{\zeta}^{(0|1)}$.
Let us recall 
that the two homogeneous variations
$$\delta s(\bx)=\delta s^i(\bx)\cdot\vec{e}_i(\bx)\quad\text{and}\quad
\delta s^{\dagger}(\bx)=\delta s^{\dagger}_i(\bx)\cdot\vec{e}^{{}\,\dagger i}(\bx)$$
in the right-hand side of~\eqref{EqSplitVariations} are the canonically dual to each other.

Moreover, by Remark~\ref{RemRescaleBoth} it is then possible to have $\delta s$ and $\delta s^{\dagger}$ normalized, for
every $i$ running from 1 to the dimension $N$, by the equalities
\begin{equation}\label{EqNormalize}
\delta s^i(\bx)\cdot\delta s_i^{\dagger}(\bx)\equiv+1
\end{equation}
at every $\bx\in M^n$ where the smooth fields of dual bases $\vec{e}_i$ and $\vec{e}^{\,\dagger i}$ are defined for the section
$s$. From now on, let us deal only with such \textsl{normalized} variations. This implies that the coupling of these
geometric objects are ``invisible'' but still the order in which the co-multiples $\delta s$ and $\delta s^{\dagger}$
occur in~\eqref{EqChoiceSign} does determine the signs in various formulas (e.\,g., in the definition of Schouten
bracket, see p.~\pageref{DefSchouten} below).

\subsection{Definitions of the BV-\/Laplacian and Schouten bracket}\label{SecDefinitions}
We now combine the geometry of graded-permutable iterated variations, which we explored in section~\ref{SecEL} and 
which absorbs a new copy of the underlying base manifold $M^n$ for each new infinitesimal test shift
$\delta\bolds(\bx)\in T_{s(\bx)}W_{\bx}$ of the functionals' arguments at $\bx\in M^n$, with the algebra of two
couplings~\eqref{EqTwoCouplings} between ghost parity-homogeneous halves of infinitesimal variations in the BV-setup
$T_{s(\bx)}W_{\bx}\cong V_{\bx}\mathbin{\widehat{\oplus}}\Pi V_{\bx}^{\dagger}$; the absolute locality of such coupling events is a
fundamental principle. 

To avoid an agglomeration of formulas and to match the notation with that in section~\ref{SecEL}, we omit an explicit
reference to field configuration $\{\phi(\bx),\,\bx\in M^n\}$, indicating only the base points $\bx\in M^n$. We also
denote by $\pi_{\BV}$ the composite-structure superbundle over $M^n$ (see Fig.~\ref{FigBVSetup}) so that the notation for 
the vector bundle of BV-sections' infinitesimal variations is $T\pi_{\BV}$. However, let us remember that only the linear
BV-fibre variables $(\bq,\bq^{\dagger})$ but not the physical fields $\phi$ are subjected to variations at points
$s(\bx)\in\bigl(\boldsymbol{\zeta}^{(0|1)}\bigr)^{-1}(\bx,s(\bx))$ over $(\bx,\phi(\bx))\in\pi^{-1}(\bx)$. 
A brute force labelling of Euler\/--\/Lagrange equations by the respective unknowns is an act of will by the one who writes
formulas but it is not a prescription from the model's geometry.

This section contains rigorous, self-regularizing definitions of the BV-Laplacian and Schouten
bracket for integral functionals from $\ov{H}^n(\pi_{\BV})\subsetneq\ov{\gM}^n(\pi_{\BV})\subsetneq
\ov{\gN}^n(\pi_{\BV},T\pi_{\BV})$.
We shall extend the definition to the space $\ov{\gN}^n(\pi_{\BV},T\pi_{\BV})$ of products of integral functionals, possibly with
earlier-absorbed variations, in the subsequent sections of this paper. We then establish the main properties of these
structures and prove relations between them.
We note that the definitions which we give here are operational: each of them is a surgery for the couplings and their
reconfiguration algorithm. (The locality postulate ensures the restrictions onto diagonals in the product 
$M\times\ldots\times M$ so that those recombinations make sense at every point of $M$.)
\enlargethispage{0.7\baselineskip}

\subsubsection{The BV-\/Laplacian $\Delta$}
Let us first introduce some shorthand notation. Let $F=\int f(\bx,[\bq],[\bq^{\dagger}])\cdot\dvol(\bx)$
be an integral functional and $\delta\bolds=(\delta s;0)+(0;\delta s^{\dagger})$ be a variation's splitting in two ghost
parity-homogeneous variations. From section~\ref{SecEL} we know that each of the two is referred to its own copy of the
base: let it be $\delta s(\bby_1)$ and $\delta s^{\dagger}(\bby_2)$ so that formula~\eqref{EqStartSum} defines the response of
$F$ to an infinitesimal shift of its argument along each of the two directions. 

\begin{define}\label{DefBV}
Let $\delta\bolds\in\Gamma(T\pi_{\BV})$ be a test shift normalized by~\eqref{EqNormalize} and then split to the sum
$(\delta s;0)+(0;\delta s^{\dagger})$ of ghost parity-homogeneous, $\langle\,,\,\rangle$-dual halves. 
The BV-\emph{Laplacian} is the linear operator $\Delta\colon\ov{H}^n(\pi_{\BV})\to\ov{\gN}^n(\pi_{\BV},T\pi_{\BV})$;
for a ghost parity\/-\/homogeneous integral functional $F\in\ov{H}^n(\pi)$ and its argument $s$, the operator $\Delta$
is an algorithm for reconfiguration of couplings in the second variation
\begin{multline*}
\left.\frac{\Id}{\Id\veps}\right|_{\veps=0}\left.\frac{\Id}{\Id\veps^{\dagger}}\right|_{\veps^{\dagger}=0}
F(s+\veps\cdot\overleftarrow{\delta}\!\!s+\veps^{\dagger}\cdot\overleftarrow{\delta}\!\!s^{\dagger})=
\sum_{\substack{i_1,i_2\\j_1,j_2}}\sum_{\substack{|\sigma_1|\ge0\\|\sigma_2|\ge0}}
\int_M\Id\bby_1\int_M\Id\bby_2\int_M\dvol(\bx)\\
\left\{
\begin{matrix}\phantom{\hookrightarrow}
(\delta s^{i_1})\left(\frac{\overleftarrow{\dd}}{\dd\bby_1}\right)^{\sigma_1}(\bby_1)\,
\langle\phantom{+}\vec{e}_{i_1}(\bby_1),\vec{e}^{{}\,\dagger j_1}(\bx)\rangle\hookleftarrow\\
\hookrightarrow(\delta s^{\dagger}_{i_2})\left(\frac{\overleftarrow{\dd}}{\dd\bby_2}\right)^{\sigma_2}(\bby_2)
\langle-\vec{e}^{{}\,\dagger i_2}(\bby_2),\vec{e}_{j_2}(\bx)\rangle\phantom{\hookleftarrow}
\end{matrix}
\right\}\cdot
\left.\frac{\overrightarrow{\dd^2}f(\bx,[\bq],[\bq^\dagger])}{\dd q^{j_1}_{\sigma_1}
\dd q^{\dagger}_{j_2,\sigma_2}}\right|_{j^\infty_{\bx}(s)}\,.
\end{multline*}
This second variation's integrand contains the couplings $\langle\,,\,\rangle$:
$$T_{s(\bby_1)}V_{\bby_1}\times T^*_{s(\bx)}(\Pi)V_{\bx}\to\Bbbk\quad\text{ and }\quad
T_{s(\bby_2)}\Pi V_{\bby_2}^{\dagger}\times T^*_{s(\bx)}(\Pi)V_{\bx}^{\dagger}\to\Bbbk
$$
which are defined only if the attachment points coincide for these (co)vectors;
an optional presence of the parity reversion operator indicates a possibility of having ghost parity-odd functional $F$.

At the moment when the object~$\Delta F$ under construction --\,or a larger object of 
which~$\Delta F$ is an element, see~\eqref{EqZimes}\,-- is evaluated at 
a section~$s\in\Gamma(\pi_\BV)$, the integrations by parts carry the derivatives away 
from the variations' components: 
$\overleftarrow{\dd}/\dd\bby_i\mapsto\overrightarrow{\dd}/\dd\bby_i$ as explained in 
section~\ref{SecByParts}.
The third step in definition of~$\Delta$ acting on~$F$ is a surgery algorithm
for an on\/-\/the\/-\/diagonal reattachment of the couplings, see Figure~\ref{Fig1234}.
\begin{figure}[htb]
\[
\begin{array}{rrll}
\langle\,{}^1\mars\,| &  & |\,{}^3\venus\,\rangle & { }\\
{ } & \langle\,{}^2\venus\,| & & |\,{}^4\mars\,\rangle
\end{array}
\qquad\longmapsto\qquad
\begin{array}{rlrl}
\langle\,{}^1\mars\,| &  & \langle\,{}^3\venus\,| & { }\\
{ } & |\,{}^2\venus\,\rangle & & |\,{}^4\mars\,\rangle
\end{array}
\]
\caption{The on\/-\/the\/-\/diagonal coupling of variations versus taking the trace of
bi\/-\/linear form.}\label{Fig1234}
\end{figure}
In other words, \emph{after} the integration by parts 
the surgery 
yields the following:
\begin{multline}\label{EqDefBV}
(\Delta F)\Bigr|_s^{\delta\bolds} = 
\sum_{\substack{i_1,i_2\\j_1,j_2}}\sum_{\substack{|\sigma_1|\ge0\\|\sigma_2|\ge0}}
\int_M\Id\bby_1\int_M\Id\bby_2\int_M\dvol(\bx)\cdot\\
\cdot\left\{
\delta s^{i_1}(\bby_1)\left(-\frac{\overrightarrow{\dd}}{\dd\bby_1}\right)^{\sigma_1} 
\underbrace{\langle\vec{e}_{i_1}(\bby_1),-\vec{e}^{{}\,\dagger i_2}(\bby_2)\rangle}_{-1}\cdot
\delta s^{\dagger}_{i_2}(\bby_2)\left(-\frac{\overrightarrow{\dd}}{\dd\bby_2}\right)^{\sigma_2}\right\}
\cdot\\
\cdot\left\{
\underbrace{\langle \vec{e}^{{}\,\dagger j_1}(\bx), \vec{e}_{j_2}(\bx)\rangle}_{-1}\cdot
\left.\frac{\vec{\dd^2}f(\bx,[\bq],[\bq^\dagger])}{\dd q^{j_1}_{\sigma_1}
\dd q^{\dagger}_{j_2,\sigma_2}}\right|_{j^\infty_{\bx}(s)}
\right\}\,.
\end{multline}
Note that the left-to-right 
order in $\left\langle\vec{e}_{i_1}(\bby_1),\vec{e}^{{}\,\dagger j_1}(\bx)\right\rangle\cdot \left\langle-\vec{e}^{{}\,\dagger i_2}(\bby_2),\vec{e}_{j_2}(\bx)\right\rangle$
is preserved by the respective couplings' arguments in
$\langle\vec{e}_{i_1}(\bby_1),-\vec{e}^{{}\,\dagger i_2}(\bby_2)\rangle\cdot
\langle \vec{e}^{{}\,\dagger j_1}(\bx), \vec{e}_{j_2}(\bx)\rangle$, cf.\ Fig.~\ref{Fig1234}.
\end{define}

\begin{rem}
Until the moment when the integrations by parts are performed in $\Delta F$, the derivatives $\dd/\dd\bby_1$ and
$\dd/\dd\bby_2$ refer to different copies of the manifold $M^n$ in the base $M^n\times M^n\times M^n$ of the product bundle
$\pi_{\BV}\times T\pi_{\BV}\times T\pi_{\BV}$. This implies that the two variations of $F$ in the definition of $\Delta$
are graded\/-\/permutable between each other and with 
all other variations falling on $f(\bx,[\bq],[\bq^{\dagger}])$ whenever
$\Delta F$ is a constituent element of a larger object (e.\,g., see~(\ref{EqZimes}--\ref{EqDeltaSquareIntro}) 
on p.~\pageref{EqZimes}).
\end{rem}

\begin{rem}\label{RemKeepNotation}
To keep track of multiple copies of the
base $M^n$ for functionals and variations (here $\bx\in M^n,\ \bby_1\in M^n,\ \bby_2\in M^n$) in the course of integration by
parts (see section~\ref{SecByParts}), we indicate the respective variations' bases by explicitly writing
$\bq(\bby_1)$ and $\bq^{\dagger}(\bby_2)$ in the denominators \emph{and} we denote by
$\dd/\dd\bby_1$ and $\dd/\dd\bby_2$ the derivatives which now fall on the functional's density 
$f(\bx,[\bq],[\bq^{\dagger}])$~--- for instance, we do so in Example~\ref{Countercounterexample} on
p.~\pageref{Countercounterexample} below. Namely, we put
\begin{subequations}\label{EqHomogVariations}
\begin{align}
\left.\frac{\overleftarrow{\delta}\!\!{f}(\bx,[\bq],[\bq^\dagger])}{\delta q^\alpha(\bby_1)}\right|_{j^\infty_{\bx}(s)} &=
\sum\limits_{|\sigma_1|\geqslant0}\Bigl(-\frac{\vec{\dd}}{\dd\bby_1}\Bigr)^{\sigma_1}
\left.\left(\frac{\vec{\dd}f(x,[\bq],[\bq^\dagger])}{\dd q^\alpha_{\sigma_1}}
\right)\right|_{j^\infty_{\bx}(s)}=\\&=\notag
\sum_{|\sigma_1|\ge0}\left.\left(\left(-\frac{\vec{\Id}}{\Id\bby_1}\right)^{\sigma_1}
\frac{\vec{\dd}f(\bx,[\bq],[\bq^{\dagger}]}{\dd q^{\alpha}_{\sigma_1}}\right)\right|_{j^{\infty}_{\bx}(s)}
\\
\intertext{and} 
\left.\frac{\overleftarrow{\delta}\!\!{f}(\bx,[\bq],[\bq^\dagger])}{\delta q^\dagger_\beta(\bby_2)}\right|_{j^\infty_{\bx}(s)} &=
\sum\limits_{|\sigma_2|\geqslant0}\Bigl(-\frac{\vec{\dd}}{\dd\bby_2}\Bigr)^{\sigma_2}
\left.\left(\frac{\vec{\dd}f(\bx,[\bq],[\bq^\dagger])}{\dd q^\dagger_{\beta,\sigma_2}}
\right)\right|_{j^\infty_{\bx}(s)}=\\&=\notag
\sum_{|\sigma_2|\ge0}\left.\left(\left(-\frac{\vec{\Id}}{\Id\bby_2}\right)^{\sigma_2}
\frac{\vec{\dd}f(\bx,[\bq],[\bq^{\dagger}]}{\dd q^{\dagger}_{\beta,\sigma_2}}\right)\right|_{j^{\infty}_{\bx}(s)}
\end{align}
\end{subequations}
for the ghost parity-homogeneous components of variational derivative. At every point $(\bx,\phi(\bx),s(\bx))$ of the total
space for the bundle $\pi_{\BV}$, and for a given functional $F$ which is assumed ghost parity-homogeneous, we have that
\[
\left.\frac{\overleftarrow{\delta}\!\!{f}(\bx,[\bq],[\bq^\dagger])}{\delta q^\alpha(\bby_1)}\right|_{j^\infty_{\bx}(s)}
 \in T^*_{s(\bx)}(\Pi)V_{\bx}
   \qquad\text{and}\qquad 
\left.\frac{\overleftarrow{\delta}\!\!{f}(\bx,[\bq],[\bq^\dagger])}{\delta q^\dagger_\beta(\bby_2)}\right|_{j^\infty_{\bx}(s)}
 \in T^*_{s(\bx)}(\Pi)V^\dagger_{\bx}.
\]
Let us remember 
that an attribution of denominators to $\bby_1$ or $\bby_2$ is a matter of notation in~\eqref{EqHomogVariations}; 
whenever happening, everything happens at $\bx\in M^n$.
\end{rem}

\begin{lemma}\label{LAnyChoice}
The BV-Laplacian $\Delta$ is independent of a choice of the variation $\delta\bolds$ normalized by ~\eqref{EqNormalize}.
\end{lemma}

Indeed, whenever the integrations by parts are performed, products~\eqref{EqNormalize} of the dual components are always
the same at all points 
of the intersection of their domains of definition.%
\footnote{The assertion of Lemma~\ref{LAnyChoice} extends to the variational Schouten bracket, which is a derivative
structure with respect to the BV-Laplacian (see Definition~\ref{DefSchouten} on p.~\pageref{DefSchouten}). Moreover, the independence
of a specific choice of variations implies that their coefficients $(\delta s_1,\delta s_1^{\dagger})$ and 
$(\delta s_2,\delta s_2^{\dagger})$, which are built into $\Delta$ and $\lshad\,,\,\rshad$, can be swapped, not altering 
an object that contains these test shifts $\delta\bolds_1$ and $\delta\bolds_2$ (see the proof of Lemma~\ref{LemmaBaseLapSchouten}
on p.~\pageref{LemmaBaseLapSchouten}).
}
We illustrate the definition of BV-\/Laplacian~$\Delta$ by using Fig.~\ref{FigPants};
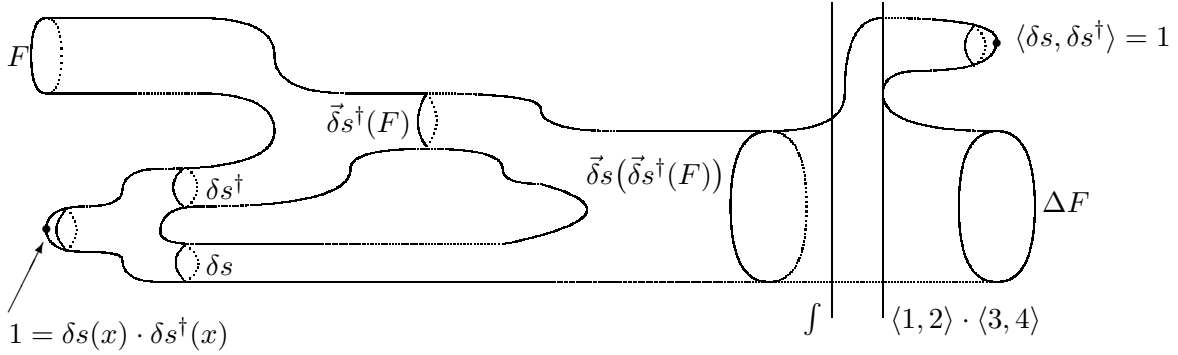
\begin{figure}[htb]
\begin{center}{\unitlength=1mm
\linethickness{0.4pt}
\begin{picture}(136.00,43.00)(10,0)
\bezier{40}(5.00,12.00)(5.00,15.00)(10.00,15.00)
\bezier{40}(5.00,12.00)(5.00,9.00)(10.00,9.00)
\bezier{40}(10.00,9.00)(15.00,9.00)(15.00,8.00)
\bezier{28}(10.00,15.00)(15.00,15.00)(15.00,17.33)
\bezier{32}(15.00,17.00)(15.00,20.00)(20.00,20.00)
\bezier{80}(20.00,20.00)(35.00,20.00)(35.00,25.00)
\bezier{60}(35.00,25.00)(35.00,30.00)(25.00,30.00)
\bezier{80}(25.00,30.00)(7.33,30.00)(5.00,30.00)
\bezier{10}(5.00,30.00)(7.00,30.00)(7.00,35.00)
\bezier{10}(7.00,35.00)(7.00,40.00)(5.00,40.00)
\bezier{38}(5.00,40.00)(3.00,40.00)(3.00,35.00)
\bezier{38}(3.00,35.00)(3.00,30.00)(5.00,30.00)
\bezier{80}(5.00,40.00)(15.67,40.00)(25.00,40.00)
\bezier{60}(25.00,40.00)(35.00,40.00)(35.00,35.00)
\bezier{60}(35.00,35.00)(35.00,30.00)(45.00,30.00)
\bezier{40}(45.00,30.00)(49.33,30.00)(55.00,30.00)
\bezier{32}(55.00,30.00)(52.67,26.00)(55.00,23.00)
\bezier{12}(55.00,30.00)(57.67,26.33)(55.00,23.00) 
\bezier{68}(55.00,30.00)(70.00,30.00)(70.00,28.00)
\bezier{52}(70.00,28.00)(70.00,25.00)(80.00,25.00)
\bezier{80}(80.00,25.00)(87.67,25.00)(100.00,25.00)
\bezier{100}(100.00,25.00)(95.00,25.00)(95.00,15.00)
\bezier{30}(100.00,25.00)(105.00,25.00)(105.00,15.00)
\bezier{100}(100.00,5.00)(95.00,5.00)(95.00,15.00)
\bezier{30}(100.00,5.00)(105.00,5.00)(105.00,15.00)
\bezier{60}(100.00,25.00)(110.00,25.00)(110.00,30.00)
\bezier{60}(110.00,30.00)(110.00,40.00)(115.00,40.00)
\bezier{72}(115.00,40.00)(130.00,40.00)(130.00,37.00)
\bezier{56}(130.00,37.00)(130.00,33.00)(120.00,33.00)
\bezier{32}(120.00,33.00)(115.00,33.00)(115.00,30.00)
\bezier{80}(115.00,30.00)(115.00,25.00)(130.00,25.00)
\bezier{100}(125.00,15.00)(125.00,25.00)(130.00,25.00)
\bezier{100}(130.00,25.00)(135.00,25.00)(135.00,15.00)
\bezier{100}(135.00,15.00)(135.00,5.00)(130.00,5.00)
\bezier{100}(130.00,5.00)(125.00,5.00)(125.00,15.00)
\bezier{440}(130.00,5.00)(32.00,5.00)(20.00,5.00)
\bezier{32}(20.00,5.00)(15.00,5.00)(15.00,8.00)
\put(5.00,12){\circle*{1}}
\bezier{32}(20.00,11.67)(20.00,15.00)(25.00,15.00)
\bezier{100}(25.00,15.00)(45.00,15.00)(45.00,20.00)
\bezier{52}(45.00,20.00)(45.00,22.67)(55.00,22.67)
\bezier{52}(55.00,22.67)(65.00,22.67)(65.00,20.00)
\bezier{28}(65.00,20.00)(65.00,18.00)(70.00,18.00)
\bezier{140}(70.00,18.00)(84.33,14.00)(65.00,10.00)
\bezier{160}(65.00,10.00)(27.33,10.00)(25.00,10.00)
\bezier{28}(25.00,10.00)(20.00,10.00)(20.00,11.67)
\bezier{28}(23.00,15.00)(20.33,17.67)(23.00,20.00)
\bezier{10}(23.67,20.00)(26.00,17.67)(23.67,15.00) 
\bezier{28}(23.33,10.00)(20.67,7.67)(23.33,5.00)
\bezier{10}(23.33,10.00)(26.67,7.33)(23.33,5.00) 
\bezier{28}(7.67,14.67)(5.00,11.92)(7.67,9.17)
\bezier{10}(7.67,14.67)(10.33,11.92)(7.67,9.17)
\bezier{28}(127,33.67)(124.67,36.33)(127,39)
\bezier{10}(127,33.67)(130,36.33)(127,39)
\put(108.33,0.33){\line(0,1){41.67}}
\put(115.00,42.00){\line(0,-1){41.67}}
\put(130.00,36.67){\circle*{1}}
\put(-0.2,33.67){\makebox(0,0)[lb]{$F$}}
\put(53,24.00){\makebox(0,0)[rb]{$\vec{\delta}s^{\dagger}(F)$}}
\put(94.33,16.67){\makebox(0,0)[rb]{$\vec{\delta}s\bigl(\vec{\delta}s^{\dagger}(F)\bigr)$}}
\put(104.5,-2){\makebox(0,0)[lb]{$\int$}}
\put(116.00,-2){\makebox(0,0)[lb]{$\langle1,2\rangle\cdot\langle3,4\rangle$}}
\put(136.00,14.33){\makebox(0,0)[lb]{$\Delta F$}}
\put(132.33,35.33){\makebox(0,0)[lb]{$\langle\delta s,\delta s^{\dagger}\rangle=1$}}
\put(26,16){\makebox(0,0)[lb]{$\delta s^{\dagger}$}}
\put(26,6){\makebox(0,0)[lb]{$\delta s$}}
\put(0.00,-4.00){\makebox(0,0)[lb]{$1=\delta s(x)\cdot\delta s^{\dagger}(x)$}}
\put(0,1){\vector(1,2){4.6}}
\end{picture}
}\end{center}
\caption{A variational update of the cyclic wor(l)d from~\cite{KontsevichCyclic}:
the (anti)\/words $\delta s$ and~$\delta s^\dagger$ are pasted into a necklace~$F$
according to the graded Leibniz rule. Then they annihilate in such a way that
the respective loose ends of the string join, the cyclic order of gems preserved;
this yields~$\Delta F$.}\label{FigPants}
\end{figure}
let us notice that it properly renders the assertion of Lemma~\ref{LAnyChoice} in a wider,
noncommutative setup of~\cite{KontsevichCyclic} and~\cite{SQS11,Lorentz12} 
(see Remark~\ref{RemStartNonGraded} on p.~\pageref{RemStartNonGraded}).

\begin{cor}\label{CorTowardsConventional}
In particular, we obtain the equality for immediate numeric value of $\Delta F$ at $s$.
Namely, we have that 
\begin{multline*}
(\Delta F)(s)=\sum_{i_1,i_2} \sum_{\substack{|\sigma_1|\geq0\\ |\sigma_2|\geq0}} \int_M\dvol(\bx)\,\\
\left.\left\{\delta s^{i_1}(\bby_1)\cdot\boldsymbol{\delta}^{i_2}_{i_1}
\cdot\delta s^{\dagger}_{i_2}(\bby_2)\cdot
\left. \left(\left(-\frac{\vec{\Id}}{\Id\bx}\right)^{\sigma_1\cup\sigma_2}
\frac{\vec{\dd}^2 f(\bx,[\bq],[\bq^{\dagger}])}{\dd q^{i_1}_{\sigma_1}  
\dd q^{\dagger}_{i_2,\sigma_2}  
}\right)
\right|_{j^{\infty}_{\bx}(s)}\right\}\right|_{\substack{\bby_1\,=\,\bx\\\bby_2\,=\,\bx}}\in\Bbbk.
\end{multline*}
By taking one sum containing Kronecker's $\boldsymbol{\delta}$-\/symbol, one arrives at a conventional formula with a
summation over the diagonal:
\begin{multline}\label{EqEvalBVAtSection}
(\Delta F)(s)=\sum_{i=1}^N\sum_{\substack{|\sigma_1|\ge0\\|\sigma_2|\ge0}}
\int_M\dvol(\bx)\,
\left.
\left(
\left(-\frac{\vec{\Id}}{\Id\bx}\right)^{\sigma_1\cup\sigma_2}
\frac{\vec{\dd}^2 f(\bx,[\bq],[\bq^\dagger])}{\dd q^i_{\sigma_1}\dd q^{\dagger}_{i,\sigma_2}}
\right)
\right|_{j^{\infty}_{\bx(s)}}
\mathrel{\stackrel{\text{def}}{=}}{}
\\
{}\mathrel{\stackrel{\text{def}}{=}}\sum_{i=1}^N\int_M\dvol(\bx)\,
\frac{\overleftarrow{\delta^2}f(\bx,[\bq],[\bq^\dagger])}{\delta q^i\delta q^{\dagger}_i}\,.
\end{multline}
We refer to footnote~\ref{FootMKEW} on p.~\pageref{FootMKEW} in this context.
\end{cor}

\begin{rem}
The conventional formula
$$\left.\frac{\overleftarrow{\delta^2}f(\bx,[\bq],[\bq^\dagger])}{\delta\bq(\bby_1)\delta\bq^{\dagger}(\bby_2)}
\right|_{\substack{\bby_1\,=\,\bx\\\bby_2\,=\,\bx}}
$$
itself is not the definition of a density of the BV-Laplacian $\Delta F$ for an integral functional
$F=\int f(\bx,[\bq],[\bq^{\dagger}])\cdot\dvol(\bx)$.
Not containing any built-in sources of divergence, the geometric definition and its implication~\eqref{EqEvalBVAtSection}
yield identical results only when one calculates the numeric value $(\Delta F)(s)\in\Bbbk$~--- but not earlier:
structurally different objects~\eqref{EqDefBV} and~\eqref{EqEvalBVAtSection} belong to non\/-\/isomorphic spaces
(so that the former contains more information then the latter), and their analytic behaviour is also different,
see Example~\ref{Countercounterexample} on p.~\pageref{Countercounterexample}.
\end{rem}

The following two examples are quoted from~\cite{Laplace13}; they show that the structure~$\Delta$ defined above coincides --\,but only in the simplest situation\/-- with the one which is in\-tu\-i\-ti\-ve\-ly known from the literature. We refer to the main Example~\ref{Countercounterexample} on p.~\pageref{Countercounterexample} which illustrates the multiple\/-\/base geometry in a logically more complex situation of~\eqref{EqZimes}.

\begin{example}\label{ex:YM}
Take a compact, semisimple Lie group~$G$ with Lie algebra~$\mathfrak{g}$ and consider the corresponding Yang\/--\/Mills theory.
Write $A^a_i$ for the (coordinate expression of) the gauge potential $A$~-- a lower index $i$ because $A$ is a one\/-\/form on
the base manifold (i.\,e., a covector), and an upper index $a$ because $A$ is a vector in the Lie algebra $\mathfrak{g}$ of
the Lie group $G$. Defining the field strength $\mathcal{F}$ by 
$\mathcal{F}^a_{ij} = \partial_i A^a_j - \partial_j A^a_i + f^a_{bc}A^b_i A^c_j$ where $f^a_{bc}$ are the structure
constants of the Lie algebra~$\mathfrak{g}$, the Yang\/--\/Mills action is\footnote{The action functional~$S_{\text{YM}}$ is referred to Minkowski flat coordinates such that $\dvol(\bx)=\sqrt{|-1|}\,\Id^4 x$ in the weak gauge field limit.}
\[
S_{\text{YM}} = \tfrac14\int\mathcal{F}^a_{ij}\mathcal{F}^{a,ij}\,\Id^4 x,
\]
and the full BV-\/action $S_\BV$ is\footnote{We denote by $A_a^{i\dagger}$ the 
parity\/-\/odd \textsl{antifields}, by $\gamma^a$ the odd \textsl{ghosts}, and
by $\gamma^\dagger_a$ the parity\/-\/even \textsl{antighosts}.}
\[
S_\BV = S_{\text{YM}}
+ \int A_a^{i\dagger}(\tfrac{\Id}{\Id x^i}
\gamma^a + f_{bc}^a A_i^b\gamma^c) \,\Id^4x
- \tfrac12\int f_{ab}^c\gamma^a\gamma^b\gamma^\dagger_c\,\Id^4x.
\]
Let us calculate the BV-Laplacian of this functional. By Corollary~\ref{CorTowardsConventional}, the
only terms which survive in $\Delta(S_{\text{BV}})$ are those which contain both $A$ and $A^\dagger$, or both $\gamma$ and~$\gamma^\dagger$. 
Therefore,   
\begin{align*}
\Delta(S_{\text{BV}})
 &= \int\left(
    \frac{\overleftarrow{\delta}}{\delta A_j^d}\frac{\overleftarrow{\delta}}{\delta A^{j\dagger}_d}(f_{bc}^a A_a^{i\dagger}A_i^b\gamma^c)
   - \frac12\frac{\overleftarrow{\delta}}{\delta\gamma^\dagger_d}\frac{\overleftarrow{\delta}}{\delta\gamma^d}(f^c_{ab}\gamma^a\gamma^b\gamma^\dagger_c)
   \right)\Id^4x \\
 &= \int\left(
   \frac{\overleftarrow{\delta}}{\delta A_j^d}(f_{bc}^d A_j^b\gamma^c)
   - \frac12\frac{\overleftarrow{\delta}}{\delta\gamma^\dagger_d}(f^c_{db}\gamma^b\gamma^\dagger_c - f^c_{ad}\gamma^a\gamma^\dagger_c)
   \right)\Id^4x \\
 &= \int\left(
   f_{dc}^d\gamma^c - \tfrac12\bigl(f^d_{db}\gamma^b - f^d_{ad}\gamma^a\bigr)
   \right)\Id^4x 
  = 0.
\end{align*}
Let us note also that,
since the BV-action $S_{\text{BV}}$ is by construction such that the horizontal cohomology class of $\lshad{S_\text{BV},S_\text{BV}}\rshad$ is zero,
as one easily checks by using Definition~\ref{DefSchouten} below,
the functional $S_\text{BV}$ satisfies 
quantum master\/-\/equation~\eqref{QME} tautologically: both sides 
are, by independent calculations, equal to zero~--- should one inspect those values at any section~$s$ of the BV-\/bundle.
\end{example}

\begin{example}\label{ex:CF}
Consider the nonlinear Poisson sigma model introduced in~\cite{CattaneoFelderCMP2000}. Since its fields are not all purely even, we 
have to generalize all of our reasoning so far to a $\BBZ_2$-graded setup~--- which is, as noted in Remark~\ref{RemBiGraded}, tedious but
straightforward. A verification that $\Delta(S_\text{CF})(s)=0$ for the 
BV-\/action $S_\text{CF}$ of this model and a section~$s$ of the respective BV-\/bundle would, up to minor
differences in conventions and notations, proceed just as it does in that paper itself, in section 3.2 thereof~--- except that no
infinite constants or Dirac's $\boldsymbol{\delta}$-\/function appear.
\end{example}

\begin{rem}
The BV-Laplacian $\Delta$ is extended by using Leibniz' rule from the space $\ov{H}^n(\pi_{\BV})$ of building blocks in 
$\ov{\gM}^n(\pi_{\BV})$ to the space $\ov{\gN}^n(\pi_{\BV},T\pi_{\BV})$, see Theorem~\ref{ThLaplaceOnProduct}
on p.~\pageref{ThLaplaceOnProduct}. The couplings' (re)attachment algorithm then
results in formula~\eqref{EqDeviationDerivationIntro} on p.~\pageref{EqAllIntro}, which is 
taken as a \emph{definition} of the variational Schouten bracket $\lshad\,,\,\rshad$, see~\cite{YKS2008SIGMA}. 
In turn, 
that structure's extention from 
$\ov{H}^n(\pi_{\BV})\times\ov{H}^n(\pi_{\BV})$ to $\ov{\gN}^n(\pi_{\BV},T\pi_{\BV})\times\ov{\gN}^n(\pi_{\BV},T\pi_{\BV})$
is immediate (see Theorem~\ref{ThSchoutenOnProduct} below).

The correspondence between $\Delta$ and $\lshad\,,\,\rshad$ is furthered to an equivalence between the property
$\Delta^2=0$ of BV-Laplacian to be a differential and, on the other hand, Jacobi's identity for the variational
Schouten bracket. We emphasize that the latter \textsl{can be} verified within the old approach~\cite{KuperCotangent}
to geometry of variations. (We refer to~\cite{Lorentz12} for a proof; its crucial idea is that with evolutionary vector
fields it does not matter under ``whose'' total derivatives, $\Id/\Id\bx$ or $\Id/\Id\bby_i$, such fields dive.)
Nevertheless, the traditional paradigm fails to reveal that the operator $\Delta$ is a differential because of a necessity
to have the variations graded-permutable and for that, to distinguish between the functionals' and variations' domains
of definition. Our geometric approach resolves that obstruction and ensures the validity of identities~\eqref{EqZimes}
and~\eqref{EqDeltaSquareIntro} (see Theorems~\ref{ThLaplaceOnSchouten} and~\ref{ThBVDifferential}, respectively).
\end{rem}

\subsubsection{The variational Schouten bracket $\lshad\,,\,\rshad$}\label{SecSchouten}
The parity-odd Laplacian $\Delta$ is the parent object%
\footnote{In particular, the definition of BV-Laplacian logically precedes the construction of Schouten bracket in
BV-formalism (although such parity-odd variational Poisson bracket is often introduced through postulated 
formula~\eqref{EqFormulaSchouten} in the context of Hamiltonian dynamics and infinite\/-\/dimensional completely integrable
systems~\cite{Dorfman,GelfandDorfman,Lstar,KuperCotangent,Magri}). Indeed, the entire Schouten\/-\/bracket machinery of
(quantum) BV-\/cohomology groups and their automorphisms, which we consider in secction~\ref{SecGauge}, stems from quantum
master-equation~\eqref{QME}, see p.~\pageref{QME}.
}
which induces the variational Schouten bracket $\lshad\,,\,\rshad$. Namely, the bracket appears in the course of that
operator's extension from the space $\ov{H}^n(\pi_{BV})\ni F$ to the space
$\ov{\gN}^n(\pi_{\BV},T\pi_{\BV})\supseteq\ov{\gM}^n(\pi_{\BV})$ of local functionals $F_1\cdot\ldots\cdot F_{\ell}$
(it is possible that $F_i$'s already contain some normalized variations).

A distinction between \emph{left} and \emph{right} in the directed operators $\overleftarrow{\delta}\!\!\bolds$ and
$\overrightarrow{\delta}\!\!\bolds$, the orientation $\vec{e}_i\prec\vec{e}^{{}\,\dagger i}$ in the composite BV-fibres
$W_{\bx}\cong V_{\bx}\mathbin{\widehat{\oplus}}\Pi V_{\bx}^{\dagger}$ equipped with two couplings~\eqref{EqTwoCouplings}, and the 
ordering of variations $\delta\bolds_1,\,\ldots,\,\delta\bolds_k$ specify the logic of operational 
Definition~\ref{DefSchouten}, which is given in this section.

\begin{rem}
For the sake of brevity, we extend the BV-Laplacian~$\Delta$ from the space $\ov{H}^n(\pi_{\BV})$ of integral
functionals $F_1,\,\dots\,,F_{\ell}$ to the space $\ov{\gM}^n(\pi_{\BV})$ of local functionals such as
$F_1\cdot\,\dots\,\cdot F_{\ell}$, the factors of which do not explicitly contain any built-in variations. To further this
extension verbatim onto the full space $\ov{\gN}^n(\pi_{\BV},T\pi_{\BV})\supsetneq\ov{\gM}^n(\pi_{\BV})$, one must remember
that it is forbidden to break the order in which the directed variation operators $\overrightarrow{\delta}\!\!\bolds_k$ and
$\overleftarrow{\delta}\!\!\bolds_k$ appear in the (ordered collection of) objects at hand. (Such concept is illustrated by
the third term in~\eqref{EqSpreadTwoOverTwo} below.)

Likewise, we extend $\Delta$ to products of just two factors; in the case of arbitrary number $\ell\ge2$ of building blocks
$F_1,\,\dots\,,F_{\ell}$ one proceeds inductively by using the ghost parity-graded Leibniz rule, then extending $\Delta$ 
onto the vector space $\ov{\gN}^n(\pi_{\BV},T\pi_{\BV})$ by linearity.
\end{rem}

Let $F=\int f(\bx_1,[\bq],[\bq^{\dagger}])\dvol(\bx_1)$ and $G=\int g(\bx_2,[\bq],[\bq^{\dagger}])\dvol(\bx_2)$ be integral
functionals $\Gamma(\pi_{\BV})\to\Bbbk$ and let $\delta\bolds=(\delta s;\delta s^{\dagger})$ be a normalized test shift of
their product's argument $s\in\Gamma(\pi_{\BV})$. We now define the operator $\Delta$ 
acting on the element $F\cdot G$ at $\bolds$ by variations
first along $(0;\delta s^{\dagger})$ and then along $(\delta s;0)$.

According to~\eqref{EqLeftRightVariations}, the object to start with is
\begin{multline*}
\int_M\Id\bby_1\int_M\Id\bby_2\sum_{\substack{i_1,i_2\\j_1,j_2}}\sum_{\substack{|\sigma_1|\ge0\\|\sigma_2|\ge0}}\Biggl\{
(\delta s^{i_1})\left(\frac{\overleftarrow{\dd}}{\dd\bby_1}\right)^{\sigma_1}(\bby_1)
\left\langle\vec{e}_{i_1}(\bby_1),\vec{e}^{{}\,\dagger j_1}(\cdot)\right\rangle
\frac{\overrightarrow{\dd}}{\dd q^{j_1}_{\sigma_1}}\circ\\\circ
(\delta s^{\dagger}_{i_2})\left(\frac{\overleftarrow{\dd}}{\dd\bby_2}\right)^{\sigma_2}(\bby_2)
\left\langle-\vec{e}^{{}\,\dagger i_2}(\bby_2),\vec{e}_{j_2}(\cdot)\right\rangle
\frac{\overrightarrow{\dd}}{\dd q^{\dagger}_{j_2,\sigma_2}}\Biggr\}\\
\left(\int f(\bx_1,[\bq],[\bq^{\dagger}])\dvol(\bx_1)\cdot\int g(\bx_2,[\bq],[\bq^{\dagger}])\dvol(\bx_2\right)(s).
\end{multline*}
Their order preserved, the directed operators $\overrightarrow{\delta}\!\!s$ and $\overrightarrow{\delta}\!\!s^{\dagger}$ 
spread over the two factors $F$ and $G$ by the binomial formula because of the Leibniz rule for graded derivations
$\overrightarrow{\dd}/\dd q^{j_1}_{\sigma_1}$ and $\overrightarrow{\dd}/\dd q^{\dagger}_{j_2,\sigma_2}$. Note that whenever
the ghost parity-odd object $\overrightarrow{\dd}/\dd q^{\dagger}_{j_2,\sigma_2}$ overtakes the density $f$ of ghost
parity $\GH(F)$, there appears an overall sign factor $(-)^{\GH(F)}$. 
We thus obtain
\begin{multline}\label{EqSpreadTwoOverTwo}
(\overrightarrow{\delta}\!\!s\circ\!\overrightarrow{\delta}\!\!s^{\dagger})(F)+(-)^{\GH(F)}\overrightarrow{\delta}\!\!s(F)
\overrightarrow{\delta}\!\!s^{\dagger}(G)+\overrightarrow{\delta}\!\!s
\xrightarrow{{}\cdot\overrightarrow{\delta}\!\!s^{\dagger}(F)\cdot{}}G
\\
+(-1)^{\GH(F)}F\cdot(\overrightarrow{\delta}\!\!s\circ\!\overrightarrow{\delta}\!\!s^{\dagger})(G).
\end{multline}
The next step is to push right through $F$ its single variations in the middle two terms of the above expression. This
yields the equality
\begin{multline}\label{EqFourTerms}
{}=
(\overrightarrow{\delta}\!\!s\circ\!\overrightarrow{\delta}\!\!s^{\dagger})(F)\cdot G+(-)^{\GH(F)}
\left\{(F)\overleftarrow{\delta}\!\!s\cdot\!\overrightarrow{\delta}\!\!s^{\dagger}(G)+\overrightarrow{\delta}\!\!s
\xrightarrow{{}\cdot(F)\overleftarrow{\delta}\!\!s^\dagger\cdot{}}G\right\}\\
{}+(-)^{\GH(F)}F\cdot(\overrightarrow{\delta}\!\!s\circ\!\overrightarrow{\delta}\!\!s^{\dagger})(G).
\end{multline}
We emphasize that the operators $\overleftarrow{\delta}\!\!s$ and $\overleftarrow{\delta}\!\!s^{\dagger}$ in the
variations $(F)\overleftarrow{\delta}\!\!s=\overrightarrow{\delta_{\bq}F}$ and
$(F)\overleftarrow{\delta}\!\!s^{\dagger}=\overrightarrow{\delta_{\bq^{\dagger}}F}$
are temporarily redirected to the left so that the middle terms in~\eqref{EqFourTerms} are $(-)^{\GH(F)}$ times
\begin{equation}\label{EqSchoutenInput}
\mathstrut\smash{
\overrightarrow{\delta_{\bq}F}\cdot\overleftarrow{\delta_{\bq^{\dagger}}F}+\overleftarrow{\delta_{\bq}}
\xrightarrow{\overrightarrow{\delta_{\bq^{\dagger}}F}\cdot}\overline{G}}\ ;
\end{equation}
this is the input datum for a traditional definition of the variational Schouten bracket (e.\,g., 
see~\cite{CattaneoFelderCMP2000} vs~\cite{YKS2008SIGMA}). Let us remember that the BV-fibres orientation 
$\bq\prec\bq^{\dagger}$ expressed by~\eqref{EqTwoCouplings} is built into the last term of~\eqref{EqSchoutenInput} even if
it is written as follows,
$$(F)\overleftarrow{\delta_{\bq}}\cdot\overrightarrow{\delta_{\bq^{\dagger}}}(G)+
(F)\overleftarrow{\delta_{\bq^{\dagger}}}\cdot\overrightarrow{\delta_{\bq}}(G).
$$
Should this be the notation for input, one then usually proclaims that ``differential 1-forms anticommute'' so that
$\langle\delta\bq^{\dagger}\wedge\delta\bq\rangle=-\langle\delta\bq\wedge\delta\bq^{\dagger}\rangle=-1$ in
$\lshad F,G\rshad=\langle\overrightarrow{\delta F}\wedge\overleftarrow{\delta G}\rangle$.

We now are almost in a position to (re)configure the couplings in the four terms of~\eqref{EqFourTerms}. The first term will
of course become $\Delta F\cdot G$, and the last will provide $(-)^{\GH(F)}F\cdot\Delta G$; one is here allowed to integrate
by parts (as explained in section~\ref{SecByParts}) in order to shake the derivatives off $\delta s^{i_1}$ and
$\delta s^{\dagger}_{i_2}$ prior to evaluation of couplings in the resulting object's numeric value at its argument $s$. Yet there
remains one more logical step to be done with~\eqref{EqSchoutenInput}: let us reverse back 
$\overleftarrow{\delta}\!\!s\mapsto\overrightarrow{\delta}\!\!s$ and
$\overleftarrow{\delta}\!\!s^{\dagger}\mapsto\overrightarrow{\delta}\!\!s^{\dagger}$ so that on one hand, the vertical
differentials fall on $F$ but on the other hand, the normalization of the basis which stands near
$\delta s^i(\bby_1)$ and $\delta s^{\dagger}_i(\bby_2)$ is the \emph{first} not second column in~\eqref{EqAlmostComplex}.
This yields the following integrand of 
$(-)^{\GH(F)}\int_M\Id\bby_1\int_M\Id\bby_2\int_M\dvol(\bx_1)\int_M\dvol(\bx_2)$,
with a summation over $i_1,i_2,j_1,j_2$, and $|\sigma_1|\ge0,\ |\sigma_2|\ge0$,
\begin{align}\label{EqPriorToReconfigSchouten}
{}&\left.
\left(f(\bx_1,[\bq],[\bq^{\dagger}])\frac{\overleftarrow{\dd}}{\dd q^{j_1}_{\sigma_1}}\right)
\right|_{j^{\infty}_{\bx_1}(s)}
\langle\vec{e}^{{}\,\dagger j_1}(\bx_1),+\vec{e}_{i_1}(\bby_1)\rangle
\left(\tfrac{\overrightarrow{\dd}}{\dd\bby_1}\right)^{\sigma_1}(\delta s^{i_1})(\bby_1)\cdot\\
{}&{}\qquad{}
\cdot(\delta s^{\dagger}_{i_2})\left(\tfrac{\overleftarrow{\dd}}{\dd\bby_2}\right)^{\sigma_2}(\bby_2)
\langle-\vec{e}^{{}\,\dagger i_2}(\bby_2),\vec{e}_{j_2}(\bx_2)\rangle
\left.\left(
\frac{\overrightarrow{\dd}}{\dd q^{\dagger}_{j_2,\sigma_2}}g(\bx_2,[\bq],[\bq^{\dagger}])
\right)\right|_{j^{\infty}_{\bx_2}(s)}+{} \notag\\
{}+& 
\left.\left(f(\bx_1,[\bq],[\bq^{\dagger}])\frac{\overleftarrow{\dd}}{\dd q^{\dagger}_{j_1,\sigma_1}}\right)
\right|_{j^{\infty}_{\bx_1}(s)}\cdot{} \notag\\
{}&\Biggl\{
\begin{matrix}
&(\delta s^{i_1})\left(\tfrac{\overleftarrow{\dd}}{\dd\bby_1}\right)^{\sigma_2}(\bby_1)
\langle\vec{e}_{i_1}(\bby_1)|&&|\vec{e}^{{}\,\dagger j_2}(\bx_2)\rangle\\
\langle\vec{e}_{j_1}(\bx_1)|&&|-\vec{e}^{{}\,\dagger i_2}(\bby_2)\rangle
\left(\tfrac{\overrightarrow{\dd}}{\dd\bby_2}\right)^{\sigma_1}(\delta s_{i_2}^{\dagger})(\bby_2)\rangle&
\end{matrix}
\Biggr\}\cdot \notag\\
{}&\mbox{\hbox to 97mm {{ }\hfil { }}}\left.\left(
\frac{\overrightarrow{\dd}}{\dd q^{j_2}_{\sigma_2}}g(\bx_2,[\bq],[\bq^{\dagger}])
\right)\right|_{j^{\infty}_{\bx_2}(s)}.\notag
\end{align}
The integrations by parts are performed and couplings are reconfigured at the end of the day
in 
exactly same manner as it has been done in Definition~\ref{DefBV}; 
let us recall that we now define the BV-\/Laplacian on a larger space. Namely, the variations 
couple with the dual variations whereas the differentials of functionals' densities attach to each other.

\begin{define}\label{DefSchouten}
The \emph{variational Schouten bracket} of two integral functionals
\[
{F=\int f(\bx_1,[\bq],[\bq^{\dagger}]\cdot\dvol(\bx_1)} 
\qquad
\text{and} 
\qquad
{G=\int g(\bx_2,[\bq],[\bq^{\dagger}]\cdot\dvol(\bx_2)}
\]
is the on\/-\/the\/-\/diagonal couplings surgery which, by using a normalized test shift
$\delta\bolds=(\delta s;0)+(0;\delta s^{\dagger})\in\Gamma(T\pi_{\BV})$,
yields the functional from $\ov{\gN}^n(\pi_{\BV},T\pi_{\BV})$ whose 
construction 
at a BV-\/section $s\in\Gamma(\pi_{\BV})$ is\footnote{Note that the directions of $\dd/\dd\bby_i$ are reversed so that the minus signs appear. We emphasize that, prior to the evaluation of reconfigured couplings, 
the (co)\/vectors at~$\bx_j$ channel the partial derivatives to~$f$ or~$g$ according to the couplings' old arrangement.}
\begin{align*}
{}&\int_M\Id\bby_1\int_M\Id\bby_2\int_M\Id\bx_1 
\int_M\dvol(\bx_2)
\Biggl[\left.\left(
f(\bx_1,[\bq],[\bq^{\dagger}])\frac{\overleftarrow{\dd}}{\dd q^{j_1}_{\sigma_1}}
\right)\right|_{j^{\infty}_{\bx_1}(s)}\\
{}&\Biggl\{
\begin{matrix}
\left(-\frac{\overleftarrow{\dd}}{\dd\bby_1}\right)^{\sigma_1}\delta s^{i_1}(\bby_1)
\overbrace{
\langle\vec{e}_{i_1}(\bby_1),-\vec{e}^{{}\,\dagger i_2}(\bby_2)\rangle
}^{-1}
\,\delta s^{\dagger}_{i_2}(\bby_2)\left(-\frac{\overrightarrow{\dd}}{\dd\bby_2}\right)^{\sigma_2}
\\
\underbrace{
\langle\vec{e}^{{}\,\dagger j_1}(\bx_1)|\qquad\qquad\qquad\mathstrut,
\mathstrut\qquad\qquad\qquad|\vec{e}_{j_2}(\bx_2)\rangle
}_{-1}
\end{matrix}
\Biggr\}\left.\left(\frac{\overrightarrow{\dd}}{\dd q^{\dagger}_{j_2,\sigma_2}}g(\bx_2,[\bq],[\bq^{\dagger}])
\right)\right|_{j^{\infty}_{\bx_2}(s)}\\ 
\end{align*}
\begin{multline*}
{}
+\left.\left(
f(\bx_1,[\bq],[\bq^{\dagger}])
\frac{\overleftarrow{\dd}}{\dd q^{\dagger}_{j_1,\sigma_1}}
\right)\right|_{j^{\infty}_{\bx_1}(s)}\\
{}
{}\qquad\quad\Biggl\{
\begin{matrix}
\delta s^{i_1}(\bby_1)\left(-\tfrac{\overrightarrow{\dd}}{\dd\bby_1}\right)^{\sigma_2}
\overbrace{ \langle\vec{e}_{i_1}(\bby_1),(-\vec{e}^{{}\,\dagger i_2})(\bby_2)\rangle }^{-1}
\cdot\left(-\tfrac{\overleftarrow{\dd}}{\dd\bby_2}\right)^{\sigma_1}\delta s^{\dagger}_{i_2}(\bby_2)\\
\underbrace{\langle\vec{e}_{j_1}(\bx_1)|\qquad\qquad
\qquad\qquad,\mathstrut\qquad\qquad\qquad\qquad|\vec{e}^{{}\,\dagger j_2}(\bx_2)\rangle}_{+1}
\end{matrix}
\Biggr\}\\
{}
\mbox{\hbox to 97mm {{ }\hfil { }}}
\left.\left(
\frac{\overrightarrow{\dd}}{\dd q^{j_2}_{\sigma_2}}g(\bx_1,[\bq],[\bq^{\dagger}])
\right)\right|_{j^{\infty}_{\bx_2}(s)}
\Biggr].
\end{multline*}
Note that the inner couplings between variations provide a restriction to the diagonal $\bby_1=\bby_2$ and yield the
singular integral operators which then act to the right via multiplication by $-1$ only if $\bby_2=\bx_2$ and
$\bby_1=\bx_1$, respectively. The outer coupling then furnishes the main diagonal $\bx_1=\bby_1=\bby_2=\bx_2$,
restricting the objects further to the same BV-fibre point in the total space of the BV-bundle.
This reveals why over each point of the base $M^n$ the (derivatives of the) densities $f$ and $g$ are restricted to the
infinite jet of the same section $s$; this also means that, since the moment when the couplings
are reconfigured, the volume element $\dvol(\bx_1)$ is discarded because appears a new
singular linear integral operator with a standard sign~$\int\Id\bx_1$.
\end{define}

We conclude the reasoning and sum up the definitions and notations in the following theorem.

\begin{theor}\label{ThLaplaceOnProduct}
The BV-Laplacian $\Delta$ is the linear operator
$$\ov{\gN}^n(\pi_{\BV},T\pi_{\BV})\times\ov{\gN}^n(\pi_{\BV},T\pi_{\BV})\to\ov{\gN}^n(\pi_{\BV},T\pi_{\BV})$$
which acts on products of functionals $F$ and $G\in\ov{\gN}^n(\pi_{\BV},T\pi_{\BV})$
by the rule
\begin{equation}\label{EqLaplacianOnSchouten}
\Delta(F\cdot G)=\Delta(F)\cdot G+(-)^{\GH(F)}\lshad F,G\rshad+(-)^{\GH(F)}F\cdot\Delta G.
\end{equation}
The variational Schouten bracket $\lshad\,,\,\rshad$ measures the deviation for the BV-Laplacian $\Delta$ from being a derivation.\\
$\bullet$\quad 
After integration by parts, Definition~\textup{\ref{DefSchouten}} implies the renouned coordinate formula
\begin{multline}\label{EqFormulaSchouten}
\lshad F,G\rshad=\int\dvol(\bx)
\Biggl(
\frac{\overrightarrow{\delta}\!f(\bx,[\bq],[\bq^{\dagger}])}{\delta\bq}\cdot
\frac{\overleftarrow{\delta}\!g(\bx,[\bq],[\bq^{\dagger}])}{\delta\bq^{\dagger}}-\\-
\frac{\overrightarrow{\delta}\!f(\bx,[\bq],[\bq^{\dagger}])}{\delta\bq^{\dagger}}\cdot
\frac{\overleftarrow{\delta}\!g(\bx,[\bq],[\bq^{\dagger}])}{\delta\bq}
\Biggr).
\end{multline}
\end{theor}

\begin{rem}\label{RemTwoToOneInSchouten}
Let us recall from Remark~\ref{RemRolesInt} that the building blocks of local functionals are encoded by equivalence classes of their densities, whereas the underlying integration manifold~$M^n$ is endowed with the field\/-\/dependent volume element~$\dvol(\bx,\phi)$. The variational Schouten bracket transforms two given integral functionals~$F$ and~$G$ into~$\lshad F,G\rshad$. For every configuration of physical fields~$\phi\in\Gamma(\pi)$, the integration measure is the same in~$F$,\ $G$,\ and~$\lshad F,G\rshad$. This is because the couplings are local over points $\bigl(\bx,\phi(\bx)\bigr)$ in the total space of the bundle~$\pi$ of physical fields, see Remark~\ref{RemTwoCouplings} on p.~\pageref{RemTwoCouplings}\,; the equality of local sections~$\phi$ at which all (derivatives of) functionals' densities are evaluated ensures the equality of metric tensor elements~$g_{\mu\nu}$ in all functionals by virtue of Einstein's general relativity equations.
\end{rem}

The operational definition of the \textsl{antibracket} $\lshad\,,\,\rshad$ determines the way how this structure acts on 
the square $\ov{\gN}^n(\pi_{\BV},T\pi_{\BV}) \times \ov{\gN}^n(\pi_{\BV},T\pi_{\BV})$ of entire space $\ov{\gN}^n(\pi_{\BV},T\pi_{\BV})$ containing formal products of functionals.

\begin{theor}\label{ThSchoutenOnProduct}
Let $F$, $G$, and $H\in\ov{\gN}^n(\pi_{\BV},T\pi_{\BV})$ be ghost parity\/-\/homogeneous functionals. 
The variational Schouten bracket
$\lshad\,,\,\rshad\colon\ov{\gN}^n(\pi_{\BV},T\pi_{\BV})\times\ov{\gN}^n(\pi_{\BV},T\pi_{\BV})\to
\ov{\gN}^n(\pi_{\BV},T\pi_{\BV})$
has the following properties\textup{:}
\begin{itemize}
\item[(i)]
The value of $\lshad\,,\,\rshad$ at two arguments $F$ and $G\cdot H$ is
\begin{equation}\label{EqSchoutenOnProduct}
\lshad F,G\cdot H\rshad=\lshad F,G\rshad\cdot H + (-)^{(\GH(F)-1)\GH(G)}G\cdot\lshad F,H\rshad.
\end{equation}
This formula recursively extends to products of arbitrary finite number of factors in the second argument.
\item[(ii)]
The bracket $\lshad\,,\,\rshad$ is shifted\/-\/graded skew\/-\/symmetric\textup{:}
\begin{equation}\label{EqSchoutenSkew}
\lshad F,G\rshad=-(-)^{(\GH(F)-1)\cdot(\GH(G)-1)}\lshad G,F\rshad,
\end{equation}
which extends $\lshad\,,\,\rshad$ to products of arbitrary finite number of factors taken as its first 
argument in~\eqref{EqSchoutenOnProduct}.
\item[(iii)]
The bracket $\lshad\,,\,\rshad$ satisfies the shifted\/-\/graded Jacobi identity
\begin{multline}
(-)^{(\GH(F)-1)(\GH(H)-1)}\lshad F,\lshad G,H\rshad\rshad+
(-)^{(\GH(F)-1)(\GH(G)-1)}\lshad G,\lshad H,F\rshad\rshad+{}\\
{}+
(-)^{(\GH(G)-1)(\GH(H)-1)}\lshad H,\lshad F,G\rshad\rshad=0,\label{EqJacobiSchouten}
\end{multline}
which stems from graded Leibniz rule~\eqref{EqJacobiLeibniz}
for evolutionary vector fields $\bQ^F$ defined by the rule 
$\bQ^F(\cdot)\cong\lshad F,\,\cdot\,\rshad$ \textup{(}here the equivalence up to integration by parts 
is denoted by $\cong$\,\textup{)}.
\end{itemize}
Finally, the variational Schouten bracket extends by linearity to 
formal sums of elements from $\ov{\gN}^n(\pi_{\BV},T\pi_{\BV})$.
\end{theor}

\begin{proof}
The bilinearity of $\lshad\,,\,\rshad$ is obvious. It is also clear that the terms in $\lshad F,G\cdot H\rshad$ are
grouped in two parts: those in which the ghost-parity graded derivations $\overrightarrow{\dd}/\dd\bq^{\dagger}$ act
on $G$ and those for $H$; the former do not contribute with any extra sign factors whereas the latter do --- in a way
which depends on the parity $\GH(G)$. This means that $\lshad F,G\cdot H\rshad=\lshad F,G\rshad\cdot H+\ldots$;
to grasp the sign in front of the term which has been omitted, let us swap the graded multiples $G$ and $H$. We have that
$G\cdot H=(-)^{\GH(G)\GH(H)}H\cdot G$, whence $\lshad F,G\cdot H\rshad=(-)^{\GH(G)\GH(H)}\lshad F,H\rshad\cdot G+\cdots$.
By recalling that $\GH(\lshad F,H\rshad)=\GH(F)+\GH(H)-1$, we conclude that
$$\lshad F,G\cdot H\rshad=\lshad F,G\rshad\cdot H+(-)^{\GH(G)\GH(H)}(-)^{(\GH(F)+\GH(H)-1)\cdot\GH(G)}G\cdot
\lshad F,H\rshad,
$$
which yields formula~\eqref{EqSchoutenOnProduct}.

Proving~\eqref{EqSchoutenSkew} amounts to a count of signs whenever the bracket~$\lshad F,G\rshad$ of an ordered pair of ghost parity\/-\/graded objects is virtually transformed into~$\lshad G,F\rshad$. By using the 
rule of signs for odd\/-\/parity coordinates, $q^\dagger_{\alpha,\sigma}\cdot q^\dagger_{\beta,\tau}=-q^\dagger_{\beta,\tau}\cdot q^\dagger_{\alpha,\sigma}$, we first note that
\[
\frac{\overrightarrow{\dd}}{\dd q^\dagger_{j_2,\sigma_2}} g\bigl(\bx_2,[\bq],[\bq^\dagger]\bigr) =
(-)^{\GH(G)-1}\left(g\bigl(\bx_2,[\bq],[\bq^\dagger]\bigr)\right)
\frac{\overleftarrow{\dd}}{\dd q^\dagger_{j_2,\sigma_2}},
\]
with a similar formula for the left-{} and right\/-\/acting graded derivative of~$f$. By swapping the (variational) derivatives of the densities~$f$ and~$g$, we gain the signs~$(-)^{\GH(F)\cdot(\GH(G)-1)}$ and~$(-)^{(\GH(F)-1)\cdot\GH(G)}$ for the respective terms in~\eqref{EqPriorToReconfigSchouten} on p.~\pageref{EqPriorToReconfigSchouten}. Combined together, the two steps accumulate equal factors $(-)^{(\GH(F)+1)\cdot(\GH(G)-1)}=(-)^{(\GH(F)-1)\cdot(\GH(G)+1)}=
(-)^{(\GH(F)-1)\cdot(\GH(G)-1)}$. 
Thirdly, by comparing $(-)^{(\GH(F)-1)\cdot(\GH(G)-1)}\,\lshad F,G\rshad$ --\,in which the derivatives of~$f$ and~$g$ are interchanged and the derivations' directions are reversed\,-- with~$\lshad G,F\rshad$, we conclude that the reconfiguration of couplings in the second term in~\eqref{EqPriorToReconfigSchouten} for the former expression yields \emph{minus} the first term in~$\lshad G,F\rshad$. Likewise, the couplings reattachment in the first term of such~\eqref{EqPriorToReconfigSchouten} produces minus the second term in~$\lshad G,F\rshad$. This is because the (co)\/vectors in the differentials of densities remain unswapped, now going in the `wrong' order.

We now refer to~\cite[Proposition~3]{Lorentz12} for a proof of property 
(iii) in a wider, non-commutative setup of cyclic words
(cf.~\cite{SQS11,KontsevichCyclic,OlverSokolovCMP1997}). It is remarkable that the 
reasoning persists within a 
na\"\i ve theory of variations, not referring to our main idea that each test shift brings its own copy of the base $M^n$
into the picture. A key point in the proof is that the rule $\bQ^F(\cdot)\cong\lshad F,\,\cdot\,\rshad$
naturally associates with functionals $F$ the evolutionary fields $\bQ^F$ on the infinite jet superbundles at hand, and
with \emph{evolutionary} vector fields it does not matter under `whose'' total derivatives such fields dive, obeying their
defining property $[\bQ^F,\overrightarrow{\Id}/\Id\bx]=0$ (i.\,e., any integrations by parts, which transform the derivatives
$\overrightarrow{\dd}/\dd\bby_i$ falling on test shifts into total derivatives $\overrightarrow{\Id}/\Id\bx$ falling on the
functionals' densities, do not mar the outcome even if one attempts to perform such integrations ahead of time).
\end{proof}

\subsection{Main result\,: the proof of properties \protect{(\ref{EqZimes}\/--\/\ref{EqDeltaSquareIntro})}}\label{SecProof}
We are ready to \emph{prove} the main interrelations between the BV-Laplacian $\Delta$ and variational Schouten bracket
$\lshad\,,\,\rshad$. Let us recall that either a validity of these properties was postulated (see~\cite{GomisParisSamuel})
or an \textsl{ad hoc} regularization technique was formally employed in the literature in order to mask the seemingly present
divergencies (which are actually not there), cf.~\cite[\S15]{HenneauxTeitelboim}.

Let us fix the terms. In what follows we refer to building blocks from $\ov{H}^n(\pi_{\BV})$ and their descendants 
--~containing reconfigured variations~-- from $\ov{H}^{n(1+k)}(\pi_{\BV}\times T\pi_{\BV}\times\ldots\times T\pi_{\BV})$
as \textsl{integral functionals}. Such objects will be used for bases of inductive proofs of Lemmas~\ref{LemmaBaseLapSchouten} and~\ref{LBVOperatorDifferential}.
We then extend the properties~\eqref{EqZimes} and $\Delta^2=0$ to the space 
$\ov{\gN}^n(\pi_{\BV},T\pi_{\BV})\supseteq\ov{\gM}^n(\pi_{\BV})$ of \textsl{local functionals}, that is, of formal sums of 
products of (varied descendants of) building blocks.

\begin{lemma}\label{LemmaBaseLapSchouten}
Let $F\in\overline{H}^{n(1+k)}\bigl(\pi_{\BV}\times T\pi_{\BV}\times\ldots\times T\pi_{\BV}\bigr)$ and 
$G\in\overline{H}^{n(1+\ell)}\bigl(\pi_{\BV}\times T\pi_{\BV}\times\ldots\times T\pi_{\BV}\bigr)$ be two 
integral functionals\textup{;} here $k,\ell\geqslant0$. Then
\begin{equation}\label{EqLapSchouten}
\Delta\bigl(\schouten{F,G}\bigr) = \schouten{\Delta F,G} + (-)^{\GH(F)-1}\schouten{F,\Delta G}.
\end{equation}
\end{lemma}

\begin{proof}
The key idea is that the structures $\Delta$ and $\lshad\,,\,\rshad$ yield equivalence classes of integral functionals
which, after an integration by parts at the end of the day, are \emph{independent} of a choice of the built-in test shifts
normalized by~\eqref{EqNormalize}. Consequently, the composite structure $\Delta(\lshad{\cdot},{\cdot}\rshad)$ does not change
under swapping $\delta s_1^{\alpha}\rightleftarrows\delta s_2^{\beta}$,
$\delta s_{1,\alpha}^{\dagger}\rightleftarrows\delta s_{2,\beta}^{\dagger}$ of the respective variations $\delta\bolds_1$ and
$\delta\bolds_2$ in $\Delta$ and $\lshad\,,\,\rshad$. Hence the terms which are skew-symmetric under such
exchange necessarily vanish.

For the sake of clarity, let us assume that $F=\int f(\bx_1,[\bq],[\bq^{\dagger}])\,\dvol(\bx_1)$ and
$G=\int g(\bx_2,[\bq],[\bq^{\dagger}])\,\dvol(\bx_2)$ are just building blocks from the cohomology group
$\ov{H}^n(\pi_{\BV})$; this simplification is legitimate because new variations which come from $\Delta$ and
$\lshad\,,\,\rshad$ do not interfere with any other test shifts if those are already absorbed by the densities $f$ and $g$.
Suppose that $\delta\bolds_1$ and $\delta\bolds_2$ are two normalized variations of a section $s\in\Gamma(\pi_{\BV})$. 
By definition, we have that\footnote{To keep track of their origin, we let the directed derivatives~$\dd/\dd\bby_i$ or~$\dd/\dd\bz_j$ remain falling on the respective coefficients in~$\delta\bolds_1$ and~$\delta\bolds_2$; the integration by parts is performed in a standard way prior to the reconfigurations which are shown in the formula.}
\begin{align*}
&\Delta\left(\lshad F,G\rshad\right)(s)=
\int_M\Id\bz_1\int_M\Id\bz_2\int_M\Id\bby_1\int_M\Id\bby_2\int_M\Id\bx_1 
\int_M\dvol(\bx_2)\ \cdot
\\
&\Biggl\{
(\delta s_1^{\alpha})\left(\tfrac{\overleftarrow{\dd}}{\dd\bz_1}\right)^{\sigma_1}(\bz_1)\,
\left\langle\vec{e}_{\alpha}(\bz_1),-\vec{e}^{{}\,\dagger \alpha}(\bz_2)\right\rangle\,
(\delta s^{\dagger}_{1,\alpha})\left(\tfrac{\overleftarrow{\dd}}{\dd\bz_2}\right)^{\sigma_2}(\bz_2)
\langle\vec{e}^{{}\,\dagger \alpha}(\cdot),\vec{e}_{\alpha}(\cdot)\rangle
\frac{\overrightarrow{\dd}}{\dd q^{\alpha}_{\sigma_1}}\frac{\overrightarrow{\dd}}{\dd q^{\dagger}_{\alpha,\sigma_2}} 
\end{align*}
\begin{align*}
&\quad\smash{\Biggl[}
f(\bx_1.[\bq],[\bq^{\dagger}])\smash{\frac{\overleftarrow{\dd}}{\dd q^{\beta}_{\tau_1}}}
\,\underline{\langle\vec{e}^{{}\,\dagger \beta}(\bx_1)|}\,
\Bigl\langle\left(\tfrac{\overrightarrow{\dd}}{\dd\bby_1}\right)^{\tau_1}(\delta s_2^{\beta})(\bby_1)\,\vec{e}_{\beta}(\bby_1),
\\
&\mbox{\hbox to 25mm {{ }\hfil { }}}
{-}\vec{e}^{{}\,\dagger.\beta}(\bby_2)\,(\delta s^{\dagger}_{2,\beta})
\left(\tfrac{\overleftarrow{\dd}}{\dd\bby_2}\right)^{\tau_2}(\bby_2)\Bigr\rangle
\,\underline{|\vec{e}_{\beta}(\bx_2)\rangle}\,
\frac{\overrightarrow{\dd}}{\dd q^{\dagger}_{\beta,\tau_2}}g(\bx_2,[\bq],[\bq^{\dagger}])+{}
\\
&{}\quad{}+f(\bx_1,[\bq],[\bq^{\dagger}])
\frac{\overleftarrow{\dd}}{\dd q^{\dagger}_{\beta,\tau_2}}
\,\underline{\langle\vec{e}_{\beta}(\bx_1)|}\,
\Bigl\langle(\delta s_2^{\beta})
\left(\tfrac{\overleftarrow{\dd}}{\dd\bby_1}\right)^{\tau_1}(\bby_1)\,\vec{e}_{\beta}(\bby_1),
\\
&\mbox{\hbox to 25mm {{ }\hfil { }}}
{-}\vec{e}^{{}\,\dagger \beta}(\bby_2)\,
\left(\tfrac{\overrightarrow{\dd}}{\dd\bby_2}\right)^{\tau_2}
(\delta s^{\dagger}_{2,\beta})(\bby_2)\Bigr\rangle
\,\underline{|\vec{e}^{{}\,\dagger \beta}(\bx_2)\rangle}\,
\frac{\overrightarrow{\dd}}{\dd q^{\beta}_{\tau_1}}g(\bx_2,[\bq],[\bq^{\dagger}])
\smash{\Biggr]
\left.\Biggr\}\right|_{\substack{j^{\infty}(s)\\\bx_i=\bby_j=\bz_k}}}.
\end{align*}
The partial derivatives 
$\overrightarrow{\dd}/\dd q^{\alpha}_{\sigma_1}\circ\overrightarrow{\dd}/\dd q^{\dagger}_{\alpha,\sigma_2}$
are distributed between the arguments 
$f$ and $g$ by the graded Leibniz rule. Whenever \emph{none} of the two operators overtakes the density of $F$,
the reconfiguration yields $\lshad\Delta F,G\rshad(\bolds)$. Likewise, if \emph{both} derivatives indexed by $\alpha$
overtake $F$ and an old derivative that fell on $g$, then we obtain $(-)^{\GH(F)-1}\lshad F,\Delta G\rshad(\bolds)$,
which is the second term in the right-hand side of~\eqref{EqLapSchouten}.
We claim that the remaining four terms cancel out
by virtue of independence of $\Delta$ and $\lshad\,,\,\rshad$ from a choice of normalized variations. To prove this claim,
we consecutively inspect the behaviour of those four terms under a swap $\delta\bolds_1\rightleftarrows\delta\bolds_2$
of coefficients in the normalized test shifts.

The first and second terms sum up to the difference
\begin{multline}\label{EqTwoCancel}
\left\langle(\delta s^{\alpha}_1)\left(\tfrac{\overleftarrow{\dd}}{\dd\bz_1}\right)^{\sigma_1}(\bz_1)\,
\overbrace{\vec{e}_{\alpha}(\bz_1),(-\vec{e}^{{}\,\dagger \alpha})(\bz_2)}^{-1}\,
(\delta s^{\dagger}_{1,\alpha})\left(\tfrac{\overleftarrow{\dd}}{\dd\bz_2}\right)^{\sigma_2}(\bz_2)
\right\rangle
\underbrace{\langle\vec{e}^{{}\,\dagger \alpha}(\bx_2),\vec{e}_{\alpha}(\bx_1)\rangle}_{-1}\,\cdot{}\\
\cdot\frac{\overrightarrow{\dd}}{\dd q^{\dagger}_{\alpha,\sigma_2}}f(\bx_1,[\bq],[\bq^{\dagger}])
\frac{\overleftarrow{\dd}}{\dd q^{\beta}_{\tau_1}}\,
\underline{\bigl\langle\vec{e}^{{}\,\dagger \beta}(\bx_1)\bigr|}  
\Biggl\langle\left(\tfrac{\overrightarrow{\dd}}{\dd\bby_1}\right)^{\tau_1}
(\delta s_2^{\beta})(\bby_1)\,
\overbrace{\vec{e}_{\beta}(\bby_1),(-\vec{e}^{{}\,\dagger \beta})(\bby_2)}^{-1}\\
(\delta s^{\dagger}_{2,\beta})
\left(\tfrac{\overleftarrow{\dd}}{\dd\bby_2}\right)^{\tau_2}(\bby_2)\Biggr\rangle
\underbrace{\bigl|\vec{e}_{\beta}(\bx_2)\bigr\rangle}_{-1}
\frac{\overrightarrow{\dd}}{\dd q^{\alpha}_{\sigma_1}}
\frac{\overrightarrow{\dd}}{\dd q^{\dagger}_{\beta,\tau_2}}g(\bx_2,[\bq],[\bq^{\dagger}])+{}\\
{}+
\left\langle(\delta s^{\alpha}_1)\left(\tfrac{\overleftarrow{\dd}}{\dd\bz_1}\right)^{\sigma_1}(\bz_1)\,
\overbrace{\vec{e}_{\alpha}(\bz_1),(-\vec{e}^{{}\,\dagger \alpha})(\bz_2)}^{-1}\,
(\delta s^{\dagger}_{1,\alpha})\left(\tfrac{\overleftarrow{\dd}}{\dd\bz_2}\right)^{\sigma_2}(\bz_2)
\right\rangle
\underbrace{\langle\vec{e}^{{}\,\dagger \alpha}(\bx_1),\vec{e}_{\alpha}(\bx_2)\rangle}_{-1}\,\cdot\\
\cdot(-)^{\GH(F)-1}
\frac{\overrightarrow{\dd}}{\dd q^{\alpha}_{\sigma_1}}f(\bx_1,[\bq],[\bq^{\dagger}])
\frac{\overleftarrow{\dd}}{\dd q^{\dagger}_{\beta,\tau_2}}
\underline{\bigl\langle\vec{e}_{\beta}(\bx_1)\bigr|}   
\Biggl\langle(\delta s_2^{\beta})\left(\tfrac{\overleftarrow{\dd}}{\dd\bby_1}\right)^{\tau_1}(\bby_1)\,
\overbrace{\vec{e}_{\beta}(\bby_1),(-\vec{e}^{{}\,\dagger \beta})(\bby_2)}^{-1}\\
\left(\tfrac{\overrightarrow{\dd}}{\dd\bby_2}\right)^{\tau_2}(\delta s^{\dagger}_{2,\beta})(\bby_2)\Biggr\rangle
\underbrace{\bigl|\vec{e}^{{}\,\dagger \beta}(\bx_2)\bigr\rangle}_{+1}\,
\frac{\overrightarrow{\dd}}{\dd q^{\dagger}_{\alpha,\sigma_2}}
\frac{\overrightarrow{\dd}}{\dd q^{\beta}_{\tau_1}}g(\bx_2,[\bq],[\bq^{\dagger}]).
\end{multline}
Recalling that 
$$f(\bx_1,[\bq],[\bq^{\dagger}])\frac{\overleftarrow{\dd}}{\dd q^{\dagger}_{\beta,\tau_2}}=(-)^{\GH(F)-1}
\frac{\overrightarrow{\dd}}{\dd q^{\dagger}_{\beta,\tau_2}}f(\bx_1,[\bq],[\bq^{\dagger}]),$$
let us swap the derivations which fall on $f$ from the left and right; this eliminates the sign $(-)^{\GH(F)-1}$.
We proceed likewise for $g$ and then transport the variations $\delta\bolds_1$ and $\delta\bolds_2$, exchanging their places
(and their r\^oles with respect to $\Delta$ and $\lshad\,,\,\rshad$). The second term in formula ~\eqref{EqTwoCancel} becomes
\begin{multline*}
\smash{\Biggl\langle(\delta s_2^{\beta})
\left(\tfrac{\overleftarrow{\dd}}{\dd\bby_1}\right)^{\tau_1}(\bby_1)\,
\overbrace{\vec{e}_{\beta}(\bby_1),(-\vec{e}^{{}\,\dagger \beta})(\bby_2)}^{-1}\,
\left(\tfrac{\overrightarrow{\dd}}{\dd\bby_2}\right)^{\tau_2}(\delta s^{\dagger}_{2,\beta})(\bby_2)\Biggr\rangle}
\underbrace{\bigl\langle\vec{e}_{\beta}(\bx_1),\vec{e}^{{}\,\dagger \beta}(\bx_2)\bigr\rangle}_{+1}\\
\cdot
\frac{\overrightarrow{\dd}}{\dd q^{\dagger}_{\beta,\tau_2}}
f(\bx_1,[\bq],[\bq^{\dagger}])\frac{\overleftarrow{\dd}}{\dd q^{\alpha}_{\sigma_1}}
\underline{\bigl\langle\vec{e}^{{}\,\dagger \alpha}(\bx_1)\bigr|}   
\Biggl\langle
\left(\tfrac{\overrightarrow{\dd}}{\dd\bz_1}\right)^{\sigma_1}(\delta s^{\alpha}_1)(\bz_1)
\overbrace{\vec{e}_{\alpha}(\bz_1),(-\vec{e}^{{}\,\dagger \alpha})(\bz_2)}^{-1}
(\delta s^{\dagger}_{1,\alpha})\left(\tfrac{\overleftarrow{\dd}}{\dd\bz_2}\right)(\bz_2)\Biggr\rangle\\
\underbrace{\bigl|\vec{e}_{\alpha}(\bx_2)\bigr\rangle}_{-1}\,
\frac{\overrightarrow{\dd}}{\dd q^{\beta}_{\tau_1}}
\frac{\overrightarrow{\dd}}{\dd q^{\dagger}_{\alpha,\sigma_2}}
g(\bx_2,[\bq],[\bq^{\dagger}]).
\end{multline*}
It is now readily seen that the first term in~\eqref{EqTwoCancel} and this equivalent expression of its second term are
opposite to each other. Indeed, relabel the summation indexes $\alpha\rightleftarrows\beta$, $\sigma\rightleftarrows\tau$
so that $\delta s^{\alpha}_1\rightleftarrows\delta s^{\beta}_2$, 
$\delta s^{\dagger}_{1,\alpha}\rightleftarrows\delta s^{\dagger}_{2,\beta}$,
and swap the copies of base manifold $M^n$ by $\bby\rightleftarrows\bz$.
Due to the second factors in the products $(-1)\cdot(-1)\cdot(-1)\cdot(-1)=+1$ versus $(-1)\cdot(+1)\cdot(-1)\cdot(-1)=-1$,
the two terms in~\eqref{EqTwoCancel} cancel out after the integration by parts and evaluation of the couplings in view
of~\eqref{EqNormalize}.

Next, the integrand of $\Delta\bigl(\lshad F,G\rshad\bigr)(s)$ contains a restriction to the infinite jet $j^{\infty}(s)$
of the third term, which is
\begin{align*}
&\Biggl\langle(\delta s_1^{\alpha})
\left(\tfrac{\overleftarrow{\dd}}{\dd\bz_1}\right)^{\sigma_1}(\bz_1)\,
\overbrace{\vec{e}_{\alpha}(\bz_1),(-\vec{e}^{{}\,\dagger \alpha})(\bz_2)}^{-1}\,
(\delta s^{\dagger}_{1,\alpha})\left(\tfrac{\overleftarrow{\dd}}{\dd\bz_2}\right)^{\sigma_2}(\bz_2)\Biggr\rangle
\underbrace{\bigl\langle\vec{e}^{{}\,\dagger \alpha}(\bx_2),\vec{e}_{\alpha}(\bx_1)\bigr\rangle}_{-1}\\
&\ {}\cdot
\frac{\overrightarrow{\dd}}{\dd q^{\dagger}_{\alpha,\sigma_2}}
\left(f(\bx_1,[\bq],[\bq^{\dagger}])\frac{\overleftarrow{\dd}}{\dd q^{\dagger}_{\beta,\tau_2}}\right)
\\
&\quad\underline{\bigl\langle\vec{e}_{\beta}(\bx_1)\bigr|}  
\Biggl\langle
(\delta s^{\beta}_2)\left(\tfrac{\overleftarrow{\dd}}{\dd\bby_1}\right)^{\tau_1}(\bby_1)\,
\overbrace{\vec{e}_{\beta}(\bby_1),(-\vec{e}^{{}\,\dagger \beta})(\bby_2)}^{-1}\,
\left(\tfrac{\overrightarrow{\dd}}{\dd\bby_2}\right)^{\tau_2}
(\delta s^{\dagger}_{2,\beta})(\bby_2)\Biggr\rangle
\underbrace{\bigl|\vec{e}^{{}\,\dagger \beta}(\bx_2)\bigr\rangle}_{+1}
\\
&\mbox{\hbox to 107mm {{ }\hfil { }}}
\frac{\overrightarrow{\dd}}{\dd q^{\alpha}_{\sigma_1}}
\frac{\overrightarrow{\dd}}{\dd q^{\beta}_{\tau_1}}
g(\bx_2,[\bq],[\bq^{\dagger}]).
\end{align*}
Let the summation indexes be relabelled as above: $\alpha\rightleftarrows\beta$, $\sigma\rightleftarrows\tau$, and
$\smash{\delta s^{\alpha}_1\rightleftarrows\delta s^{\beta}_2}$, 
$\delta s^{\dagger}_{1,\alpha}\rightleftarrows\delta s^{\dagger}_{2,\beta}$ on top of $\bby\rightleftarrows\bz$.
The transformation of graded derivations falling from the left and right on $f$ is then
\begin{multline*}
\frac{\overrightarrow{\dd}}{\dd q^{\dagger}_{\alpha,\sigma_2}}
\left(f\frac{\overleftarrow{\dd}}{\dd q^{\dagger}_{\beta,\tau_2}}\right)\longmapsto
\frac{\overrightarrow{\dd}}{\dd q^{\dagger}_{\beta,\tau_2}}
\left(f\frac{\overleftarrow{\dd}}{\dd q^{\dagger}_{\alpha,\sigma_2}}\right)=
\frac{\overrightarrow{\dd}}{\dd q^{\dagger}_{\beta,\tau_2}}
\left((-)^{\GH(F)-1}\frac{\overrightarrow{\dd}}{\dd q^{\dagger}_{\alpha,\sigma_2}}f\right)=\\
=(-)^{\GH(F)-2}\cdot(-)^{\GH(F)-1}
\left(\frac{\overrightarrow{\dd}}{\dd q^{\dagger}_{\alpha,\sigma_2}}f\right)
\frac{\overleftarrow{\dd}}{\dd q^{\dagger}_{\beta,\tau_2}}=
-\frac{\overrightarrow{\dd}}{\dd q^{\dagger}_{\alpha,\sigma_2}}
\left(f\frac{\overleftarrow{\dd}}{\dd q^{\dagger}_{\beta,\tau_2}}\right).
\end{multline*}
This minus sign shows that the third term as it was written initially, and the newly produced one in which the 
variations $\delta\bolds_1$ and $\delta\bolds_2$ are interchanged have opposite signs. At the same time, these integral
functionals must be equal to each other due to independence of $\Delta$ and $\lshad\,,\,\rshad$ of a choice of the 
test shifts. Therefore, each of those expressions vanishes.

The fourth term is processed analogously; its integrand is
\begin{align*}
&(-)^{\GH(F)}\Biggl\langle(\delta s_1^{\alpha})
\left(\tfrac{\overleftarrow{\dd}}{\dd\bz_1}\right)^{\sigma_1}(\bz_1)\,
\overbrace{\vec{e}_{\alpha}(\bz_1),(-\vec{e}^{{}\,\dagger \alpha})(\bz_2)}^{-1}\,
(\delta s^{\dagger}_{1,\alpha})\left(\tfrac{\overleftarrow{\dd}}{\dd\bz_2}\right)^{\sigma_2}(\bz_2)\Biggr\rangle
\cdot\underbrace{\bigl\langle\vec{e}^{{}\,\dagger \alpha}(\bx_1),\vec{e}_{\alpha}(\bx_2)\bigr\rangle}_{-1}
\\
&\ {}\cdot
\frac{\overrightarrow{\dd}}{\dd q^{\alpha}_{\sigma_1}}
f(\bx_1,[\bq],[\bq^{\dagger}])\frac{\overleftarrow{\dd}}{\dd q^{\beta}_{\tau_1}}
\\
&\qquad
\underline{\bigl\langle\vec{e}^{{}\,\dagger \beta}(\bx_1)\bigr|}
\Biggl\langle
\left(\tfrac{\overrightarrow{\dd}}{\dd\bby_1}\right)^{\tau_1}(\delta s^{\beta}_2)(\bby_1)\,
\overbrace{\vec{e}_{\beta}(\bby_1),(-\vec{e}^{{}\,\dagger \beta})(\bby_2)}^{-1}\,
(\delta s^{\dagger}_{2,\beta})\left(\tfrac{\overleftarrow{\dd}}{\dd\bby_2}\right)^{\tau_2}
(\bby_2)\Biggr\rangle
\underbrace{\bigl|\vec{e}_{\beta}(\bx_2)\bigr\rangle}_{-1} 
\\ 
&\mbox{\hbox to 107mm {{ }\hfil { }}}
\frac{\overrightarrow{\dd}}{\dd q^{\dagger}_{\alpha,\sigma_2}}
\frac{\overrightarrow{\dd}}{\dd q^{\dagger}_{\beta,\tau_2}}
g(\bx_2,[\bq],[\bq^{\dagger}]).
\end{align*}
The very same procedure of two variations interchange and relabelling restores an almost identical expression in which,
however, the parity-odd derivations go in the inverse order
$\overrightarrow{\dd}/\dd q^{\dagger}_{\beta,\tau_2}\circ\overrightarrow{\dd}/\dd q^{\dagger}_{\alpha,\sigma_2}$.
Equal to minus itself, the fourth term vanishes. This concludes the proof.
\end{proof}

The following example illustrates the assertion of Lemma~\ref{LemmaBaseLapSchouten} (but not a technique of its proof which
itself accompanies Lemma~\ref{LAnyChoice}). We use the convention from Remark~\ref{RemKeepNotation}, denoting by $\Id/\Id\bby_i$ or $\Id/\Id\bz_j$ 
the total derivatives which act on the functionals' densities at points $\bx_k$; this keeps track of those derivatives
origin and lets us indicate the couplings' values as they appear after the integrations by parts, contributing only with
sign factors $\pm1$. For the sake of brevity we do not write the (co)vectors $\vec{e}_i$ and $\vec{e}^{{}\,\dagger\,i}$
in the formulas below, referring to the proofs in preceding sections. Likewise, we do not indicate the base point
congruences that occur due to the absolute locality of couplings.

An overall comment to Example~\ref{Countercounterexample} below is that, fully aware of the goal which is to calculate 
$\Delta\left(\lshad F,G\rshad\right)$ or, respectively, $\lshad\Delta F,G\rshad$ and $\lshad F,\Delta G\rshad$, we do not
interrupt the logic of our reasoning by attempting to view the intermediate objects $\lshad F,G\rshad$ or $\Delta F$ and
$\Delta G$ as mappings $\Gamma(\pi_{\BV})\to\Bbbk$, cf.\ Corollary~\ref{CorTowardsConventional} 
on p.~\pageref{CorTowardsConventional}. 
Such mappings would not be elements of the structures which stand in the left- and right-hand sides of the identity under 
examination. The slogan is that a step-by-step evaluation is illegal; 
derivations of the end-product from input data
must not be interrupted at half-way.

We also emphasize that the example below is a prototype reasoning which is equally well applicable to any other arguments
$F$ and $G$ in~\eqref{EqLapSchouten};
a choice of the functionals is here not specific to any model. The point is that
equality~\eqref{EqLapSchouten} holds and does not require any manual regularization.

\begin{example}\label{Countercounterexample}
Consider the integral functionals
$$F=\int q^{\dagger}qq_{x_1x_1}\,\Id x_1\quad\text{and}\quad G=\int q^{\dagger}_{x_2x_2}\cos q\,\Id x_2.$$
Let us show that equality~\eqref{EqLapSchouten} is satisfied for $F$ and $G$, that is,
\begin{equation}\label{EqNotHolds}
\Delta\left(\lshad F,G\rshad\right)=\lshad\Delta F,G\rshad+\lshad F,\Delta G\rshad,\qquad\GH(F)=1,
\end{equation}
in the frames of product-bundle geometry of variations and operational definitions of the BV-Laplacian $\Delta$ and
variational Schouten bracket $\lshad\,,\,\rshad$.

We have
\begin{multline*}
\lshad F,G\rshad = \iiiint \Id x_1\Id x_2\Id y_1\Id y_2
\Bigl\langle
\Bigl(\underbrace{q^\dagger q_{xx}+\tfrac{\Id^2}{\Id y_1^2}(q^\dagger q)}_{x_1}
\Bigr) \cdot \underbrace{\langle\delta s_2(y_1),\delta s_2^\dagger(y_2)\rangle}_{+1} \cdot \tfrac{\Id^2}{\Id y_2^2}
\bigl(\underbrace{\cos q}_{x_2}\bigr)
\Bigr\rangle \\
{}+ \iiiint \Id x_1\Id x_2\Id y_1\Id y_2 \Bigl\langle
\bigl(\underbrace{qq_{xx}}_{x_1}\bigr) \cdot \underbrace{\langle\delta s_2^\dagger(y_2),
\delta s_2(y_1)\rangle}_{-1} \cdot 
\bigl(\underbrace{-q^\dagger_{xx}\,\sin q}_{x_2}\bigr)\Bigr\rangle.
\end{multline*}
Therefore, one side of the expected equality is
\begin{multline*}
\Delta
\bigl(\lshad F,G\rshad\bigr)=\int\!\!\Id z_1\int\!\!\Id z_2\int\!\!\Id x_1\int\!\!\Id x_2\int\!\!\Id y_1\int\!\!\Id y_2
\underbrace{\langle\delta s_1(z_1),\delta s_1^\dagger(z_2)\rangle}_{+1}\cdot
\underbrace{\langle\delta s_2(y_1),\delta s_2^\dagger(y_2)\rangle}_{+1}\cdot{}\\
{}\cdot\Bigl\langle
\tfrac{\Id^2}{\Id z_1^2}\underbrace{(1)}_{x_1}\cdot\tfrac{\Id^2}{\Id y_2^2}\bigl(\underbrace{\cos q}_{x_2}\bigr) 
+ \underline{\underbrace{q_{xx}}_{x_1}\cdot\tfrac{\Id^2}{\Id y_2^2}\bigl(\underbrace{-\sin q}_{x_2}\bigr)} 
+ \tfrac{\Id^2}{\Id y_1^2}\underbrace{(1)}_{x_1}\cdot\tfrac{\Id^2}{\Id y_2^2}\bigl(\underbrace{\cos q}_{x_2}\bigr) 
+ \underline{\underline{\tfrac{\Id^2}{\Id y_1^2}\underbrace{(q)}_{x_1}\cdot\tfrac{\Id^2}{\Id y_2^2}
\bigl(\underbrace{-\sin q}_{x_2}\bigr)}} 
\Bigr\rangle\\
+
\int\!\!\Id z_1\int\!\!\Id z_2\int\!\!\Id x_1\int\!\!\Id x_2\int\!\!\Id y_1\int\!\!\Id y_2
\underbrace{\langle\delta s_1(z_1),\delta s_1^\dagger(z_2)\rangle}_{+1}\cdot
\underbrace{\langle\delta s_2^\dagger(y_2),\delta s_2(y_1)\rangle}_{-1}\cdot{}\\
\cdot\Bigl\langle
\underline{\underbrace{q_{xx}}_{x_1}\cdot\tfrac{\Id^2}{\Id z_2^2}\bigl(\underbrace{-\sin q}_{x_2}\bigr)}
+\underline{\underline{\tfrac{\Id^2}{\Id z_1^2}\underbrace{(q)}_{x_1}\cdot\tfrac{\Id^2}{\Id z_2^2}
\bigl(\underbrace{-\sin q}_{x_2}\bigr)}}
+\bigl(\underbrace{qq_{xx}}_{x_1}\bigr)\cdot\tfrac{\Id^2}{\Id z_2^2}\bigl(\underbrace{-\cos q}_{x_2}\bigr)
\Bigr\rangle\ .
\end{multline*}
The respective pairs of underlined terms cancel out and there remains only
\begin{multline}\label{EqLeft}
{}=\int{\cdots}\int\Id z_1\,\Id z_2\,\Id x_1\,\Id x_2\,\Id y_1\,\Id y_2
\underbrace{\langle\delta s_1(z_1),\delta s_1^\dagger(z_2)\rangle}_{+1}\cdot
\bigl\langle\bigl(\underbrace{qq_{xx}}_{x_1}\bigr)\cdot\tfrac{\Id^2}{\Id z_2^2}
\bigl(\underbrace{-\cos q}_{x_2}\bigr)\bigr\rangle\cdot\\
\underbrace{\langle\delta s_2^\dagger(y_2),\delta s_2(y_1)\rangle}_{-1}.
\end{multline}
On the other hand, we obtain that
\[
\Delta F=\iiint\Id z_1\Id z_2\Id x_1 
\underbrace{\langle\delta s_1(z_1),\delta s_1^\dagger(z_2)\rangle}_{+1}\cdot
\bigl\langle\underbrace{q_{xx}}_{x_1} + \tfrac{\Id^2}{\Id z_1^2}\underbrace{(q)}_{x_1}\bigr\rangle,
\]
which yields
\begin{multline*}
\lshad\Delta F,G\rshad=\int\!\!\Id z_1\int\!\!\Id z_2\int\!\!\Id x_1\int\!\!\Id x_2\int\!\!\Id y_1\int\!\!\Id y_2\,
\underbrace{\langle\delta s_1(z_1),\delta s_1^\dagger(z_2)\rangle}_{+1}\cdot{}\\
\Bigl\langle
\Bigl(\tfrac{\Id^2}{\Id y_1^2}\underbrace{(1)}_{x_1} +
\tfrac{\Id^2}{\Id z_1^2}\underbrace{(1)}_{x_1}\Bigr)\cdot
\tfrac{\Id^2}{\Id y_2^2}\bigl(\underbrace{\cos q}_{x_2}\bigr)
\Bigr\rangle\cdot
\underbrace{\langle\delta s_2(y_1),\delta s_2^\dagger(y_2)\rangle}_{+1}=0.
\end{multline*}
From the fact that the other BV-\/Laplacian,
\[
\Delta G=\iiint\,\Id z_1\,\Id z_2\,\Id x_2 
\underbrace{\langle\delta s_1(z_1),\delta s_1^\dagger(z_2)\rangle}_{+1}\cdot
\bigl\langle 
\tfrac{\Id^2}{\Id z_2^2}\bigl(\underbrace{-\sin q}_{x_2}\bigr)\bigr\rangle,
\]
does not contain~$q^\dagger$ so that the first half of the Schouten bracket~$\lshad F,\Delta G\rshad$ drops out,
we deduce that
\begin{multline}
\lshad F,\Delta G\rshad=\int{\cdots}\int\,\Id z_1\,\Id z_2\,\Id x_1\,\Id x_2\,\Id y_1\,\Id y_2
\underbrace{\langle\delta s_1(z_1),\delta s_1^\dagger(z_2)\rangle}_{+1}\cdot{}\\
\bigl\langle\bigl(\underbrace{qq_{xx}}_{x_1}\bigr)\cdot\tfrac{\Id^2}{\Id z_2^2}
\bigl(\underbrace{-\cos q}_{x_2}\bigr)\bigr\rangle\cdot
\underbrace{\langle\delta s_2^\dagger(y_2),\delta s_2(y_1)\rangle}_{-1}.
\label{EqRight}
\end{multline}
Consequently, the two sides of~\eqref{EqNotHolds}, namely, $\Delta\bigl(\lshad F,G\rshad\bigr)$ expressed by~\eqref{EqLeft}
and $\lshad\Delta F,G\rshad+\lshad F,\Delta G\rshad$ accumulated in~\eqref{EqRight}, match perfectly for the functionals~$F$
and~$G$ at~hand.\label{pEndCountercounterexample}
\end{example}

\begin{theor}\label{ThLaplaceOnSchouten}
Let $F$,\ $G\in\overline{\mathfrak{N}}^n(\pi_\BV,T\pi_{\BV})$ be two functionals. 
The Batalin\/--\/Vilkovisky Laplacian~$\Delta$ satisfies the relation
\begin{equation}
\Delta\bigl(\schouten{F,G}\bigr) = \schouten{\Delta F, G} + (-)^{\GH(F)-1}\schouten{F,\Delta G}.
\tag{\ref{EqLapSchouten}}
\end{equation}
\end{theor}

\noindent%
In other words, the operator $\Delta$ is a graded derivation of the variational Schouten bracket~$\schouten{\,,\,}$.

\begin{proof}
We prove this by induction over the number of building blocks in each argument of the Schouten bracket in the left 
hand side of~\eqref{EqLapSchouten}. If $F$ and $G$ both belong 
to $\overline{H}^*(\pi_\BV\times T\pi_\BV\times\ldots\times T\pi_\BV)$, then Lemma~\ref{LemmaBaseLapSchouten} 
states the assertion, which is the base of induction. 
To make an inductive step, without loss of generality let us assume that the second argument 
of $\schouten{\,,\,}$ in~\eqref{EqLapSchouten} is a product of two elements 
from~$\overline{\mathfrak{N}}^n(\pi_\BV,T\pi_{\BV})$, 
each of them containing less multiples from~$\overline{H}^*(\pi_\BV\times T\pi_\BV\times\ldots T\pi_\BV)$ 
than the product. Denote such factors by $G$ and $H$ and recall that by Theorem~\ref{ThSchoutenOnProduct},
\[
\schouten{F,G\cdot H} = \schouten{F,G}\cdot H+ (-)^{(\GH(F)-1)\cdot\GH(G)}G\cdot\schouten{F,H}.
\]
Therefore, using Theorem~\ref{ThLaplaceOnProduct} we have that
\begin{align}
\Delta(\lshad &F,G\cdot H\rshad)\nonumber\\
={}& \Delta(\schouten{F,G})\cdot H + (-)^{\GH(F)+\GH(G)-1}[\![\schouten{F,G},H]\!] 
+ (-)^{\GH(F)+\GH(G)-1}\schouten{F,G}\cdot\Delta H\nonumber\\
& + (-)^{(\GH(F)-1)\GH(G)}\left(\Delta G\cdot\schouten{F,H} + (-)^{\GH(G)}[\![ G,\schouten{F,H}]\!] 
+ (-)^{\GH(G)}G\cdot\Delta(\schouten{F,H})\right).\nonumber
\intertext{Using the inductive hypothesis in the first and last terms of the right\/-\/hand side in the above formula, 
we continue the equality and obtain}
={}&\schouten{\Delta F,G}\cdot H + (-)^{\GH(F)-1}\schouten{F,\Delta G}\cdot H 
+ (-)^{\GH(F)+\GH(G)-1}[\![\schouten{F,G},H]\!] \nonumber\\
&+ (-)^{\GH(F)\GH(G)}[\![ G,\schouten{F,H}]\!] + (-)^{\GH(F)+\GH(G)-1}\schouten{F,G}\cdot\Delta H \nonumber\\
&+ (-)^{(\GH(F)-1)\GH(G)}\Delta G\cdot\schouten{F,H} + (-)^{\GH(F)\GH(G)}G\cdot\schouten{\Delta F,H} \nonumber\\
& + (-)^{\GH(F)\GH(G)+\GH(F)-1}G\cdot\schouten{F,\Delta H}.\label{EqIndStep}
\end{align}
On the other hand, let us expand the formula
\[
\schouten{\Delta F, G\cdot H} + (-)^{\GH(F)-1}\schouten{F,\Delta(G\cdot H)},
\]
which is the right hand side of~\eqref{EqLapSchouten} in the inductive claim. We obtain
\begin{align}
={}& \schouten{\Delta F, G}\cdot H + (-)^{(\GH(\Delta F)-1)\GH(G)}G\cdot\schouten{\Delta F,H} \nonumber\\
&+ (-)^{\GH(F)-1}[\![ F,\ \Delta G\cdot H  + (-)^{\GH(G)}\schouten{G,H} + (-)^{\GH(G)}G\cdot\Delta H\ ]\!] \nonumber\\
={}& \schouten{\Delta F, G}\cdot H + (-)^{\GH(F)\GH(G)}G\cdot\schouten{\Delta F,H} 
+ (-)^{\GH(F)-1}\schouten{F,\Delta G}\cdot H 
  \label{EqClaimExpand}\\
&+ (-)^{\GH(F)-1}(-)^{(\GH(F)-1)(\GH(G)-1)}\Delta G\cdot\schouten{F,H} 
+ (-)^{\GH(F)-1}(-)^{\GH(G)}[\![ F,\schouten{G,H}]\!] \nonumber\\
&+ (-)^{\GH(F)-1}(-)^{\GH(G)}\schouten{F,G}\cdot\Delta H 
+ (-)^{\GH(F)-1}(-)^{\GH(G)}(-)^{(\GH(F)-1)\GH(G)}G\cdot\schouten{F,\Delta H}. \nonumber
\end{align}
Comparing~\eqref{EqClaimExpand} with~\eqref{EqIndStep}, which was derived from the inductive hypothesis, 
we see that all terms match except for
\[
(-)^{\GH(F)+\GH(G)-1}[\![\schouten{F,G},H]\!] + (-)^{\GH(F)\GH(G)}[\![ G,\schouten{F,H}]\!]
\]
from~\eqref{EqIndStep} versus
\[
(-)^{\GH(F)+\GH(G)-1}[\![ F,\schouten{G,H}]\!]
\]
from~\eqref{EqClaimExpand}. However, these three terms constitute Jacobi's identity~\eqref{EqJacobiSchouten} 
for the variational Schouten bracket. Namely, we have that (cf.~\cite{Lorentz12})
\begin{equation}\label{EqJacobiLeibniz}
[\![ F,\schouten{G,H}]\!] = [\![\schouten{F,G},H]\!] + (-)^{(\GH(F)-1)(\GH(G)-1)}[\![ G,\schouten{F,H}]\!],
\end{equation}
so that by multiplying both sides of the identity by $(-)^{\GH(F)+\GH(G)-1}$, 
we fully balance~\eqref{EqIndStep} and~\eqref{EqClaimExpand}. This completes the inductive step and concludes the proof.
\end{proof}

\begin{lemma}\label{LBVOperatorDifferential}
The linear operator
\[
\Delta\colon\overline{H}^{n(1+k)}\bigl(\pi_\BV\times T\pi_\BV\times\ldots\times 
T\pi_\BV\bigr)\longrightarrow\overline{H}^{n(2+k)}\bigl(\pi_\BV\times T\pi_\BV\times\ldots\times T\pi_\BV\bigr)
\]
is a differential for every $k\geqslant0$.
\end{lemma}

The proof of Lemma~\ref{LBVOperatorDifferential} is conceptually close to the second and third steps in the proof
of Lemma~\ref{LemmaBaseLapSchouten}. Namely, two normalized variations are swapped in an integral functional within
the image of $\Delta^2$, which yields an indistinguishable result of opposite sign.

\begin{proof}
Let $\delta\bolds_1$ and $\delta\bolds_2$ be normalized test shifts of a section $s\in\Gamma(\pi_{\BV})$, and let 
$H=\int h(\bx,[\bq],[\bq^{\dagger}])\cdot\dvol(\bx)$ be an integral functional. (It suffices to consider a simplified
picture $H\in\ov{H}^n(\pi_{\BV})$, not taking into account any built-in variations in the construction of $H$.)
By definition, we have that
\begin{multline*}
\Delta(\Delta H)(s)=\int_M\Id\bz_1\int_M\Id\bz_2\int_M\Id\bby_1\int_M\Id\bby_2\int_M\dvol(\bx)\cdot\\
\cdot
\Biggl\{
\left\langle
(\delta s^{\alpha}_1)
\left(\tfrac{\overleftarrow{\dd}}{\dd\bz_1}\right)^{\sigma_1}(\bz_1)\,
\overbrace{\vec{e}_{\alpha}(\bz_1),(-\vec{e}^{{}\,\dagger \alpha})(\bz_2)}^{-1}\,
(\delta s^{\dagger}_{1,\alpha})
\left(\tfrac{\overleftarrow{\dd}}{\dd\bz_2}\right)^{\sigma_2}(\bz_2)
\right\rangle
\underbrace{\left\langle
\vec{e}^{{}\,\dagger \alpha}(\bx),\vec{e}_{\alpha}(\bx)\right\rangle}_{-1}\\
{}\quad\left\langle
(\delta s^{\beta}_2)
\left(\tfrac{\overleftarrow{\dd}}{\dd\bby_1}\right)^{\tau_1}(\bby_1)\,
\overbrace{\vec{e}_{\beta}(\bby_1),(-\vec{e}^{{}\,\dagger \beta})(\bby_2)}^{-1}\,
(\delta s^{\dagger}_{2,\beta})
\left(\tfrac{\overleftarrow{\dd}}{\dd\bby_2}\right)^{\tau_2}(\bby_2)
\right\rangle
\underbrace{\left\langle
\vec{e}^{{}\,\dagger \beta}(\bx),\vec{e}_{\beta}(\bx)\right\rangle}_{-1}\\
\frac{\overrightarrow{\dd}}{\dd q^{\alpha}_{\sigma_1}}
\frac{\overrightarrow{\dd}}{\dd q^{\dagger}_{\alpha,\sigma_2}}
\frac{\overrightarrow{\dd}}{\dd q^{\beta}_{\tau_1}}
\frac{\overrightarrow{\dd}}{\dd q^{\dagger}_{\beta,\tau_2}}
h(\bx,[\bq],[\bq^{\dagger}])
\left.\Biggr\}\right|_{j^{\infty}_{\bx}(s)}.
\end{multline*}
By exchanging the integrand's upper two lines and then relabelling $\alpha\rightleftarrows\beta$,

$\sigma\rightleftarrows\tau$ so that $\delta s_1^{\alpha}\rightleftarrows\delta s_2^{\beta}$ and
$\delta s^{\dagger}_{1,\alpha}\rightleftarrows\delta s^{\dagger}_{2,\beta}$, and by swapping the reference
$\bby\rightleftarrows\bz$ to copies of the base manifold $M^n$, we almost recover the initial expression (which should
be the case), yet the order in which the parity-odd partial derivatives follow is inverse,
$$
\frac{\overrightarrow{\dd}}{\dd q^{\dagger}_{\alpha,\sigma_2}}
\circ\frac{\overrightarrow{\dd}}{\dd q^{\dagger}_{\beta,\tau_2}}
\longmapsto
\frac{\overrightarrow{\dd}}{\dd q^{\dagger}_{\beta,\tau_2}}
\circ\frac{\overrightarrow{\dd}}{\dd q^{\dagger}_{\alpha,\sigma_2}}
=-\frac{\overrightarrow{\dd}}{\dd q^{\dagger}_{\alpha,\sigma_2}}
\circ\frac{\overrightarrow{\dd}}{\dd q^{\dagger}_{\beta,\tau_2}}.
$$
Therefore the integrand of functional $\Delta^2H$ vanishes, which proves the assertion.
\end{proof}

\begin{theor}\label{ThBVDifferential}
The Batalin\/--\/Vilkovisky Laplacian~$\Delta$ is a differential\textup{:} 
for all $H \in \overline{\mathfrak{N}}^n(\pi_\BV$,\ $T\pi_{\BV})$ we have 
\[
\Delta^2(H) = 0.
\]
\end{theor}

\begin{proof}
We prove Theorem~\ref{ThBVDifferential} by induction over the number of building blocks 
from~$\overline{H}^*\bigl(\pi_\BV\times T\pi_\BV\times\ldots\times T\pi_\BV\bigr)$ 
in the argument $H \in \overline{\mathfrak{N}}^n(\pi_\BV,T\pi_{\BV})$ of~$\Delta^2$. 
If $H \in \overline{H}^*\bigl(\pi_\BV\times T\pi_\BV\times\ldots\times T\pi_\BV\bigr)$ itself is an integral functional, 
then by Lemma~\ref{LBVOperatorDifferential}
there remains nothing to prove. Suppose now that $H = F\cdot G$ for 
some $F,G \in \overline{\mathfrak{N}}^n(\pi_\BV,T\pi_{\BV})$. Then Theorem~\ref{ThLaplaceOnProduct} yields that
\begin{align*}
\Delta^2&(F\cdot G) = \Delta\left(\Delta F\cdot G + (-)^{\GH(F)}\schouten{F,G} + 
 (-)^{\GH(F)}F\cdot\Delta G\right).
\intertext{Using Theorem~\ref{ThLaplaceOnProduct} again and also Theorem~\ref{ThLaplaceOnSchouten}, 
we continue the equality:}
={}& \Delta^2F\cdot G + (-)^{\GH(\Delta F)}\schouten{\Delta F,G} 
+ (-)^{\GH(\Delta F)}\Delta F\cdot\Delta G \\
{}&{}+(-)^{\GH(F)}\schouten{\Delta F,G} + (-)^{\GH(F)}(-)^{\GH(F)-1}\schouten{F,\Delta G}
\\
{}&{}+(-)^{\GH(F)}\Delta F\cdot\Delta G 
+(-)^{\GH(F)}(-)^{\GH(F)}\schouten{F,\Delta G} + (-)^{\GH(F)}(-)^{\GH(F)}F\cdot\Delta^2G.
\end{align*}
By the inductive hypothesis, the first and last terms in the above formula vanish; 
taking into account that $\GH(\Delta F) = \GH(F)-1$ in $\BBZ_2$, the terms with $\Delta F\cdot\Delta G$ 
cancel against each other, as do the terms containing $\schouten{\Delta F,G}$ and $\schouten{F,\Delta G}$. 
The proof is complete.
\end{proof}

\section{The quantum master\/-\/equation}\label{SecFeynman}
\subsection{The Laplace equation}\label{SecMaster}
In this section we inspect the conditions upon functionals $F \in \overline{\mathfrak{N}}^n(\pi_\BV,T\pi_\BV)$ under which the Feynman path integrals $\int_{\Gamma(\bzeta^0)}[Ds]\,F([s],[s^\dagger])$ are (infinitesimally) independent 
of the unphysical 
anti\/-\/objects $s^\dagger \in \Gamma( 
\boldsymbol{\zeta}^1 
)$. The derivation of such a condition (see equation~\eqref{EqLaplace} below) relies on an extra assumption of the translation invariance of a measure in the path integral. It must be noted, however, that we do not define Feynman's integral here and do not introduce that measure which essentially depends on the agreement about the classes of `admissible' sections~$\Gamma(\pi)$ or~ $\Gamma(\boldsymbol{\zeta}^{(0|1)}   
)$. Consequently, our reasoning is to some extent heuristic.

The basics of path integration, which we recall here for consistency, are standard: they illustrate how the geometry of the BV-\/Laplacian works in practice. We draw the experts' attention only to the fact that in our notation $\Psi$ is not the gauge fixing fermion $\boldsymbol\Psi$ such that the odd\/-\/component's section $s^\dagger\in\Gamma(\boldsymbol{\zeta}^1)$ 
is the restriction of $\delta\boldsymbol\Psi/\delta q$ 
to the jet of a section for~$\boldsymbol{\zeta}^0$\,;
instead, we let $\Psi$ determine the infinitesimal shift $\dot{q}^\dagger = \delta\Psi/\delta q$ of coordinates along the fibre's parity\/-\/odd half. 
We also note that the preservation of parity is not mandatory here and thus an even-parity $\Psi \in \overline{H}^n(\boldsymbol{\zeta}^0) \hookrightarrow \overline{H}^n(\pi_\BV)$ is a legitimate choice.

Let $F=\int f(\bx,\bq,\bq^\dagger)\,\dvol(\bx) 
\in \overline{H}^n(\pi_\BV)$ be a functional;
here and in what follows we proceed over the building blocks of 
elements from~$\overline{\mathfrak{N}}^n(\pi_\BV,T\pi_\BV)$
by the graded Leibniz rule.
Let $\Psi=\int\psi(\bby,\bq)\,\dvol(\bby)\in \overline{H}^n(\boldsymbol{\zeta}^0) \hookrightarrow \overline{H}^n(\pi_\BV)$ be an integral functional which, by assumption, is constant along 
ghost parity\/-\/odd variables: $\Psi(s^\alpha,s^\dagger_\beta) = \Psi(s^\alpha,t^\dagger_\beta)$ for any sections $\{s^\alpha\} \in \Gamma(\bzeta^0)$ and $\{s^\dagger_\beta\},\{t^\dagger_\beta\} \in \Gamma(
\bzeta^1)$. We investigate under which conditions the path integral $\int_{\Gamma(\bzeta^0)}[Ds^\alpha]\,F(s^\alpha,s^\dagger_\beta) \colon \Gamma(
\bzeta^1
) \to \Bbbk$ is infinitesimally independent of a choice of the anti\/-\/objects:
\begin{align}\label{eq:PathIntegralIndependent}
\diftat{\varepsilon^\dagger}0\int_{\Gamma(\bzeta^0)}[Ds^\alpha]\,
F\Bigl(s^\alpha,s^\dagger_\beta + \varepsilon^\dagger\,\frac{\vec{\delta}\psi}
{\delta q^\beta}
\bigg|_{s^\alpha
}\Bigr) = 0 \quad\text{for all $s^\dagger \in \Gamma(\bzeta^1
)$.}
\end{align}
Note that this formula makes sense because the bundles $\bzeta^0$ and 
$\bzeta^1
$ are dual so that a variational covector in the geometry of $\bzeta^0$ acts as a shift vector in the geometry of $\bzeta^1
$.  
The left\/-\/hand side of~\eqref{eq:PathIntegralIndependent} equals
\[
\int_{\Gamma(\bzeta^0)}[Ds^\alpha]\int_M\dvol(\bx)\,\frac{\overrightarrow{\delta}\!\psi}
{\delta q^\beta}(\bx,\bq)\bigr|_{j^\infty_\bx(s^\alpha
)} \cdot
\frac{\overleftarrow{\delta}\!f}{\delta q^\dagger_\beta}(\bx,\bq,\bq^\dagger)\bigr|_{j^\infty_\bx(s^\alpha,s^\dagger_\gamma)}, \qquad s^\dagger\in\Gamma(\bzeta^1).
\]
Take any auxiliary section $\delta\bolds = (\delta s^\alpha, \delta s^\dagger_\beta) \in 
\Gamma\bigl(T\bzeta^{(0|1)}
)\bigr)$ normalized by $\delta s^\alpha(x)\cdot\delta 
s^\dagger_\alpha(x) \equiv 1$ at every $\bx\in M^n$ for each $\alpha = 1,\dots,m + m_1+\cdots+m_\lambda = N$ 
and blow up the scalar integrand to a pointwise contraction of dual object taking their values in the 
fibres $T_{(\bx,\phi(\bx),s(\bx))}V_\bx$ and $T_{(\bx,\phi(\bx),s^\dagger(\bx))}\Pi V^\dagger_{\bx}$ of $T(\pi_\BV
)$ 
over $\phi(x)
$: for $s = (s^\alpha, s^\dagger_\beta)$ we have
\begin{multline*}
\int_M\dvol(\bx)
\,\left.\left(\frac{\overrightarrow{\delta}\!\psi}{\delta q^\alpha}
\cdot
\frac{\overleftarrow{\delta}\!f}{\delta q^\dagger_\alpha}\right)\right|
_{j^\infty_\bx(s)}\\
{}= 
\int_M\dvol(\bx_1)\int_M\dvol(\bx_2)\int_M\Id\bby_1\int_M\Id\bby_2
\left.\left(\psi(\bx_1,\bq)\frac{\overleftarrow{\dd}}{\dd q^{j_1}_{\sigma_1}}\right)\right|_{j^\infty_{\bx_1}(s)}\mbox{\hbox to 20mm {{ }\hfil { }}} \\
{}\cdot\underline{\langle\vec{e}^{\,\dagger j_1}(\bx_1)|}\,
\bigl\langle 
\left(\tfrac{\overleftarrow{\dd}}{\dd\bby_1}\right)^{\sigma_1} \delta s^{i_1}(\bby_1)\,
\vec{e}_{i_1}(\bby_1), -\vec{e}^{\,\dagger i_2}(\bby_2)\,
\delta s^\dagger_{i_2}(\bby_2)\left(\tfrac{\overrightarrow{\dd}}{\dd\bby_2}\right)^{\sigma_2}\bigr\rangle\,
\underline{|\vec{e}_{j_2}(\bx_2)\rangle} \\
{}\cdot \left.\left(\frac{\overrightarrow{\dd}}{\dd q^\dagger_{j_2,\sigma_2}} f(\bx_2,\bq,\bq^\dagger)\right)\right|_{j^\infty_{\bx_2}(s)}.
\end{multline*}
In fact, the integrand refers to a definition of the evolutionary vector field $\bQ^\Psi$ such that 
$\bQ^\Psi(F) \cong \schouten{\Psi, F}$ modulo integration by parts in the building blocks of $F$, 
cf.~\cite{Lorentz12}. Due to a special choice of the dependence of~$\Psi$ on~$s$ only, this is indeed the Schouten bracket~$\schouten{\Psi, F}$.

To rephrase the indifference of the path integral to a choice of $\Psi$ in terms of an equation upon the functional $F$ 
alone, we perform integration by parts in Feynman's integral. For this we employ the translation 
invariance $[Ds] = [D(s-\mu\cdot\delta s)]$ of the functional measure.

\begin{lemma}\label{LTI}
Let $H=\int h(\bx,\bq,\bq^\dagger)\,\dvol(\bx) \in \overline{H}^n(\pi_\BV)\subset
\overline{\mathfrak{N}}^n(\pi_\BV,T\pi_\BV)$ be an integral functional and $\delta s \in \Gamma(T\bzeta^0) \hookrightarrow \Gamma\bigl(T\bzeta^{(0|1)}\bigr)$ be a shift. 
Then we have that
\[
\int_{\Gamma(\bzeta^0)}[Ds^\alpha]\int_M\dvol(\bx)\,\delta s^\nu(\bx)\cdot
\frac{\overleftarrow{\delta}\!h}{\delta q^\nu}\bigg|_{j^\infty_\bx(s^\alpha,s^\dagger_\beta)} = 0,
\]
where the section $s^\dagger \in \Gamma(\bzeta^1)$ is a parameter.
\end{lemma}

\begin{proof}
Indeed,
\begin{align*}
0 &= \diftat\mu0\int_{\Gamma(\bzeta^0)}[Ds^\alpha]\,H(s^\alpha,s^\dagger_\beta),
\intertext{because the integral contains no parameter $\mu \in \Bbbk$. We continue the equality:}
&= \diftat\mu0\int_{\Gamma(\bzeta^0)}[D(s^\alpha-\mu\,\delta s^\alpha)]\,H(s^\alpha,s^\dagger_\beta) \\
&= \diftat\mu0\int_{\Gamma(\bzeta^0)}[Ds^\alpha]\,H(s^\alpha + \mu\,\delta s^\alpha,s^\dagger_\beta)
 = \int_{\Gamma(\bzeta^0)}[Ds^\alpha]\diftat\mu0H(s^\alpha + \mu\,\delta s^\alpha,s^\dagger_\beta),
\end{align*}
which yields the helpful formula in the lemma's assertion.
\end{proof}

Returning to the functionals $\Psi$ and $F$ and denoting $G(s) := \diftat\ell0 F(s + \ell\cdot\overleftarrow{\delta s}^\dagger
)$, we use the Leibniz rule for the derivative of $H = \Psi\cdot G$:
\[
\diftat\mu0(\Psi\cdot G)(s + \mu\cdot\overleftarrow{\delta s}
)=
\diftat\mu0(\Psi)(s + \mu\cdot\overleftarrow{\delta s})\cdot G(s)+
\Psi(s)\cdot\diftat\mu0(G)(s + \mu\cdot\overleftarrow{\delta s}).
\]
Because the path integral over $[Ds^\alpha]$ of the entire expression vanishes by Lemma~\ref{LTI} in which we were ready to proceed by the Leibniz rule over building blocks, 
we infer that the path integrals of the two terms are opposite. 
Now take the traces over indexes in both variations. 
The integral of the first term equals the initial expression for the path integral containing~$F$, i.\,e., 
the left\/-\/hand side of equation~\eqref{eq:PathIntegralIndependent}. 
Consequently, if
\begin{align}\label{ObserveInPath}
\int_{\Gamma(\bzeta^0)}[Ds^\alpha]\,\Psi(s^\alpha)\cdot\Delta F(s^\alpha,s^\dagger_\beta) = 0
\end{align}
for $\{s^\dagger_\beta\} \in \Gamma(\bzeta^1
)$ and for all $\Psi \in \overline{H}^n(\bzeta^0) 
\hookrightarrow \overline{H}^n(\pi_\BV)$, then the path integral over $F$ is infinitesimally independent of a section 
$\{s^\dagger_\beta\} \in \Gamma(\bzeta^1
)$.

The condition
\begin{equation}\label{EqLaplace}
\Delta F = 0
\end{equation}
is sufficient for equation~\eqref{ObserveInPath}, and therefore equation~\eqref{eq:PathIntegralIndependent}, to hold. By specifying a class $\Gamma(\pi_\BV
)$ of admissible sections of the BV-bundle for a concrete field model, and endowing that space of sections with a suitable metric, one could reinstate a path integral analogue of the main lemma in the calculus of variations and then argue that the condition $\Delta F = 0$ is also necessary.

Summarizing, whenever equation~\eqref{EqLaplace} holds, one can assign arbitrary admissible values to the odd\/-\/parity coordinates; for example, one can let 
$s^\dagger_\beta(\bx) = 
\delta\boldsymbol{\psi}/\delta q^\beta\bigr|_{j^\infty_\bx(s^\alpha)}
$ for a gauge\/-\/fixing integral $\boldsymbol{\Psi}=\int\boldsymbol{\psi}(\bx,\bq)\,\dvol(\bx) \in \overline{H}^n(\bzeta^0)$. This choice is reminiscent of the substitution principle, see~\cite{Lorentz12} and~\cite{Olver}.

Laplace's equation~\eqref{EqLaplace} ensures the infinitesimal independence from non-physical anti\/-\/objects 
for path integrals of functionals over physical fields -- not only in the classical BV-geometry of the 
bundle $\pi_\BV
$, but also in the quantum setup, whenever all objects are tensored 
with formal power series $\Bbbk[[\hbar,\hbar^{-1}]]$ in the Planck constant $\hbar$. 
It is accepted that each quantum field $s^\hbar$ contributes to the expectation value of a functional~$\mathcal{O}^\hbar$ 
with the factor~$\exp({\boldi}S^\hbar(s^\hbar)/{\hbar})$, where~$S^\hbar$ is the quantum BV-\/action of the model. 
Solutions $\mathcal{O}^\hbar$ of the equation $\Delta\bigl(\mathcal{O}^\hbar\cdot\exp({\boldi}S^\hbar/{\hbar})\bigr) = 0$ are 
the~\textsl{observables}. In particular, the postulate that the unit $1 \colon s^\hbar \mapsto 1 \in \Bbbk$ is averaged 
to unit by the Feynman integral of $1\cdot\exp({\boldi}S^\hbar(s^\hbar)/{\hbar})$ over the space of quantum 
fields $s^\hbar$ normalizes the integration measure and constrains the quantum BV-action by the quantum 
master\/-\/equation (see, e.g.,~\cite{BV,BRST,GitmanTyutin,HenneauxTeitelboim,Zinn-Justin:CriticalPhenomena}). 

\begin{proposition}\label{thm:QME}
Let $S^\hbar$ be the even quantum BV-action \textup{(}i.\,e., let it have a density that has an even number of ghost parity\/-\/odd 
coordinates in each of its terms\textup{)}. If the identity 
\[
\Delta\left(\exp\bigl(\tfrac{\boldi}{\hbar} S^\hbar\bigr)\right) = 0
\]
holds, then $S^\hbar$ satisfies the quantum master\/-\/equation\textup{:}
\begin{equation}\label{QME}
\boldi\hbar\,\Delta S^\hbar = \tfrac12\schouten{S^\hbar,S^\hbar}.
\end{equation}
\end{proposition}

\begin{proposition}\label{PropPrepareGaugeQME}
If an even functional $\mathcal{O}$ and the quantum BV-\/action $S^\hbar$ are such that 
$\Delta\bigl(\mathcal{O}\exp(\boldi S^\hbar/\hbar)\bigr) = 0$ and $\Delta\bigl(\exp(\boldi S^\hbar/\hbar)\bigr) = 0$ hold, 
respectively, then $\mathcal{O}$ satisfies
\begin{equation}\label{eq:def-Omega}
\Omega^\hbar(\mathcal{O}) \mathrel{{:}{=}} - \boldi\hbar\, \Delta \mathcal{O} +
\schouten{S^\hbar, \mathcal{O}}  = 0.
\end{equation}
\end{proposition}

\noindent%
We quote the standard proofs of Propositions~\ref{thm:QME} 
and~\ref{PropPrepareGaugeQME}
from~\cite{Laplace13} in
~\ref{AppProveQME}~---
yet now we gain a deeper insight on a construction of the quantum BV-\/differential~$\Omega^\hbar$.

\begin{rem}\label{EvenGhostParity}
A practical way 
to fix 
the signs which arise in the BV-\/Laplacian and Schouten bracket
from the ghost parity and a grading in the case of a superbundle
$\pi\colon E^{(m_0+n_0|m_1+n_1)}\to M^{(n_0|n_1)}$ of superfields
is by a re-derivation of the Laplace equation $\Delta(\cO\exp(\tfrac{\boldi}{\hbar}S^{\hbar}))=0$ upon 
an observable $\cO$ starting from the Schwinger\/--\/Dyson condition,
\begin{equation}\label{EqSchwDy}
\vec{\dd}^{\,(\bq^{\dagger})}_{\vec{\delta}\Psi/\delta\bq}\left(\int[D\bq]\,
\cO([\bq],[\bq^{\dagger}])
\exp\left(\tfrac{\boldi}{\hbar}S^{\hbar}\bigl([\bq],[\bq^{\dagger}]\bigr)\right)\right)=0,
\end{equation}
which postulates the Feynman path integral's independence of the non\/-\/physical BV-\/coordinates~$\bq^{\dagger}$ with odd ghost parity.
Note that the measure in the path integral involves only ghost parity-even objects (whatever be their $\BBZ_2$-\/grading).
\end{rem}

\begin{theor}\label{PropQBVDifferential}
Let $\mathcal{O} \in \overline{\mathfrak{N}}^n(\pi_\BV,T\pi_\BV)$ be a
functional and let the even functional 
$S^\hbar\in \overline{H}^n(\pi_\BV)$ satisfy quantum master\/-\/equation~\eqref{QME}. 
Then the operator $\Omega^\hbar$, defined in~\eqref{eq:def-Omega}, squares to zero\textup{:} 
\[ {(\Omega^\hbar)}^2(\mathcal{O}) = 0. \]
\end{theor}

\begin{proof}
We calculate, using Theorem~\ref{ThLaplaceOnSchouten},
\begin{align*}
{(\Omega^\hbar)}^2(\mathcal{O}) 
&= \schouten{S^\hbar,\schouten{S^\hbar,\mathcal{O}} - \boldi\hbar\,\Delta \mathcal{O}} - \boldi\hbar\,\Delta\big(\schouten{S^\hbar,\mathcal{O}} - \boldi\hbar\,\Delta \mathcal{O}\big) \\
&= \schouten{S^\hbar,\schouten{S^\hbar,\mathcal{O}}} - \boldi\hbar\,\schouten{S^\hbar,\Delta \mathcal{O}} - \boldi\hbar\,\schouten{\Delta S^\hbar,\mathcal{O}} + \boldi\hbar\,\schouten{S^\hbar,\Delta \mathcal{O}} + (\boldi\hbar)^2\Delta^2\mathcal{O}.\\
\intertext{The last term vanishes identically by Theorem~\ref{ThBVDifferential}, while the second term cancels against the fourth term. Using Jacobi's identity~\eqref{EqJacobiSchouten} for the Schouten bracket on the first term, we obtain:}
{(\Omega^\hbar)}^2(\mathcal{O}) 
&= - \boldi\hbar\,\schouten{\Delta S^\hbar, \mathcal{O}} +
{\textstyle\frac12}\schouten{\schouten{S^\hbar,S^\hbar},\mathcal{O}} 
= \schouten{-\boldi\hbar\,\Delta S^\hbar+{\textstyle\frac12}\schouten{S^\hbar,S^\hbar}, \mathcal{O}}.\\ 
\intertext{
Now is the crucial moment in the entire proof. By the logic of our reasoning's objective, the theorem's claim is that the operator~$(\Omega^\hbar)^2$ yields zero whenever acting on a functional~$\mathcal{O}$. We accordingly transform the variational Schouten bracket of two terms to the operator realization,}
&\cong \vec{\bQ}^{-\boldi\hbar\,\Delta S^\hbar + \frac{1}{2}\lshad S^\hbar,S^\hbar\rshad} (\mathcal{O}),
\end{align*}
with the evolutionary derivation now acting on the argument. Let us emphasize that a transition from the variational Schouten bracket --\,which increases the number of bases $M\times\ldots\times M$ by construction\,-- to the evolutionary vector field chops a multiplication of geometries by uniquely fixing the field's generating section.\footnote{It might happen otherwise that
a co\/-\/multiple of~$\mathcal{O}$ under~$\lshad\,,\,\rshad$ looks like zero as a map of the space~$\Gamma(\pi_\BV)$ yet the bracket with it could still be nonzero, see, e.\,g., $\Delta G$ on p.~\pageref{pEndCountercounterexample} in Example~\ref{Countercounterexample}.}
But by our initial assumption, this generating section is zero by virtue of~\eqref{QME}. Therefore the image of~$\mathcal{O}$ under such map vanishes, which proves the assertion.
\end{proof}


\subsection{Gauge automorphisms of quantum BV-\/cohomology groups}\label{SecGauge}
By using the quantum BV-differential $\Omega^{\hbar}$, let us construct a closed algebra of infinitesimal gauge symmetries
for the quantum master-equation~\eqref{QME}.

\begin{proposition}\label{PropGaugeQME}
Let $F\in\ov{\gN^n}(\pi_\BV,T\pi_\BV)$ be an arbitrary odd-parity functional and $S^{\hbar}$ the quantum master-action
satisfying~\eqref{QME}. Then the infinitesimal shift of the functional~$S^{\hbar}$,
\begin{equation}\label{EqQMEGaugeSym}
\dot S^{\hbar}=\Omega^{\hbar}(F)\quad\Longleftrightarrow\quad S^{\hbar}\mapsto 
S^{\hbar}(\veps)=S^{\hbar}+\veps\cdot\Omega^{\hbar}(F)+\ov{o}(\veps),\ \veps\in\BBR,
\end{equation}
is a symmetry of~\eqref{QME} so that $\Delta\left(\exp\left(\tfrac{\boldi}{\hbar}S^{\hbar}(\veps)\right)\right)=\ov{o}(\veps)$
in Peano's notation.

\noindent$\bullet$ The algebra of infinitesimal gauge symmetries~\eqref{EqQMEGaugeSym} of the quantum master-equation is closed,
\begin{equation}\label{EqCommutQMEGauge}
\left.\left(\frac{\Id}{\Id\veps_1}\circ\frac{\Id}{\Id\veps_2}-
\frac{\Id}{\Id\veps_2}\circ\frac{\Id}{\Id\veps_1}\right)\right|_{\veps_i=0}
(S^{\hbar})=\Omega^{\hbar}\bigl(\lshad F_1,F_2\rshad\bigr),
\end{equation}
i.e., the commutator of two even-parity symmetries with respective generators $F_1$ and~$F_2$ is the infinitesimal gauge
symmetry whose generator is the odd Poisson bracket of $F_1$ and~$F_2$.
\end{proposition}

\begin{rem}
The odd-parity generators $F_i\in\ov{\gN^n}(\pi_\BV,T\pi_\BV)$ never evolve in the course of a transformation which is induced
by any generator $F_j$ on the quantum BV-action functional $S^{\hbar}$.
\end{rem}

\begin{proof}
Assuming a smooth dependence of $S^{\hbar}(\veps)$ on $\veps$, we obtain that\footnote{This 
proof is standard: it originates from the cohomological deformation theory for solutions
of the Maurer\/--\/Cartan equation (e.~g., of~\eqref{QME}), see~\cite{KontsevichSoibelman}.}
\begin{equation*}
\frac{\Id}{\Id\veps}\Delta\left(\exp\left(\tfrac{\boldi}{\hbar}S^{\hbar}(\veps)\right)\right)=
\tfrac{\boldi}{\hbar}\dot S^{\hbar}\cdot\Delta\left(\exp\left(\tfrac{\boldi}{\hbar}S^{\hbar}(\veps)\right)\right)+
\Omega^{\hbar}(\dot S^{\hbar})\cdot\exp\left(\tfrac{\boldi}{\hbar}S^{\hbar}(\veps)\right).
\end{equation*}
Because $(\Omega^{\hbar})^2=0$ by Theorem~\ref{PropQBVDifferential}, for $\dot S^{\hbar}$ to be an infinitesimal symmetry of the equation
$\Delta\left(\exp\left(\tfrac{\boldi}{\hbar}S^{\hbar}\right)\right)=0$ 
it is sufficient that $S^{\hbar}=\Omega^{\hbar}(F)$ for some odd\/-\/parity functional~$F$.

Second, let
$$
\frac{\Id}{\Id\veps_i}(S^{\hbar})=-\boldi\hbar\,\Delta F_i+\lshad S^\hbar,F_i\rshad\qquad\text{for }i=1,2,\qquad\veps_i\in\BBR,$$
and postulate that $\frac{\Id}{\Id\veps_i}(F_j)\equiv0$ for all $i$ and $j$. Then commutator~\eqref{EqCommutQMEGauge} 
of even-parity infinitesimal transformations~\eqref{EqQMEGaugeSym} generated by the functionals $F_1$ and $F_2$ is
\begin{multline*}
\lshad-\boldi\hbar\,\Delta F_1+\lshad S^\hbar,F_1\rshad,F_2\rshad
-\lshad-\boldi\hbar\,\Delta F_2+\lshad S^\hbar,F_2\rshad,F_1\rshad\\
{}=-\boldi\hbar\,\left(\lshad\Delta F_1,F_2\rshad-\lshad\Delta F_2,F_1\rshad\right)+
\left(\lshad\lshad S^{\hbar},F_1\rshad,F_2\rshad-\lshad\lshad S^{\hbar},F_2\rshad,F_1\rshad\right).
\end{multline*}
Because $F_1$ has odd parity, we swap the factors in $-\lshad\Delta F_2,F_1\rshad=\lshad F_1,\Delta F_2\rshad$;
likewise, $+\lshad F_1,\lshad S^{\hbar},F_2\rshad\rshad$ is the last term in the above expression. 
From our main Theorem~\ref{ThLaplaceOnSchouten}
and by Jacobi identity~\eqref{EqJacobiSchouten} we conclude that the commutator is equal to
$$
-\boldi\hbar\,\Delta\bigl(\lshad F_1,F_2\rshad\bigr)+\lshad S^{\hbar},\lshad F_1,F_2\rshad\rshad
=\Omega^{\hbar}\bigl(\lshad F_1,F_2\rshad\bigr),
$$
that is, the Schouten bracket of $F_1$ and $F_2$ is the new gauge symmetry generator.
\end{proof}

\begin{rem}(cf.\ \cite[\S5]{VTSh}).
The transformation $\exp\left(\tfrac{\boldi}{\hbar}S^{\hbar}\right)\mapsto\exp\left(\tfrac{\boldi}{\hbar}S^{\hbar}(\veps)\right)$
for a finite $\veps\in\BBR$ is determined by the operator $\exp(\veps[\Delta,F])$, where $[\ ,\ ]$ is the
anticommutator of two odd-parity objects. Indeed, by Theorem~\ref{ThLaplaceOnProduct} we have that
\begin{align*}
\Delta\bigl(&F\cdot\exp\left(\tfrac{\boldi}{\hbar}S^{\hbar}\right)\bigr)+
F\cdot\Delta\bigl(\exp\left(\tfrac{\boldi}{\hbar}S^{\hbar}\right)\bigr)\\
{}&=\Delta F\cdot\exp\left(\tfrac{\boldi}{\hbar}S^{\hbar}\right)-
\lshad F,\exp\left(\tfrac{\boldi}{\hbar}S^{\hbar}\right)\rshad-
F\cdot\Delta\left(\exp\left(\tfrac{\boldi}{\hbar}S^{\hbar}\right)\right)+
F\cdot\Delta\left(\exp\left(\tfrac{\boldi}{\hbar}S^{\hbar}\right)\right)\\
{}&=\tfrac{\boldi}{\hbar}(-\boldi\hbar\,\Delta F+\lshad S^{\hbar},F\rshad)\cdot\exp\left(\tfrac{\boldi}{\hbar}S^{\hbar}\right)=
\tfrac{\boldi}{\hbar}\dot S^{\hbar}\cdot\exp\left(\tfrac{\boldi}{\hbar}S^{\hbar}\right)=
\left.\frac{\Id}{\Id\veps}\right|_{\veps=0}\left(\exp\left(\tfrac{\boldi}{\hbar}S^{\hbar}\right)\right).
\end{align*}
Note that the Schouten bracket acts on $\exp\left(\tfrac{\boldi}{\hbar}S^{\hbar}\right)$ by the Leibniz rule 
(see Theorem~\ref{ThSchoutenOnProduct})
and we then use the equality $-\lshad F,\tfrac{\boldi}{\hbar}S^{\hbar}\rshad=\tfrac{\boldi}{\hbar}\lshad S^{\hbar},F\rshad$ 
which holds by Theorem~\ref{ThSchoutenOnProduct} again.
\end{rem}

Let us now regard 
the full quantum BV-\/action as the generating functional for ghost parity\/-\/even observables~$\cO$, 
see~\cite{Zinn-Justin:CriticalPhenomena}.

\begin{lemma}
There are no observables $\cO$, other than the identically zero functional, which would be ghost parity-odd.
\end{lemma}

\begin{proof}
Indeed, Eq.~\eqref{EqSchwDy} implies that the path integral 
$$
I=\int_{\Gamma(\zeta^0)}[D\bq]\,\cO([\bq],[\bq^{\dagger}])\exp\left(\tfrac{\boldi}{\hbar}S^{\hbar}([\bq],[\bq^{\dagger}])\right)
$$
over the space of ghost parity\/-\/even BV-\/section components 
is effectively independent of the ghost parity-odd BV-variables $\bq^{\dagger}$.
Notice further that the ghost parity $\GH(I)$ of this constant function $I([\bq^{\dagger}])$ is equal to that of $\cO$;
the quantum master-action $S^{\hbar}$ is parity-even. Under a (speculative) assumption that an observable $\cO$
could be ghost parity-odd, we obtain an odd parity constant. Unless a possibility of their existence is postulated
by brute force, this odd-parity constant must be equal to zero, whence the
ghost parity\/-\/odd functional
$\cO\in\ov{H^n}(\pi_\BV)\subseteq\ov{\gN^n}(\pi_\BV,T\pi_\BV)$ itself is zero.
\end{proof} 

In what follows we accept for transparency that there is no grading in the initial geometry of physical fields, 
i.e., for sections of the bundle $\pi\colon E^{n+m}\to M^n$. Let us focus on the standard cohomological approach to 
quantum BV-models and to their gauge symmetries (cf.~\cite{KontsevichSoibelman}).

\begin{lemma}
Suppose that an infinitesimal shift $S^{\hbar}\mapsto S^{\hbar}+\lambda\cdot\cO+\ov{o}(\lambda)$ of the quantum BV-action
by using an even-parity functional $\cO\in\ov{H^n}(\pi_\BV)\subseteq\ov{\gN^n}(\pi_\BV,T\pi_\BV)$ does not destroy 
the 
quantum master\/-\/equation,
$$\left.\frac{\Id}{\Id\lambda}\right|_{\lambda=0}\Delta
\left(\exp\left(\tfrac{\boldi}{\hbar}(S^{\hbar}+\lambda\cdot\cO)\right)\right)=0.$$
Then the \textsl{observable} $\cO$ is $\Omega^{\hbar}$-closed: $-\boldi\hbar\,\Delta\cO+\lshad S^{\hbar},\cO\rshad=0.$
\end{lemma}

\begin{proof}
The proof literally repeats that of Proposition~\ref{PropGaugeQME}.
\end{proof}

For a given odd-parity functional $F\in\ov{H^n}(\pi_\BV)$, we organize the infinitesimal shift~\eqref{EqQMEGaugeSym}
of the master-functional $S^{\hbar}$ as follows:
\begin{align*}
\dot S^{\hbar}&=-\boldi\hbar\,\Delta(F)+\lshad S^{\hbar},F\rshad,\\
\dot\cO&=\lshad\cO,F\rshad.
\end{align*}
Note that, unless one has that $\Delta F=0$ incidentally, the transformation of the integral \emph{functional}
$S^{\hbar}$ is not induced 
by any infinitesimal transformation of the BV-\emph{variables}, that is, by an evolutionary vector field on the horizontal infinite jet space at hand.
No earlier than the transformation law $S^{\hbar}\mapsto S^{\hbar}(\veps)$ 
is postulated, it becomes an act of will to think that the functional $F$ is the generator of parity-preserving
evolutionary vector field $\overleftarrow{Q}^{F}=\overrightarrow{Q}^F$ acting on the BV-variables so that 
$\dot\cO\cong\overrightarrow{Q}^F(\cO)$ for all observables $\cO$.

Furthermore, 
let us extend the deformation $\cO\mapsto\cO(\veps)$ of even-parity cocycles $\cO\in\ker\Omega^{\hbar}$ to 
the space of odd-parity functionals $\xi\in\ov{H^n}(\pi_\BV)\subseteq\ov{\gN^n}(\pi_\BV,T\pi_\BV)$ which produce the coboundaries
$\Omega^{\hbar}(\xi)$. Namely, we postulate that
$$\dot\xi=\lshad\xi,F\rshad$$
for all such functionals $\xi$; here we denote by the dot over $\xi$ its velocity in
the course of the transformation generated by a given $F$. Let us remember however that the law for evolution of the
odd-parity functionals $\xi$ which produce the $\Omega^{\hbar}$-coboundaries is different from our earlier postulate
(see Proposition~\ref{PropGaugeQME}) that the odd-parity generators $F_i$ of gauge symmetries do not evolve:
$dF_i/d\veps_j\equiv0$ or, in shorthand notation,
\begin{equation}\label{EqFStaysMotionless}
\dot F\equiv0.
\end{equation}
We claim that under these hypotheses, the structure of quantum BV-cohomology group remains intact in the course of gauge
symmetry transformations of the quantum master\/-\/action, $S^{\hbar}\mapsto S^{\hbar}(\veps)$,
even though the quantum BV-differential is modified, $\Omega^{\hbar}\mapsto\Omega^{\hbar}(\veps)$,
and the cocycles and coboundaries are also deformed.

\begin{theor}
An infinitesimal shift
of the quantum BV-cohomology classes induced by \eqref{EqQMEGaugeSym}, 
\eqref{EqFStaysMotionless}, and
\begin{align*}
\dot\cO&=\lshad\cO,F\rshad,& \cO&\in\ker\Omega^{\hbar},\\
\dot\xi&=\lshad\xi,F\rshad,& \xi&\in\ov{H^n}(\pi_\BV)\subseteq\ov{\gN^n}(\pi_\BV,T\pi_\BV),\ \xi \text{ odd},
\end{align*}
yields 
an isomorphism 
of the $\Omega^{\hbar}$-cohomology group\textup{:} under such mapping, every $\Omega^{\hbar}$-closed, even-parity 
$\Omega^{\hbar}$-cocycle 
$\cO$ becomes $\Omega^{\hbar}(\veps)$-closed, whereas the transformation of an even-parity coboundary
$\Omega^{\hbar}(\xi)$ produces an $\Omega^{\hbar}(\veps)$-coboundary\textup{:}
$(\Omega^{\hbar}(\xi))(\veps)=\Omega^{\hbar}(\veps)\bigl(\xi(\veps)\bigr)$.
\end{theor}

\begin{proof}
Let $\cO\in\ker\Omega^{\hbar}$ be an even-parity observable and $F$ an odd-parity generator of gauge transformation.
Consider the equation
$\Omega^{\hbar}(\veps)(\cO(\veps))=0$
which states that the transformed functional $\cO(\veps)$ remains a coboundary. The term which is proportional to $\veps$
in this equation's left-hand side is equal to
\begin{equation*}
\left.\frac{\Id}{\Id\veps}\right|_{\veps=0}\left(-\boldi\hbar\,\Delta\cO(\veps)+
\lshad S^{\hbar}(\veps),\cO(\veps)\rshad\right)
=\Omega^{\hbar}(\dot\cO)+\lshad\dot S^{\hbar},\cO\rshad=
\Omega^{\hbar}(\lshad\cO,F\rshad)+\lshad\Omega^{\hbar}(F),\cO\rshad;
\end{equation*}
recalling once again that $\Omega^{\hbar}=-\boldi\hbar\,\Delta+\lshad S^{\hbar},\,\cdot\,\rshad$,
we continue the equality
$$
=-\boldi\hbar\,\Delta(\lshad\cO,F\rshad)+\lshad S^{\hbar},\lshad\cO,F\rshad\rshad+\lshad-\boldi\hbar\,\Delta F+
\lshad S^{\hbar},F\rshad,\cO\rshad.\mbox{\hbox to 42mm {{ }\hfil { }}}
$$
Now by Theorem~\ref{ThLaplaceOnSchouten} 
we obtain that, the observable $\cO$ being parity-even,
\begin{multline*}
=-\boldi\hbar\,\lshad\Delta\cO,F\rshad+\boldi\hbar\,\lshad\cO,\Delta F\rshad+\lshad S^{\hbar},\lshad\cO,F\rshad\rshad-
\boldi\hbar\,\lshad\Delta F,\cO\rshad+\lshad\cO,\lshad S^{\hbar},F\rshad\rshad=\\=
\lshad\Omega^{\hbar}(\cO),F\rshad
\cong -\vec{\bQ}^F\bigl(\Omega^\hbar(\cO)\bigr) =0,
\end{multline*}
because 
$\lshad S^{\hbar},\lshad\cO,F\rshad\rshad=\lshad\lshad S^{\hbar},\cO\rshad,F\rshad-\lshad\cO,\lshad S^{\hbar},F\rshad\rshad$
by Jacobi identity~\eqref{EqJacobiSchouten}, because we are inspecting the $\varepsilon$-\/linear term in the 
operator $\Omega^\hbar(\veps)\circ\bigl(\veps=0\longmapsto\veps\neq0\bigr)$ applied to~$\cO$, and $\cO$~is an $\Omega^{\hbar}$-\/cocycle. 
Therefore, the zero initial condition
$\Omega^{\hbar}(\cO)=0$ evolves at zero velocity to the $\Omega^{\hbar}(\veps)$-cocycle equation 
$\Omega^{\hbar}(\veps)\bigl(\cO(\veps)\bigr)=0$ upon $\cO(\veps)$.

Likewise, let $\Omega^{\hbar}(\xi)$ be a coboundary for some odd-parity functional $\xi$ which evolves by 
$\dot\xi=\lshad\xi,F\rshad$. Then the even-parity observable $\Omega^{\hbar}(\xi)\in\ker\Omega^{\hbar}$ evolves as fast as 
$\lshad\Omega^{\hbar}(\xi),F\rshad$ but simultaneously we have that the mapping $\Omega^{\hbar}$ and its argument $\xi$
change. We claim that the two evolutions match so that $(\Omega^{\hbar}(\xi))(\veps)$ is $\Omega^{\hbar}(\veps)$-exact.
Indeed, we have that
\begin{multline*}
\left.\frac{\Id}{\Id\veps}\right|_{\veps=0}\bigl(\Omega^{\hbar}(\veps)(\xi(\veps))\bigr)=
\Omega^{\hbar}\bigl(\lshad\xi,F\rshad\bigr)+\lshad\Omega^{\hbar}(F),\xi\rshad\\{}
=-\boldi\hbar\,\lshad\Delta\xi,F\rshad\underline{{}-\boldi\hbar\,\lshad\xi,\Delta F\rshad}+\lshad S^{\hbar},\lshad\xi,F\rshad\rshad+\lshad
\underline{-\boldi\hbar\,\Delta F}+\lshad S^{\hbar},F\rshad,\underline{\xi}\rshad;
\end{multline*}
by cancelling out the underlined Schouten brackets and then using the Jacobi identity 
we obtain
\begin{equation*}
=\lshad-\boldi\hbar\,\Delta\xi,F\rshad
+\lshad \lshad S^{\hbar},\xi\rshad, F\rshad
+\lshad\xi,\lshad S^{\hbar},F\rshad\rshad-
\lshad\xi,\lshad S^{\hbar},F\rshad\rshad=
\lshad\Omega^{\hbar}(\xi),F\rshad,
\end{equation*}
which proves our claim.

Summarizing, we see that gauge symmetries of the quantum master-equation induce automorphisms of the 
$\Omega^{\hbar}$-cohomology group.
\end{proof}

We conclude that it would be conceptually incorrect to say that the infinitesimal gauge transformations of all
functionals in a quantum BV-model are induced by a canonical transformation, determined by the evolutionary vector field
$\overrightarrow{Q}^F$ acting on the BV-variables. Let us remember that the even-parity quantum master-action 
$S^{\hbar}\in\ov{H^n}(\pi_\BV)$ and its descendants, the observables $\cO$ evolve by
\begin{align*}
\dot S^{\hbar}&=-\boldi\hbar\,\Delta F+\lshad S^{\hbar},F\rshad=\Omega^{\hbar}(F),\qquad
F\in\ov{H^n}(\pi_\BV)\subseteq\ov{\gN^n}(\pi_\BV,T\pi_\BV),\quad F\text{ odd},\\
\intertext{and}
\dot\cO&=\lshad\cO,F\rshad.\\
\intertext{%
We note that the evolution of the generating functional $S^{\hbar}_{\BV}$ is \textbf{not} determined by a vector field
on the space of BV-variables. Likewise, we recall that the odd-parity arguments $\xi$ of $\Omega^{\hbar}$ for the
coboundaries $\Omega^{\hbar}(\xi)\sim0$ do evolve,}
\dot\xi&=\lshad\xi,F\rshad,\\
\intertext{%
whereas the generators $F$ of gauge symmetries for~\eqref{QME} never change: symbolically,}
\dot F&=0
\end{align*}
(see Eq.~\eqref{EqFStaysMotionless} above). In fact, one may think that each $F$ determines 
a parity-preserving evolutionary vector field $\overrightarrow{Q}^F$ on the space of BV-variables, but it is not the objects
$\overrightarrow{Q}^F$ but the full systems of four distinct evolution equations which encode the deformation of respective
functionals. Neither the functionals' attribution to the space of building blocks
$\ov{H^n}(\pi_\BV)\ni S^{\hbar}$,\ $\cO$,\ $F$  nor a functional's parity, $\GH(S^{\hbar})=\GH(\cO)$ and $\GH(F)=\GH(\xi)$,
completely determines their individual 
transformation laws.

\begin{rem}
The supports of test shifts $\delta\bolds$ can be arbitrarily small%
\footnote{We recall that the smoothness class of variations $\delta\bolds$ is determined by smoothness of the frame fields
$\vec{e}_i(\bx),\ \vec{e}^{{}\,\dagger i}(\bx)$ and coefficient functions $\delta s^i(\bx),\ \delta s_i^{\dagger}(\bx)$.
}
and they can be chosen in such a way that all boundary terms vanish in the course of integration by parts within equivalence
classes from the horizontal cohomology groups $\ov{H}^{n(1+k)}(\pi_{\BV}\times T\pi_{\BV}\times\ldots\times T\pi_{\BV})$.
Let us note also that these integrations by parts (see section~\ref{SecByParts}) transport the derivatives from one copy
of the base manifold $M^n$ to another copy; this reasoning stays local with respect to base points $\bx$ and local volume
elements $\dvol(\bx)$ because the geometric mechanism of locality yields the diagonal in powers of the base manifold.
However, an integration by parts in functionals from $\ov{H}^n(\pi_{\BV})$ is a different issue. In fact, it refers to the
topology of $M^n$ or to a choice of the class $\Gamma(\pi_{\BV})$ of admissible sections (so that there appear no boundary
terms as well). Let us recall that the only place where such global, de~Rham cohomology aspect is explicitly used is the
proof of 
Jacobi's identity for the variational Schouten bracket (see~\cite{Lorentz12}).
In turn, Theorems~\ref{ThBVDifferential} and~\ref{PropQBVDifferential} 
relate these properties of the bracket $\lshad\,,\,\rshad$ to cohomology generators $\Delta^2=0$ and $(\Omega^{\hbar})^2=0$.
(The converse is also true: Jacobi's identity for $\lshad\,,\,\rshad$ stems from $\Delta^2=0$.) This motivates why the
de~Rham and quantum BV-cohomologies are interrelated (cf.~\cite{BarnichBrandtHenneaux1995}).
\end{rem}

\section*{Conclusion}
\noindent%
Mathematical models are designed for description of phenomena of Nature\,; a construction of the models' objects is not the
same as their evaluation at given configurations of the models, which would associate $\Bbbk$-numbers to physical fields
$\phi\in\Gamma(\pi)$ in terms of such objects. Namely, consider an Euler\/--\/Lagrange model whose primary element is the
action functional $S\colon\Gamma(\pi)\to\Bbbk$. By definition, derivative objects are obtained from $S$ by using natural
operations such as $\smash{\vec{\delta}}$ or
$\lshad\,,\,\rshad$ and~$\Delta$. The derivative objects' geometric complexity is greater than that of~$S$
because they absorb the domains of definition for test shifts $\delta s_1,\,\ldots\,,\delta s_k$ of field configurations.
We emphasize that such composite structure objects do not yet become maps $\Gamma(\pi)\to\Bbbk$ which would suit well for their
evaluation at sections $\bolds\in\Gamma(\pi)$ yielding $\Bbbk$-numbers. The intermediate objects can rather be used as
arguments of $\lshad\,,\,\rshad$ or $\Delta$ in a construction of larger, logically and geometrically more complex objects\,;
we illustrate by Fig.~\ref{FigExpandShrink} 
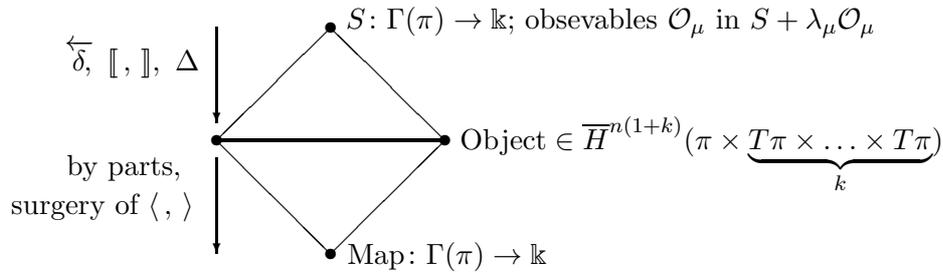
\begin{figure}[htb]
\begin{center}
\unitlength=1mm
\linethickness{0.4pt}
\begin{picture}(102.00,35.67)
{\linethickness{1.0pt}
\put(20.00,20.00){\line(1,0){30.00}}
}
\put(50.00,20.00){\line(-1,1){15.00}}
\put(35.00,35.00){\line(-1,-1){15.00}}
\put(20.00,20.00){\line(1,-1){15.00}}
\put(35.00,5.00){\line(1,1){15.00}}
\put(50.00,20.00){\circle*{1.33}}
\put(20.00,20.00){\circle*{1.33}}
\put(35.00,35.00){\circle*{1.33}}
\put(35.00,5.00){\circle*{1.33}}
\put(37.00,33.67){\makebox(0,0)[lb]{$S\colon\Gamma(\pi)\to\Bbbk$; obsevables $\cO_{\mu}$ in $S+\lambda_{\mu}\cO_{\mu}$}}
\put(52.00,13.00){\makebox(0,0)[lb]{$\text{Object}\in\ov{H}^{n(1+k)}(\pi\times\underbrace{T\pi\times\ldots\times T\pi}_k)$}}
\put(37.17,2.67){\makebox(0,0)[lb]{$\text{Map}\colon\Gamma(\pi)\to\Bbbk$}}
\put(20.00,17.67){\vector(0,-1){12.67}}
\put(20,35.00){\vector(0,-1){12.67}}
\put(0,28.33){\makebox(0,0)[lb]{$\overleftarrow{\delta}\!\!,\ \lshad\,,\,\rshad,\ \Delta$}}
\put(0.33,14.67){\makebox(0,0)[lb]{by parts,}}
\put(-7,9.33){\makebox(0,0)[lb]{surgery of $\langle\,,\,\rangle$}}
\end{picture}
\end{center}
\caption{The action~$S$ as a generator of observables, building blocks of derivative objects as horizontal cohomology classes in products of bundles over $M\times M\times\ldots\times M$, and resulting mappings as the objects' contractions over Whitney's sum of bundles.}\label{FigExpandShrink}
\end{figure}
the expansion of analytic structures and their shrinking in the course of 
integration by parts and multiplication of normalized test shifts in reconfigured couplings. Indeed, the derivative objects
become multi\/-\/linear maps with respect to $k$-\/tuples of the variations $\delta s_1$,\ $\ldots$,\ $\delta s_k\in\Gamma(T\pi)$
only when the integrations by parts carry all derivatives away from the test shifts, channelling the derivations to
densities of the object's constituent blocks such as the Lagrangian in the action functional. A surgery of couplings then
contracts the values of normalized test shifts by virtue of~\eqref{EqNormalize} at every point of the base manifold. This is
how maps $\Gamma(\pi)\to\Bbbk$ are obtained.

We conclude that 
a calculation of composite\/-\/structure object 
may not be interrupted ahead of time. Otherwise speaking, every
calculation stretches from its input data to the end value at $\bolds$\,; independently existing values at $\bolds$ for the
resulting object's constituent elements not always contribute to the sought-for value of the large structure (e.\,g.,
consider~\eqref{EqZimes} on p.~\pageref{EqZimes} and Example~\ref{Countercounterexample} on p.~\pageref{Countercounterexample}
and try to calculate consecutively
the objects $\Delta F$, $\Delta G$, and their Schouten brackets with $G$ and $F$, respectively, for that example's
functionals $F$ and $G$). Summarizing, it is illegal to construct composite
objects step by step, redundantly inspecting the elements' values at field con\-fi\-gu\-ra\-ti\-ons. One must not deviate from a way towards the appointed end of 
logical reasoning.

In fact, it is us but not Nature who calculates (e.\,g., the left\/-\/hand sides of equations of motion): 
Nature neither calculates nor evaluates\,; for there is no
built\/-\/in mechanism for doing that.
\footnote{The probabilistic
approach to evolution of Nature suggests that maxima of transition (and correlation) functions concentrate near the zero
loci of such deterministic equations' left-hand sides. At the same time, Noether symmetries of the action~$S$ are abundant
in the models. Not referring to any actual transformation of a system's components, such symmetries reflect the model's
geometry. The analytic machinery of self-regularizing structures yields the invariants --\,e.\,g., cohomology classes as 
in section~\ref{SecGauge}
\,-- which constrain the probabilistic laws of evolution.}
This implies that there is no ever\/-\/growing logical complexity in a description of the Universe\,; 
the flow of local, observer\/-\/dependent time does not require any perpetual increase of the number $k\geqslant0$ of factors in the product\/-\/bundle location of objects over $k+1$ copies of the space\/-\/time. Conversely, there always remains a unique copy of the space\/-\/time for all local functionals.

The space\/-\/time geometry of information transfer is very restrictive: its pointwise locality of events of couplings between dual objects
is an absolute principle\,; by weakening this hypothesis one could create a source of difficulties through
causality violation. 
Consequently, a count of space\/-\/time points where the couplings with a given (co)\/vector occur makes the
formalism of singular linear integral operators truly adequate in mathematical models of physical phenomena.\footnote{We recall from Remark~\ref{RemRolesInt} on p.~\pageref{RemRolesInt} that 
the volume elements $\dvol\bigl(\bx,\phi(\bx)\bigr)=
\sqrt{|\det(g_{\mu\nu})|}\,\Id\bx$ are present in the building blocks of 
composite\/-\/struc\-tu\-re objects.
Let us note further that an association of the weight factors
$\dvol(\bx)$ with point $\bx\in M^n$ is intrinsically related to the structure of space\/-\/time~$M^n$ as topological manifold (cf.~\cite{Protaras2012}).
It is readily seen that a discrete tiling of space\/-\/time converts the integrations over a measure on it to weighted
sums over the points which mark the quantum domains. This links the concept with loop quantum gravity 
(see e.\,g.~\cite{GambiniPullinQG,RovelliQG,ThiemannQG}).
}

We finally remark that the product-base approach of bundles $\pi\times T\pi\times\ldots\times T\pi$ over
$M\times M\times\ldots\times M$ to the geometry of variations highlights the concept of physical field as 
infinite-dimensional system with degrees of freedom which are attached at every point of space-time. The locality principle
for (co)vector interaction is the mechanism which distinguishes between space-time points with respect to its 
(non)Hausdorff topology.

\subsection*{Discussion}
Let us finally address two logical aspects of the geometry of variations.

\subsubsection*{Linear vector space structures}
Nature is essentially nonlinear\,; for there is no mechanism which would realize --\,under a uniform time bound\,--
an arbitrarily large number of replications 
of an object. 
This is tautological 
for those physical fields~$\phi$ which
take values in spaces without any linear structure. Moreover, even if there is a brute force labelling of 
Euler\/--\/Lagrange equations by using the fields $\phi$, a linear vector space pattern of the equations of motion is not
utilized (the same is true for the equations' descendants such as the antifields $\phi^{\dagger}$ or (anti)ghosts). Indeed,
it is only their the \emph{tangent} spaces whose linear structure is used, in particular, in order to split the variations in ghost 
parity-homogeneous components. Objects are linearized only in the course of variations under infinitesimal test shifts.
For example, this determines the distinction between finite offsets $\Delta\bx$ so that
$(\bx,\bx+\Delta\bx)\in M\times M$ and infinitesimal test shifts $\left.\mathstrut\delta\bx\right|_{\bx}\in T_{\bx}M$
which are mapped to the number field $\Bbbk$ by covectors $\left.\mathstrut\Id\bx\right|_{\bx}\in T^*_{\bx}M$.

\subsubsection*{Annual reproduction rate for interspecimen breeding of cats and whales}
An immediate comment on the title of this paragraph is as follows. One could proclaim that the annual reproduction rate for
interspecimen breeding of --\,without loss of generality\,-- cats and whales is equal to zero for a given year. 
Alternatively, one should understand that such events never happen (not that a given year brought no brood).

This grotesque illustration works equally well for the (co)tangent spaces to fibres of the BV-zoo or, in broad terms, for
a definition of Kronecker's symbol $\boldsymbol{\delta}_i^j$ by zero whenever the indices $i\ne j$ do not match so that the
couplings in~\eqref{EqChoiceSign} do not eventuate. We argue that, on top of the absolute pointwise locality for 
couplings~\eqref{EqTwoCouplings}, a superficial definition of $\langle\,,\,\rangle$ by zero for mismatching elements $\vec{e}_i$ and 
$\vec{e}^{{}\,\dagger j}$ of dual bases is a mere act of will\,; in reality those evaluations do not occur. Consequently,
the geometry dictates that
\[
\log\left\langle\vec{e}_i(\bx),{}^{\dagger}(\vec{e}_j)(\bx)\right\rangle=\log1=0
\quad\text{and}\quad
\log\left\langle\vec{e}^{{}\,\dagger j}(\bx),{}^{\dagger}(\vec{e}^{{}\,\dagger i})(\bx)\right\rangle=\log1=0.
\]
Combined with the geometric locality principle~\eqref{EqLocality} 
realized by singular linear integral operators~\eqref{EqDefVariations}, 
this argument finally resolves the paradoxical, \textsl{ad hoc} conventions ${\boldsymbol{\delta}(0)=0}$ and 
${\log\boldsymbol{\delta}(0)=0}$ for Dirac's distribution.

\ack
The author thanks the Organizing committee of XXI International conference `Integrable systems \& quantum symmetries'
(June 11\,--\,16, 2013; CVUT Prague, Czech Republic) for 
cooperation and warm atmosphere during the meeting. 
These notes follow the lecture course which was read by the author in October 2013 at
the Taras Shevchenko National University and Bogolyubov Institute for Theoretical
Physics in Kiev, Ukraine; the author is grateful to BITP for hospitality.
The author thanks M.~A.~Vasiliev and A.~G.~Nikitin for helpful discussions and constructive
criticisms.

This research was supported in part 
by 
JBI~RUG project~103511 (Groningen). 
A~part of this research was done while the author was visiting at 
the $\smash{\text{IH\'ES}}$ (Bures\/-\/sur\/-\/Yvette); 
the financial support and hospitality of this institution are gratefully acknowledged.

\setcounter{section}{1}
\appendix
\section{Proof of Propositions~\protect{\ref{thm:QME}}
and~\protect{\ref{PropPrepareGaugeQME}}}\label{AppProveQME}
\noindent We need the following two lemmas.

\begin{lemma}\label{thm:SchoutenPower}
Let $F \in \overline{H}^n(\pi_\BV)$ be an even integral functional, let $G \in \overline{\mathfrak{N}}^n(\pi_\BV,T\pi_\BV)$ be another functional, and let $n \in \mathbb{N}_{\geq1}$. Then
\[
\schouten{G,F^n} = n\schouten{G,F}F^{n-1}.
\]
\end{lemma}

\begin{proof}
We use induction on Theorem~\ref{ThSchoutenOnProduct}. Note that all signs vanish since $F$ is even, 
meaning that whenever $F$ is multiplied with any other integral functional, the factors may be freely swapped without 
this resulting in minus signs. For $n=1$ the statement is trivial. Suppose the formula holds for 
some $n \in \mathbb{N}_{>1}$, then~$\schouten{G,F^{n+1}}={}$
\[
\schouten{G,F\cdot F^n}
       = \schouten{G,F}\cdot F^n + F\cdot\schouten{G,F^n}
       = \schouten{G,F}\cdot F^n + nF\cdot\schouten{G,F}F^{n-1}
       = (n+1)\schouten{G,F}\cdot F^n,
\]
so that the statement also holds for $n+1$.
\end{proof}

\begin{lemma}
Let $F \in \overline{H}^n(\pi_\BV)$ be an even integral functional, and let $n \in \mathbb{N}_{\geq2}$. Then
\[
\Delta(F^n) = n(\Delta F)\cdot F^{n-1} + \tfrac{1}{2}n(n-1)\schouten{F,F}\cdot F^{n-2}.
\]
\end{lemma}

\begin{proof}
We use induction and the previous lemma. For $n=2$ the formula clearly holds by Theorem~\ref{ThLaplaceOnProduct}. 
Suppose that it holds for some $n \in \mathbb{N}_{>2}$, then
\begin{align*}
      \Delta(F^{n+1}) &= \Delta(F\cdot F^n) 
       = (\Delta F)\cdot F^n + \schouten{F,F^n} + F\cdot\Delta(F^n) \\
      &= (\Delta F)\cdot F^n + n\schouten{F,F}\cdot F^{n-1} + F\cdot n(\Delta F)F^{n-1} + \tfrac{1}{2}n(n-1)F\cdot\schouten{F,F}F^{n-2} \\
      &= (n+1)(\Delta F)\cdot F^n + \tfrac{1}{2}(n+1)n\,\schouten{F,F}\cdot F^{n-1},
\end{align*}
so that the statement also holds for $n+1$.
\end{proof}

\begin{proof}[Proof of Proposition~\textup{\ref{thm:QME}}]
For convenience, we denote $F = \frac{\boldi}{\hbar}S^\hbar$. Then
\begin{align*}
      0 &= \Delta(\exp F) 
         = \Delta\left(\sum_{n=0}^\infty \frac1{n!}F^n\right)
         = \sum_{n=0}^\infty \frac1{n!}\Delta(F^n) \\
        &= \sum_{n=1}^\infty \frac{n}{n!}(\Delta F)\cdot F^{n-1} + \sum_{n=2}^\infty\frac{1}{2n!}n(n-1)\schouten{F,F}\cdot F^{n-2} \\
        &= (\Delta F)\cdot\sum_{n=1}^\infty \frac{1}{(n-1)!}F^{n-1} + \frac12\schouten{F,F}\cdot\sum_{n=2}^\infty\frac{1}{(n-2)!}F^{n-2} \\
        &= \left(\Delta F + \tfrac12\schouten{F,F}\right)\cdot\exp F 
         = \left(\frac{\boldi}{\hbar}\Delta S^\hbar - \frac1{2\hbar^2}\schouten{S^\hbar,S^\hbar}\right)\cdot\exp\left(\tfrac{\boldi}{\hbar}S^\hbar\right),
\end{align*}
from which the result follows.
\end{proof}

\begin{proof}[Proof of Proposition~\textup{\ref{PropPrepareGaugeQME}}
\textup{(}cf.\ Proposition~\textup{\ref{PropGaugeQME}} on p.~\textup{\pageref{PropGaugeQME})}]
Again, 
let us set $F = \frac{\boldi}{\hbar}S^\hbar$. 
We first calculate, using Lemma~\ref{thm:SchoutenPower}, 
   \[
      \schouten{\mathcal{O},\exp F}
      = \sum_{n=0}^\infty\frac1{n!}\schouten{\mathcal{O},F^n}
      = \sum_{n=1}^\infty\frac{n}{n!}\schouten{\mathcal{O},F}F^{n-1}
      = \schouten{\mathcal{O},F}\exp F.
   \]
Then
   \begin{align*}
      0 &= \Delta(\mathcal{O}\exp F) 
         = (\Delta \mathcal{O})\exp F + \schouten{\mathcal{O}, \exp F} + \mathcal{O}\cdot\Delta(\exp F) \\
        &= \big(\Delta \mathcal{O} + \schouten{\mathcal{O}, F}\big)\exp F 
         = \left(\Delta \mathcal{O} + \tfrac{\boldi}{\hbar}\schouten{\mathcal{O}, 
S^\hbar}\right)\exp\left(\tfrac{\boldi}{\hbar}S^\hbar\right),
   \end{align*}
   from which the assertion follows.
\end{proof}

\end{document}